\documentclass[12pt,a4paper]{article}
\usepackage{custom_tex}
\usepackage[left=0.75in,right=0.75in]{geometry}

\newcommand{\ourtitle}{{Efficient and Multiply Robust Risk Estimation\\under General Forms of Dataset Shift}}

\title{\ourtitle}

\date{}
\author[1]{Hongxiang Qiu}
\author[2]{Eric Tchetgen Tchetgen}
\author[2]{Edgar Dobriban\footnote{E-mail:
\texttt{qiuhongx@msu.edu},
\texttt{dobriban@wharton.upenn.edu},
\texttt{ett@wharton.upenn.edu}
}}

\affil[1]{Department of Epidemiology and Biostatistics, Michigan State University}
\affil[2]{Department of Statistics and Data Science, The Wharton School, University of Pennsylvania}

\begin{document}

\maketitle

\begin{abstract}
    Statistical machine learning methods often face the challenge of limited data available from the population of interest. One remedy is to leverage data from auxiliary source populations, which share some conditional distributions or are linked in other ways
    with the target domain.
    Techniques leveraging such \emph{dataset shift} conditions
    are known as \emph{domain adaptation} or \emph{transfer learning}.
    Despite extensive literature on dataset shift, limited works address how to efficiently use the auxiliary populations to improve the accuracy of risk evaluation for a given machine learning task in the target population.
    
    In this paper, we study the general problem of efficiently estimating target population risk under various dataset shift conditions, leveraging semiparametric efficiency theory.
    We consider a general class of dataset shift conditions, which includes three popular conditions---covariate, label and concept shift---as special cases.
    We allow for partially non-overlapping support between the source and target populations.
    We develop efficient and multiply robust estimators along with a straightforward specification test of these dataset shift conditions.
    We also derive efficiency bounds for two other dataset shift conditions, posterior drift and location-scale shift.
    Simulation studies support the efficiency gains due to leveraging plausible dataset shift conditions.
\end{abstract}

\tableofcontents

\section{Introduction} \label{sec: intro}

\subsection{Background}

A common challenge in statistical machine learning 
approaches to prediction
is that limited data is available from the population of interest, 
despite potentially large amounts of data from similar populations.
For instance, it may be of interest to predict HIV treatment response in a community based on only a few observations.
A large dataset from another community in a previous HIV study may help improve the training of such a prediction model.
Another example is building a classification or diagnosis model based on medical images for lung diseases \citep{Christodoulidis2017}. 
A key task therein is to classify the texture in the image, but the size of a labeled medical image sample is often limited due to the high cost of data acquisition and labeling. 
It may be helpful to leverage large existing public image datasets as supplemental data to train the classifier.

In these examples and others, 
it is desirable to use data from similar source populations to supplement target population data, under plausible \emph{dataset shift} conditions relating the source and target populations  \citep[see, e.g.,][]{Storkey2013,shimodaira2000improving,Sugiyama2012}.
Such methods are known as \emph{domain adaptation} or \emph{transfer learning} \citep[see, e.g.][]{kouw2018introduction,Pan2010}.

A great deal of work has been devoted to leveraging or addressing various dataset shift conditions. 
\citet{Polo2022} proposed testing for various forms of dataset shift. 
Among dataset shift types,
popular conditions include \emph{covariate shift}, 
where only the covariate distribution changes,
as well as
\emph{label shift}, where only the outcome distribution changes---also termed \emph{choice-based sampling} or \emph{endogenous stratified sampling} in \cite{manski1977estimation}, \emph{prior probability shift} in \cite{Storkey2013}, and \emph{target shift} in \cite{Scott2019,Zhang2013}.
Another popular condition is \emph{concept shift}, where the covariate or label distribution does not change---also termed \emph{conditional shift} in \cite{Zhang2013}. See, e.g., \cite{kouw2018introduction,moreno2012unifying,scholkopf2012causal}, for reviews of common dataset shift conditions.
These three conditions---covariate, label, and concept shift---are popular because they are interpretable, broadly applicable, and analytically tractable. 
There is extensive literature on machine learning under these dataset shift conditions \citep[e.g.,][among others]{sugiyama2007covariate,lipton2018detecting,Pathak2022,Ma2023}.

Other conditions and methods have also been studied for more specific problems; 
examples include (generalized, penalized) linear models \citep{bastani2021predicting,cai2022semi,chakrabortty2018efficient,gu2022robust,liu2020doubly,Liu2023,tian2022transfer,zhang2019semi,zhang2022class,zhou2022doubly}, binary classification \citep{cai2021transfer,Scott2019}, graphical models \citep{he2022transfer,li2022transfer}, and location-scale shifts \citep{Zhang2013}, among others.

For domain adaptation where a limited amount of fully observed data from the target population is available, 
there may be multiple valid methods to incorporate source population data.
It is thus important to understand which ones efficiently extract information from data in both source and target populations. 
However, the efficient use of source population data to supplement the target population data has only been recently studied. \cite{azriel2021semi,gronsbell2022efficient,Yuval2023,zhang2021double,zhang2022high} studied this problem for mean estimation and (generalized) linear models, 
under concept shift (i.e., for semi-supervised learning). 
\cite{Li2023} studied efficiency theory for data fusion with an emphasis on causal inference applications, a setting related to ours with a somewhat different primary objective. Other related works in data fusion include \citet{Angrist1992,Chatterjee2016,Chen2000,DOrazio2006,DOrazio2010,Evans2021,Rassler2012,Robins1995}, among others.
A study of domain adaptation with more general prediction techniques under general dataset shift conditions is lacking.

\subsection{Our contributions}

In this paper, we study the general problem of 
efficient model-agnostic risk estimation in a target population for data adaptation with fully observed data from both source and target populations under various dataset shift conditions.
We take the perspective of modern semiparametric efficiency theory \citep[see, e.g.,][]{Bolthausen2002,Pfanzagl1985,Pfanzagl1990,vanderVaart1998}, 
because many dataset shift conditions can be formulated as \emph{restrictions on the observed data generating mechanism}, 
yielding a semiparametric model.

We estimate the risk due to its broad applicability and central role in training predictive models and model selection \citep[e.g.,][etc]{Vapnik1992,Gyorfi2002}. 
Empirical risk minimization (ERM) is a fundamental approach in statistics and machine learning, 
where the goal is to minimize the empirical average of the loss---i.e., the risk---over a set of candidate models. 
Recent studies have highlighted the importance of accurate risk estimation for effective model selection
in settings where the risk cannot be estimated nonparametrically due to the presence of 
high-dimensional nuisance functions
\citep[e.g.,][etc]{vaart2006oracle, brookhart2006semiparametric, nie2021quasi, foster2023orthogonal}.  
Our risk estimators have the potential to be relevant and useful in such settings, by enabling improved model selection through improved risk estimation.
Likewise, the related goal of constructing prediction sets with guaranteed coverage \citep{Vovk2013,Qiu2022,Yang2022} often depends on precise estimates of the coverage error probability of constructed prediction sets \citep{Angelopoulos2021,Park2020,Park2022,Yang2022}. 

After presenting the general problem setup in Section~\ref{sec: setup}, we consider a general dataset shift condition, which we call \emph{sequential conditionals} (Condition~\ref{DScondition: independence} in Section~\ref{sec: independence}),
for scenarios where target population data is available while the source and target populations may have only partially overlapping support.
This condition includes covariate, label and concept shift as special cases.
We consider data where an observation $Z$ can be decomposed into components 
$(Z_1,\ldots,Z_K)$.
Under this condition, some of the conditional distributions $Z_k \mid (Z_1,\ldots,Z_{k-1})$ are shared between the target and source populations, for $k=1,\ldots, K$.
As our first main contribution, 
we propose a novel risk estimator that we formally show in Theorem~\ref{thm: independence eff and DR} to be semiparametrically efficient and multiply robust \citep{TchetgenTchetgen2009,Vansteelandt2007} under this dataset shift condition.

In particular, we propose to obtain flexible estimators $\hat{\theta}^k$ of certain nuisance conditional odds functions 
$\theta^k_*$ conditional on variables $(Z_1,\ldots,Z_k)$, and estimators $\hat{\ell}^k$ of conditional mean loss functions $\ell^k_*$ by \emph{sequential regression} of  $\hat{\ell}^{k+1}(Z_1,\ldots,Z_{k+1})$ on $(Z_1,\ldots,Z_k)$.
We show that our risk estimator is efficient given sufficient convergence rates of the product of errors of (i) $\hat{\theta}^k$ for $\theta^k_*$, and of (ii) $\hat{\ell}^k$ for $\ell^k_*$.
Moreover, when $\hat{\ell}^k_v$ converges to a certain limit function $\ell^k_\infty$, our estimator is $2^{K-1}$-robust.
Specifically, it is consistent if
for every $k=1,\ldots,K-1$, 
$\hat{\theta}^k$ is consistent for $\theta^k_*$ \emph{or} $\hat{\ell}^k$ is consistent for the \emph{oracle regression} $u^k$ of $\ell^k_\infty$; but not necessarily both.
The latter oracle regression is defined as the conditional expectation of $\ell^k_\infty$ on $(Z_1,\ldots,Z_k)$ under the true distribution.

Our choice of parametrization and the sequential construction of $\hat{\ell}^k$ are key to multiple robustness. 
Suppose instead that 
each true conditional mean loss function $\ell^k_*$ is instead parameterized \emph{directly}---rather than sequentially---as the regression of the loss on the variables $(Z_1,\ldots,Z_{k})$ in the target population, and accordingly construct $\hat{\ell}^k$ 
by direct regression in the target population.
Then the resulting estimating equation-based estimator using the efficient influence function is not guaranteed to be $2^{K-1}$-robust.

Based on this estimator, we further propose a straightforward specification test \citep{Hausman1978} of whether our efficient estimator converges to the risk of interest in probability, which can be used to test the assumed sequential conditionals condition.
In doing so, we theoretically analyze the behavior of our proposed estimator when the sequential conditionals condition fails. 
We 
analytically derive the bias 
due to the failure of sequential conditionals and show that, in this case, our estimator may diverge arbitrarily as sample size increases if the support of the source populations only partially overlaps with the target population. Under the sequential conditionals condition, such a scenario for the supports is allowed, 
but does not lead to this convergence issue. 
To obtain this result, we need a more careful analysis than the standard analysis of Z-estimators \citep[e.g., Section~3.3 in][]{vandervaart1996} because of the inconsistency of our estimator without the sequential conditionals condition.

Next, we investigate the efficient risk estimation problem in more detail for concept shift in the features and for covariate shift,  in Sections~\ref{sec: X con shift} and \ref{sec: cov shift}, respectively. 
We characterize when efficiency gains are large, develop simplified efficient and robust estimators, and study their empirical performance in simulation studies.
In particular, we show that our estimator is regular and asymptotically linear (RAL) \emph{even if the nuisance function is estimated inconsistently} under concept shift (Theorem~\ref{thm: X con shift efficiency robust}). 
We also show a new impossibility result about such full robustness for covariate or label shift under common parametrizations (Lemma~\ref{lemma: cov shift no eff and robust}).

We present additional new results in the Supplementary Material.
We present additional simulation results showcasing the efficiency gain from our proposed methods in model comparison and model training in Supplement~\ref{sec: more sim}.
We illustrate our proposed estimators in an HIV risk prediction example in Supplement~\ref{sec: data analysis}.
In Supplement~\ref{sec: other DS conditions}, we derive efficiency bounds for risk estimation under
three other widely-applicable dataset shift conditions, \emph{posterior drift} \citep{Scott2019}
\emph{location-scale shift} \citep{Zhang2013},
and \emph{invariant density ratio shape} motivated by \citet{Tasche2017}.
The proof of these results requires delicate derivations involving tangent spaces and their orthogonal complements, leading to intricate linear integral equations and, in some cases, a closed-form solution for the efficient influence functions.
In Supplement~\ref{sec: other DS conditions2}, we present additional results on other widely-applicable dataset shift conditions, including the invariant density ratio condition \citep{Tasche2017} and stronger versions of posterior drift \citep{Scott2019} and location-scale shift \citep{Zhang2013} conditions.
In Supplement~\ref{sec: pred set}, we describe how our proposed risk estimators can help construct prediction sets with marginal or training-set conditional validity.
Proofs of all theoretical results can be found in Supplement~\ref{sec: proof}.
We implement our proposed methods for covariate, label and concept shift in an R package available at \url{https://github.com/QIU-Hongxiang-David/RiskEstDShift}.

\section{Problem setup} \label{sec: setup}

Let $O$ be a prototypical data point consisting of the observed data $Z$ lying in a space $\mathcal{Z}$ and an integer indexing variable $A$ in a finite set $\mathcal{A}$ containing zero. 
The variable $A$ indicates whether the data point comes from the target population ($A=0$) or a source population ($A \in \mathcal{A} \setminus \{0\}$).\footnote{Throughout this paper, we emphasize certain aspects of the observed data-generating mechanism when employing the terms ``domain'' or ``population''. For example, when a random sample is drawn from a superpopulation and variables are measured differently for two subsamples, we may treat this data as two samples from two different populations because of the different observed data-generating mechanisms.}
The observed data $(O_1,\ldots,O_n)$ is an independent and identically distributed (i.i.d.) sample from an unknown distribution $P_*$. 
We will use a subscript $*$ to denote components of $P_*$ throughout this paper.
Data $Z$ is observed from both the source and target populations, i.e., for both $A=0$ and $A\neq 0$.

The estimand of interest is the risk, namely the average value of a given loss
function $\ell:\mZ\to \R$, in the target population:
\begin{equation}
    r_* := R(P_*) := \expect_{P_*}[\ell(Z) \mid A=0]. \label{eq: def risk}
\end{equation}
We often focus on the supervised setting where $Z=(X,Y)$, with $X \in \mathcal{X}$ being the covariate or feature and $Y \in \mathcal{Y}$ being the outcome or label.
In this case, our observed data are i.i.d.~triples $(X_i,Y_i,A_i) \in \mX \times \mY \times \ma$ distributed according to $P_*$.
Next, we provide two examples of loss functions $\ell$ below.

\begin{example}[Supervised learning/regression] \label{ex: regression}
    Let $f:\mX\to \mY'$ be a given predictor---obtained, for example,  from a separate training dataset---for some space $\mY'$ that can differ from $\mY$.
    It may be of interest to estimate a measure of the accuracy of $f$. 
    One common measure is the mean squared error, which is induced by the squared error loss $\ell(x,y)=(y-f(x))^2$. 
    For binary outcomes where $\mY = \{0,1\}$, it is also common to consider the risk induced by the cross-entropy loss, namely Bernoulli negative log-likelihood $\ell(x,y)=-y \log\{f(x)\} - (1-y) \log\{1-f(x)\}$, when $\mY'$ is the unit interval $(0,1)$ and $f$ outputs a predicted probability. Another common measure of risk is $P_*(Y \neq f(X) \mid A=0)$, measuring prediction inaccuracy.
    This is induced by the loss $\ell(x,y)=\ind(y \neq f(x))$ when $\mY'=\{0,1\}$ and $f$ outputs a predicted label.
\end{example}

\begin{example}[Prediction sets with coverage guarantees] \label{ex: prediction sets}
    It is often of interest to construct prediction sets with a coverage guarantee. 
    Two popular guarantees are marginal coverage and training-set conditional---or \emph{probably approximately correct} (PAC)---coverage \citep{Vovk2013,Park2020}. 
    To achieve such coverage guarantees, one may estimate---or obtain a confidence interval for---the coverage error of a given prediction set \citep{Vovk2013,Angelopoulos2021,Yang2022}. 
    Let $C: \mathcal{X} \to 2^{\mathcal{Y}}$ be a given prediction set. With the indicator loss $\ell(x,y)=\ind(y \notin C(x))$,  the associated risk is  the coverage error probability $P_*(Y \notin C(X) \mid A=0)$ of $C$ in the target population.
\end{example}

\begin{remark}[Broader interpretation of risk estimation problem] \label{rmk: risk interpretation}
     Our results in this paper apply to a broader range of problems beyond risk estimation, provided that the problem under consideration can be mapped to our setup.
     The loss function $\ell$ may be interpreted in a broad sense. We list a few examples below.
     Moreover,
     the data point $Z$ does not necessarily have to consist of a covariate vector $X$ and an outcome $Y$. If additional variables $W$ related to $Y$ are observed---for example, outcomes other than $Y$---these can be leveraged for risk estimation, even if the prediction model only uses the covariate $X$.
\end{remark}

\begin{example}[Mean estimation]
    If the estimand of interest is the mean $\expect_{P_*}[Z \mid A=0]$ for $Z \in \real$, we may take $\ell$ to be the identity function.
\end{example}

\begin{example}[Quantile estimation]
    To estimate a quantile of $Z \mid A=0$, we may consider $\ell$ ranging over the function class $\{z \mapsto \ind(z \leq t): t \in \real\}$.
\end{example}

\begin{example}[Model comparison]
    To compare the performance of two methods in the target population, we may take $\ell$ to be the difference between the loss of these two methods. For example, with $f^{(1)}$ and $f^{(2)}$ denoting two given predictors, we may take $\ell: (x,y) \mapsto (y-f^{(1)}(x))^2 - (y-f^{(2)}(x))^2$ to be the loss difference.
\end{example}

\begin{example}[Estimating equation]
    Suppose that the estimand is the solution $\beta_*$ to an estimating equation $\expect_{P_*}[\Psi_\beta(Z) \mid A=0] = 0$ in $\beta$, which includes linear regression, logistic regression, and parametric regression models as special cases. We may consider $\ell$ ranging over the function class $\{z \mapsto \Psi_\beta(z)\}$ indexed by $\beta$ and $\hat{r}_\beta$ an estimator of $\expect_{P_*}[\Psi_\beta(Z) \mid A=0]$. Let the estimator $\hat{\beta}$ of $\beta_*$ be the solution to $\hat{r}_\beta=0$ in $\beta$. Theorem~3.3.1 of \citet{vandervaart1996} implies that reducing the asymptotic variance (as $n \to \infty$) of $\hat{r}$ leads to a reduced asymptotic variance of $\hat{\beta}$.
\end{example}

Without additional conditions\footnote{We suppose that $\Var_{P_*}(\ell(X,Y) \mid A=0) < \infty$, a very mild condition, throughout this paper.} on
the true data distribution $P_*$---under a nonparametric model---the source populations are non-informative about the target population because they may differ arbitrarily.
In this case, a viable estimator of $r_*$ is \emph{the nonparametric estimator}, the sample mean over the target population data\footnote{In this paper, we define $0/0 = 0$.}:
\begin{equation}
    \hat{r}_{\nonparametric} := \frac{\sum_{i=1}^n \ind(A_i=0) \ell(Z_i)}{\sum_{i=1}^n \ind(A_i=0)}. \label{eq: np estimator}
\end{equation}
For any scalars $\rho \in (0,1)$ and $r \in \real$, we define
\begin{equation}
    D_\nonparametric(\rho,r): o=(z,a) \mapsto \frac{\ind(a=0)}{\rho} \{\ell(z) - r\}.\label{dnp}
\end{equation}
We denote $\rho_*:=P_*(A=0) \in (0,1)$, the true proportion of data from the target population. 
It is not hard to show that 
$D_\nonparametric$ is the influence function of $\hat{r}_{\nonparametric}$;
$\hat{r}_{\nonparametric}$ is asymptotically semiparametrically efficient under a nonparametric model and $\sqrt{n} (\hat{r}_{\nonparametric} - r_*) \overset{d}{\to} \mathrm{N}(0,\sigma_{*,\nonparametric}^2)$ with $\sigma_{*,\nonparametric}^2 := \expect_{P_*}[D_\nonparametric(\rho_*,r_*)$$(O)^2]$; see Supplement~\ref{sec: proof nonparametric}.

The nonparametric estimator $\hat{r}_{\nonparametric}$ ignores data from the source population. 
If limited data from the target population is available, namely $P_*(A=0)$ is small, 
this estimator might not be accurate. 
This motivates using source population data and plausible conditions to obtain more accurate estimators.

\noindent{\bf Notations and terminology}.
We next introduce some notation and terminology. 
We will use the terms ``covariate'' and ``feature'' interchangeably, and similarly for ``label'' and ``outcome''. 
For any non-negative integers $M$ and $N$, we use $[M:N]$ to denote the index set $\{M,M+1,\ldots, N\}$ if $M \leq N$ and the empty index set otherwise; we use $[M]$ as a shorthand for $[1:M]$.
For any finite set $S$, we use $|S|$ to denote its cardinality.

For a distribution $P$, we use $P_{\sharp}$ and $P_{\sharp \mid \natural}$ to denote the marginal distribution of the random variable $\sharp$ and the conditional distribution of $\sharp \mid \natural$, respectively, under $P$; we use $P_{*,\sharp}$ and $P_{*, \sharp \mid \natural}$ to denote these distributions under $P_*$.
We use $P^n$ to denote the empirical distribution of a sample of size $n$ from $P$.
When splitting the sample into $V>0$ folds, we use $P^{n,v}$ and $P_\sharp^{n,v}$ to denote empirical distributions in fold $v\in [V]$. 
All functions considered will be measurable with respect to appropriate sigma-algebras, which will be
kept implicit.
For any function $f$, any distribution $P$, we sometimes use $P f$ to denote $\int f \intd P$.
For any $p \in [1,\infty]$, we use $\| f \|_{L^p(P)}$ to denote the $L^p(P)$ norm of $f$, namely $( \int f(x)^p P(\intd x) )^{1/p}$. We also use $L^p(P)$ to denote the space of all functions with a finite $L^p(P)$ norm, and use $L^p_0(P)$ to denote $\{f \in L^p(P): \int f \intd P = 0\}$. 
All asymptotic results are with respect to the sample size $n$ tending to infinity.

We finally review a few concepts and basic results that are central to semiparametric efficiency theory.
More thorough introductions can be found in \citet{Bickel1993,Bolthausen2002,Pfanzagl1985,Pfanzagl1990,vanderVaart1998}.
An estimator $\hat{\theta}$ of a parameter $\theta_*=\theta_*(P_*)$ is said to be asymptotically linear if $\hat{\theta} = \theta_* + n^{-1} \sum_{i=1}^n \IF(O_i) + \smallo_p(n^{-1/2})$ for a function $\IF \in L^2_0(P_*)$. This asymptotic linearity implies that $\sqrt{n} (\hat{\theta} - \theta_*) \overset{d}{\to} \mathrm{N}(0,\expect_{P_*}[\IF(O)^2])$. The function $\IF$ is called the influence function of $\hat{\theta}$.
Under a semiparametric model, there may be infinitely many influence functions, but there exists a unique efficient influence function, which is the influence function of regular asymptotically linear (RAL) estimators with the smallest asymptotic variance.
Under a nonparametric model, all RAL estimators of a parameter $\theta_*$ share the same influence function, which equals the efficient influence function.

\section{Cross-fit estimation under a general dataset shift condition} \label{sec: independence}

\subsection{Statement of condition and efficiency bound}

We consider the following general dataset shift condition characterized by sequentially identical conditional distributions introduced by \citet{Li2023}. Under this condition, some auxiliary source population datasets are informative about one component of the target population. %
In this section, we may use $Q$ to denote the target population and allow data from $Q$ not to be observed, namely $\mathcal{A}$ might not contain index 0. We still use $r_*=\expect_Q[\ell(Z)]$ to denote the target population risk.
We allow $Z$ to be a general random variable rather than just $(X,Y)$ and allow more than one source population to be present.
Thus, let $Z$ be decomposed into $K \geq 1$ components $(Z_1,\ldots,Z_K)$.
Define $\bar{Z}_0 := \emptyset$, $\bar{Z}_k := (Z_1,\ldots,Z_k)$, 
and $\mathcal{Z}_{k-1}$ to be the support of $\bar{Z}_{k-1} \mid A=0$ for $k \in [K]$.
The condition is as follows.

\setcounter{DScondition}{-1}
\begin{DScondition}[General sequential conditionals] \label{DScondition: general independence}
    For every $k \in [K]$, there exists a known nonempty subset $\mathcal{S}'_k \subset \mathcal{A}$ such that, (i) the distribution of $\bar{Z}_{k-1}$ under $Q$ is dominated by $\bar{Z}_{k-1} \mid A \in \mathcal{S}'_k$ under $P_*$, and (ii) for all $a \in \mathcal{S}'_k$, $Z_k \mid \bar{Z}_{k-1}=\bar{z}_{k-1},A=a$ is distributed identically to $Z_k \mid \bar{Z}_{k-1}=\bar{z}_{k-1}$ under $Q$ for $Q$-almost every $\bar{z}_{k-1}$ in the support 
    $\mathcal{Z}_{k-1}$ of $\bar{Z}_{k-1}$ under $Q$.
\end{DScondition}

This condition states that
conditionally on $A\in \mathcal{S}'_k$,  
$Z_k$ is independent of $A$
given $\bar{Z}_{k-1}$.
Equivalently, it states that 
every conditional distribution $Z_k \mid \bar{Z}_{k-1}$ ($k \in [K]$) under $Q$ is equal to that in the source populations with $a \in \mathcal{S}'_k$.
One important special case is when data from the target population is also observed and data from source populations are used to improve efficiency, as stated in the following condition.

\begin{DScondition2}{DScondition: general independence}[Sequential conditionals] \label{DScondition: independence}
Condition~\ref{DScondition: general independence} holds with $0 \in \mathcal{A}$ and $0 \in \mathcal{S}'_k$ for every $k$. 
\end{DScondition2}

The condition $0 \in \mathcal{S}'_k$ is purely a matter of notations since $0$ is the index for the target population $Q$. The dominance of the distribution of $\bar{Z}_{k-1}$ under $Q$ by the source population in Condition~\ref{DScondition: general independence}(i) is automatically satisfied since $0 \in \mathcal{S}'_k$. When working under this stronger condition, we use $\mathcal{S}_k$ to denote $\mathcal{S}'_k \setminus \{0\}$ for short.
We show an example of this condition in Figure~\ref{fig: illustrate}.
In particular, we allow for cases where no source population exists to supplement learning some conditional distributions, that is, $\mathcal{S}_k$ may be empty for some $k$.
We also allow irrelevant variables in some source populations to be missing; for example, in Figure~\ref{fig: illustrate}, $(Z_3,Z_4,Z_5)$ in the source population $A=2$ may be missing as these variables are not assumed to be informative about the target population.

According to the well-known review by \cite{moreno2012unifying}, the following four dataset shift conditions are among the most widely considered when one source population is available, so that $\mathcal{A}=\{0,1\}$. These conditions are all special cases of Condition~\ref{DScondition: independence}.

\begin{figure}
    \centering
    \includegraphics[scale=0.6]{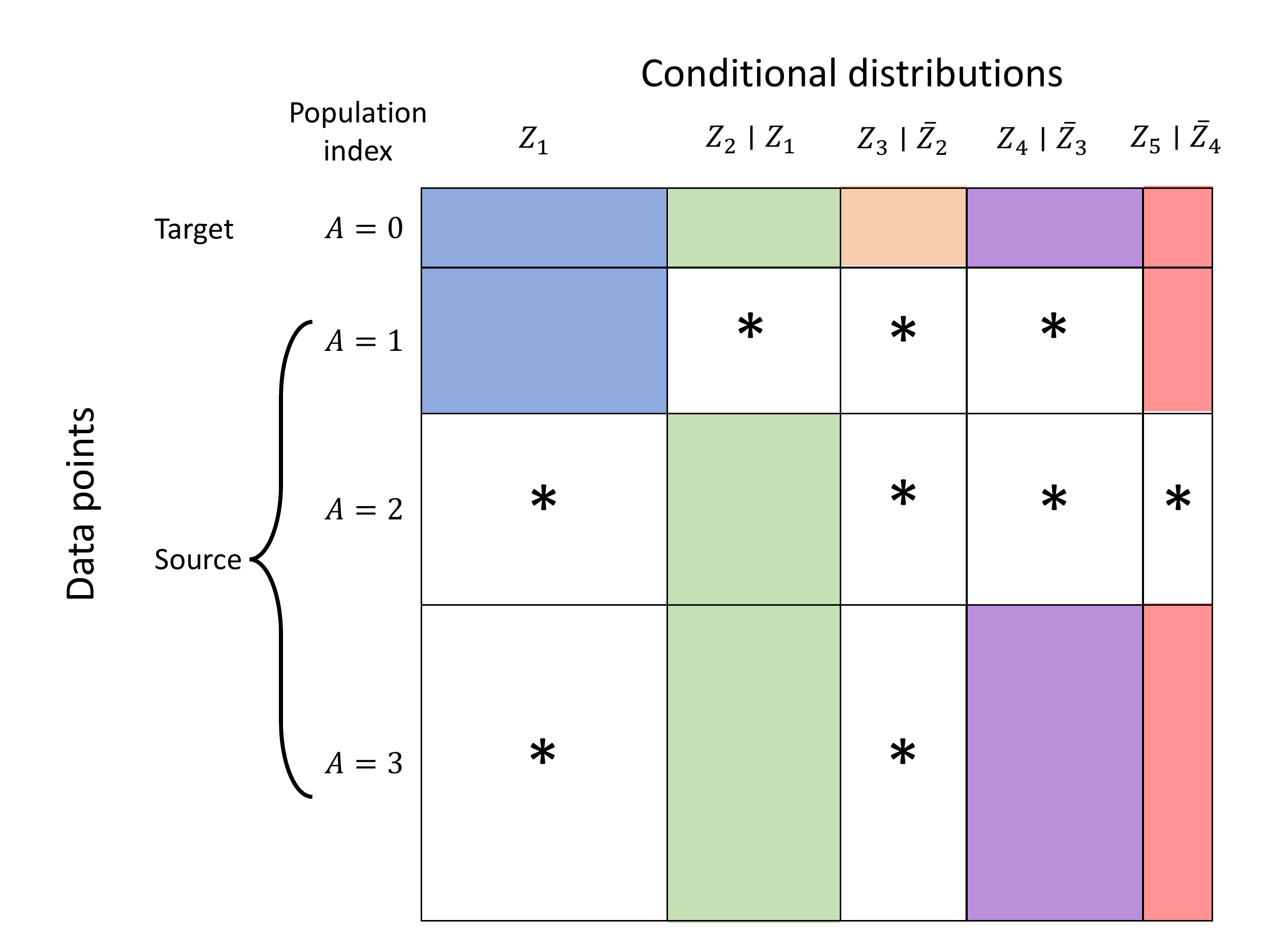}
    \caption{Illustration of Condition~\ref{DScondition: independence} with $K=5$. In each column, cells sharing the same color represent the same conditional distribution, while cells with asterisks represent conditional distributions that may arbitrarily differ from the target population ($A=0$). In this example, $\mathcal{S}_1=\{1\}$, $\mathcal{S}_2=\{2,3\}$, $\mathcal{S}_3=\emptyset$, $\mathcal{S}_4=\{3\}$, and $\mathcal{S}_5=\{1,3\}$.}
    \label{fig: illustrate}
\end{figure}

\begin{DScondition}[Concept shift in the features] \label{DScondition: X con shift}
    $X \independent A$.
\end{DScondition}
\begin{DScondition}[Concept shift in the labels] \label{DScondition: Y con shift}
    $Y \independent A$.
\end{DScondition}
\begin{DScondition}[Full-data covariate shift] \label{DScondition: cov shift}
    $Y \independent A \mid X$.
\end{DScondition}
\begin{DScondition}[Full-data label shift] \label{DScondition: label shift}
    $X \independent A \mid Y$.
\end{DScondition}

Condition~\ref{DScondition: independence} reduces to \ref{DScondition: X con shift}---concept shift in the features---by setting $K=2$, $\mathcal{A}=\{0,1\}$, $\mathcal{S}_1=\{1\}$, $\mathcal{S}_2 = \emptyset$, $Z_1=X$, and $Z_2=Y$.
Indeed, Condition~\ref{DScondition: independence} for $k=1$ states that $(Z_1|A=1) =_d (Z_1|A=0)$, or equivalently that 
 $(X|A=1) =_d (X|A=0)$, which means that $X \independent A$.
Since $\mathcal{S}_2 = \emptyset$,
Condition~\ref{DScondition: independence} for $k=2$ does not impose additional constraints.
Similarly, 
Condition~\ref{DScondition: independence}  reduces to \ref{DScondition: Y con shift}---concept shift in the labels---with the above choices but switching 
$Z_1=Y$ and $Z_2=X$. 

Condition~\ref{DScondition: independence} reduces to \ref{DScondition: cov shift}---full-data covariate shift---by setting $K=2$, $\mathcal{A}=\{0,1\}$, $\mathcal{S}_1 = \emptyset$, $\mathcal{S}_2=\{1\}$, $Z_1=X$, and $Z_2=Y$.
Indeed, since $\mathcal{S}_1 = \emptyset$, Condition~\ref{DScondition: independence} for $k=1$ does not impose  constraints.
For $k=2$, Condition~\ref{DScondition: independence} states that $(Z_2|Z_1,A=1) =_d (Z_2|Z_1,A=0)$, or equivalently that 
$(Y|X,A=1) =_d (Y|X,A=0)$, which means that $Y \independent A \mid X$.
Similarly, 
it reduces to \ref{DScondition: label shift}---full-data label shift---with the above choices but switching 
$Z_1=Y$ and $Z_2=X$. 
We refer to Conditions~\ref{DScondition: cov shift} and \ref{DScondition: label shift} as \emph{full-data} covariate and label shift, respectively, to emphasize that we have full observations $(X,Y,A)$ from the target population. 
For brevity, we refer to them as covariate or label shift when no confusion shall arise.
Condition~\ref{DScondition: independence} also includes more sophisticated dataset shift conditions and we provide a few examples in Supplement~\ref{sec: ex}.

Compared to Condition~\ref{DScondition: independence}, the more general condition \ref{DScondition: general independence} may also be applicable to cases without observing data from the target population, for example, covariate shift with unlabeled target population data and labeled source population data.

In our problem, a plug-in estimation approach could be to pool all data sets that share the same conditional distributions for each $k$, namely $A \in \mathcal{S}'_k$, and estimating these distributions nonparametrically. 
Then, the risk can be estimated by integrating the loss over the estimated distribution. However, as it is well known, this approach can suffer from a large bias or may limit the choice of distribution or density estimators, and typically requires a delicate case-by-case analysis to establish its accuracy
\citep[e.g.,][etc.]{McGrath2022}.

We now describe and review results on efficiency, which characterize the smallest possible asymptotic variance of a sequence of regular estimators under the dataset shift condition \ref{DScondition: general independence}. 
This will form the basis of our proposed estimator in the next section.
We first introduce a few definitions.
For Condition~\ref{DScondition: general independence}, let $\lambda^{k-1}_*$ denote the Radon-Nikodym derivative of the distribution of $\bar{Z}_{k-1}$ under $Q$ relative to that of $\bar{Z}_{k-1} \mid A \in \mathcal{S}'_k$ under $P_*$.
For Condition~\ref{DScondition: independence}, define the conditional probability of each population
$$\Pi^{k,a}_* : \bar{z}_{k} \mapsto P_*(A=a \mid \bar{Z}_{k}=\bar{z}_{k}) \ \  \text{for } k \in [0:K-1],\, a \in \mathcal{A},$$
and let $\pi^{a}_*:=\Pi^{0,a}_*= P_*(A=a)$ denote  the marginal probabilities of all populations ($a \in \mathcal{A}$);
thus, $\pi^0_*=\rho_*$ for $\rho_*$ from Section \ref{sec: setup}.
    Define the conditional odds of relevant source populations versus the target population:
    $$\theta^{k-1}_* := \frac{\sum_{a \in \mathcal{S}_k} \Pi^{k-1,a}_*}{\Pi^{k-1,0}_*} \ \  \text{for } k \in [K].$$
Under the stronger condition \ref{DScondition: independence}, it follows from Bayes' Theorem that
\begin{equation}
    \lambda^{k-1}_*=\frac{\sum_{a \in \mathcal{S}'_k} \pi^{a}_*}{ \pi^{0}_* (1+\theta^{k-1}_*) }. 
    \label{eq: independence density ratio}
\end{equation}
    For both Conditions~\ref{DScondition: general independence} and \ref{DScondition: independence}, we also define conditional means of the loss starting with $\ell^K_* := \ell$ and letting recursively
    \begin{align}
    & \ell^k_*: \bar{z}_k \mapsto \expect_{P_*}[\ell^{k+1}_*(\bar{Z}_{k+1}) \mid \bar{Z}_k=\bar{z}_k, A \in \mathcal{S}'_{k+1}] \ \  \text{for } \bar{z}_k \in \mathcal{Z}_k, \ k \in [K-1]. \label{genell}
    \end{align}
    We allow $\ell^k_*$ to take any value outside $\mathcal{Z}_k$, for example, when the support of $\bar{Z}_k \mid A \in \mathcal{S}'_{k+1}$ is larger than the support $\mathcal{Z}_k$ of  $\bar{Z}_{k} \mid A=0$. 
    Under Condition~\ref{DScondition: general independence}, $\ell^k_*(\bar{z}_k) = \expect_{Q}[\ell(Z) \mid \bar{Z}_k=\bar{z}_k]$ for $\bar{z}_k \in \mathcal{Z}_k$.
    We discuss the consequences of the non-unique definition of $\ell^k_*$ without Condition~\ref{DScondition: general independence} in more detail in Section~\ref{sec: independence bias test}.
    Let $\boldsymbol \ell_* := (\ell^k_*)_{k=1}^{K-1}$, $\boldsymbol \lambda_* := (\lambda^k_*)_{k=1}^{K-1}$, $\boldsymbol \theta_* := (\theta^k_*)_{k=1}^{K-1}$, and $\boldsymbol \pi_* := (\pi^{a}_*)_{a \in \mathcal{A}}$ be collections of true nuisances.
    For any given collections $\boldsymbol \ell = (\ell^k)_{k=1}^{K-1}$, $\boldsymbol \lambda := (\lambda^k)_{k=1}^{K-1}$, $\boldsymbol \theta= (\theta^k)_{k=1}^{K-1}$ and $\boldsymbol \pi := (\pi^{a})_{a \in \mathcal{A}}$ of nuisances, a scalar $r$ and $\ell^K := \ell$,
    define the pseudo-losses $\widetilde{\mathcal{T}}(\boldsymbol \ell,\boldsymbol \lambda,\boldsymbol \pi)$ and
    $\mathcal{T}(\boldsymbol \ell, \boldsymbol \theta,\boldsymbol \pi): \mathcal{O}\to \R$
    based on these nuisances, so that for $o=(z,a)$,\footnote{If $\theta^{k-1}(\bar{z}_{k-1})=\infty$, we set $1/(1 + \theta^{k-1}(\bar{z}_{k-1}))$ to be zero. When $\theta^{k-1}$ equals the truth $\theta^{k-1}_*$, this case can happen for $\bar{z}_{k-1}$ outside the support of $\bar{Z}_{k-1} \mid A=0$ but inside the support of $\bar{Z}_{k-1} \mid A \in \mathcal{S}_k$.}
    \begin{align}
        \widetilde{\mathcal{T}}(\boldsymbol \ell,\boldsymbol \lambda,\boldsymbol \pi)(o) &= \sum_{k=2}^K \frac{\ind(a \in \mathcal{S}'_k)} {\sum_{b \in \mathcal{S}'_k} \pi^{b}} \lambda^{k-1}(\bar{z}_{k-1}) [\ell^k(\bar{z}_k) - \ell^{k-1}(\bar{z}_{k-1})] + \frac{\ind(a \in \mathcal{S}'_1)}{\sum_{b \in \mathcal{S}'_1} \pi^{b}} \ell^1(z_1), \label{ttilde} \\
        \mathcal{T}(\boldsymbol \ell, \boldsymbol \theta,\boldsymbol \pi)(o)
         &=\sum_{k=2}^K
         \frac{\ind(a \in \mathcal{S}'_k)}
         {\pi^{0}(1 + \theta^{k-1}(\bar{z}_{k-1}))} 
         \left\{ \ell^k(\bar{z}_k) - \ell^{k-1}(\bar{z}_{k-1}) \right\} + \frac{\ind(a \in \mathcal{S}'_1)}{\pi^{0} (1+\theta^0)} \ell^1(z_1). \label{eq: DR transform}
    \end{align}
    The motivation for this transformation is similar to that for the unbiased transformation from \citet{Rotnitzky2006,Rubin2007} and the pseudo-outcome from \citet{Kennedy2020}. 
    Further, given any scalar $r$, with $\theta^0 := \sum_{a \in \mathcal{S}_1} \pi^{a}/\pi^{0}$, define
    \begin{align}
        D_\generalindependence(\boldsymbol \ell,\boldsymbol \lambda,\boldsymbol \pi,r) : o = (z,a) &\mapsto \widetilde{\mathcal{T}}(\boldsymbol \ell,\boldsymbol \lambda,\boldsymbol \pi)(o) - \frac{\ind(a \in \mathcal{S}'_1)}{\sum_{b \in \mathcal{S}'_1} \pi^{b}} r, \label{eq: general independence EIF} \\
        D_\independence(\boldsymbol \ell,\boldsymbol \theta,\boldsymbol \pi,r) : o = (z,a) &\mapsto \mathcal{T}(\boldsymbol \ell, \boldsymbol \theta,\boldsymbol \pi)(o) - \frac{\ind(a \in \mathcal{S}'_1)}{\pi^{0} (1+\theta^0)} r. \label{eq: independence EIF}
    \end{align}

A key result we will use is that the efficient influence function for estimating $r_*$ 
(i) equals $D_\generalindependence(\boldsymbol \ell_*,\boldsymbol \lambda_*,\boldsymbol \pi_*,r_*)$
under Condition~\ref{DScondition: general independence} 
and when $\lambda^{k}_*$ are uniformly bounded away from zero and infinity as a function of $\bar{z}_k \in \mathcal{Z}_k$ for all $k \in [0:K-1]$, and (ii) specializes to $D_\independence(\boldsymbol \ell_*,\boldsymbol \theta_*,\boldsymbol \pi_*,r_*)$
under Condition~\ref{DScondition: independence} 
and when $\theta^{k}_*$ are bounded functions for all $k \in [0:K-1]$.
This follows by Theorem~2 and Corollary~1 in \citet{Li2023}.
Consequently, the smallest possible asymptotic variance of a sequence of RAL estimators (scaled by $n^{1/2}$) is 
\begin{equation}
    \sigma_{*,\generalindependence}^2 := \expect_{P_*}[D_\generalindependence(\boldsymbol \ell_*,\boldsymbol \lambda_*,\boldsymbol \pi_*,r_*)(O)^2] \label{eq: general independence variance}
\end{equation}
under Condition~\ref{DScondition: general independence} and specializes to $\sigma_{*,\independence}^2 := \expect_{P_*}[D_\independence(\boldsymbol \ell_*,\boldsymbol \theta_*,\boldsymbol \pi_*,r_*)(O)^2]$ under Condition~\ref{DScondition: independence}.
Despite the possible non-unique definition of conditional mean loss $\boldsymbol \ell_*$, both $D_\generalindependence(\boldsymbol \ell_*,\boldsymbol \lambda_*,\boldsymbol \pi_*,r_*)$ and $D_\independence(\boldsymbol \ell_*,\boldsymbol \theta_*,\boldsymbol \pi_*,r_*)$ are uniquely defined under Condition~\ref{DScondition: general independence} and \ref{DScondition: independence}, respectively.
Here, we have used the odds parametrization rather than the density ratio or Radon-Nikodym derivative parametrization from \citet{Li2023} for Condition~\ref{DScondition: independence} because the former is often more convenient for estimation.

\subsection{Cross-fit risk estimator} \label{sec: independence estimator}

We next present our proposed estimator along with the motivation.
All estimators will implicitly depend on the sample size $n$, but we will sometimes omit this dependence from notations for conciseness.
We take as given 
a flexible regression method $\mathcal{K}$ estimating conditional means and a flexible classifier $\mathcal{C}$ estimating conditional odds---both taking outcome, covariates, and an index set for data points being used as inputs in order.
For example, $\mathcal{K}$ and $\mathcal{C}$ may be random forests, neural networks, gradient boosting, or an ensemble learner.
We also take as given
a flexible density ratio estimator $\mathcal{W}$, taking an index set for data points being used as input, 
which may be transformed from a classifier $\mathcal{C}$ by Bayes' Theorem or based on kernel density estimators.

We take Condition~\ref{DScondition: independence} as an example to illustrate ideas behind our proposed estimator.
One approach to constructing an efficient estimator of $r_*$ is to solve the estimating equation $\sum_{i=1}^n D_{\independence}(\widehat{\boldsymbol \ell},\widehat{\boldsymbol \theta},\widehat{\boldsymbol \pi},r)(O_i)=0$ for $r$, 
where $\widehat{\boldsymbol \ell}=\widehat{\boldsymbol \ell}_n$, $\widehat{\boldsymbol \theta}=\widehat{\boldsymbol \theta}_n$
and $\widehat{\boldsymbol \pi}=\widehat{\boldsymbol \pi}_n$ are estimators of nuisances $\boldsymbol \ell_*$, $\boldsymbol \theta_*$ and $\boldsymbol \pi_*$, respectively, and use the solution as the estimator. See Section~7.1, Part III in \cite{Bolthausen2002} for a more thorough introduction to achieving efficiency by solving an estimating equation. 
We further use sample splitting \citep{Hajek1962,bickel1982adaptive,Schick1986,Chernozhukov2018debiasedML} to allow for more flexible estimators, leading to our proposed estimator in Alg.~\ref{alg: independence estimator}. 
Splitting the sample into a fixed number $V$ of folds $I_v$ ($v\in [V]$) leads to the estimating equation $\smash{\sum_{i \in I_v} D_{\independence}(\widehat{\boldsymbol \ell}_{v},\widehat{\boldsymbol \theta}_{v},\widehat{\boldsymbol \pi}_{v},r)(O_i)=0}$ averaging over data in each every fold $v$,
with preliminary estimators $(\widehat{\boldsymbol \ell}_{v},\widehat{\boldsymbol \theta}_{v},\widehat{\boldsymbol \pi}_{v})$ using  data outside of the fold.
The solution is given in \eqref{eq: independence foldwise estimator}.
The corresponding estimator for more general condition \ref{DScondition: general independence} is similar and described in Alg.~\ref{alg: general independence estimator} in the Supplemental Material.

\begin{algorithm2}{alg: general independence estimator}
\caption{Cross-fit estimator of $r_*  = \expect_{P_*}[\ell(Z) \mid A=0]$ under Condition~\ref{DScondition: independence}} \label{alg: independence estimator}
\begin{algorithmic}[1]
\Require{Data $\{O_i=(Z_i,A_i)\}_{i=1}^n$, relevant source population sets $\mathcal{S}_k'$ ($k \in [K]$), number $V$ of folds, classifier $\mathcal{C}$, regression estimator $\mathcal{K}$}
\State Randomly split data into $V$ folds of approximately equal sizes. Denote the index set of data points in fold $v$ by $I_v$, and the index set of data points with $A \in \mathcal{S}_k'$ by $J_k$ for $k \in [K]$. 
\label{step: split data}
\For{$v \in [V]$}
    \State For all $k \in [K-1]$, estimate $\theta^k_*$
    using data out of fold $v$, by classifying $A=0$ against $A \in \mathcal{S}_{k+1}$ via the classifier $\mathcal{C}$ with covariates $\bar{Z}_{k}$ in the subsample with $A \in \mathcal{S}'_{k+1}$; that is, set
    $\hat{\theta}^k_{v} := \mathcal{C}(\ind(A \neq 0),\bar{Z}_{k},([n] \setminus I_v) \cap J_{k+1} )$
    and $\widehat{\boldsymbol \theta}_{v} := (\hat{\theta}^k_{v})_{k=1}^{K-1}$.
    \State Set $\hat{\pi}^{a}_{v} := |I_v|^{-1} \sum_{i \in I_v} \ind(A_i = a)$ for all $a \in \mathcal{A}$, $\widehat{\boldsymbol \pi}_{v} := (\hat{\pi}^{a}_{v})_{a \in \mathcal{A}}$, 
    $\hat{\theta}^0_{v} := \sum_{a \in \mathcal{S}_1} \hat{\pi}^{a}_{v}/\hat{\pi}^{0}_{v}$, and $\hat{\ell}^K_{v}$ to be $\ell$.
    \For{$k=K-1,\ldots,1$}
        \State Estimate $\ell^k_*$ using data out of fold $v$ by regressing $\hat{\ell}^{k+1}_{v}(\bar{Z}_{k+1})$ on $\bar{Z}_k$ in the subsample with $A \in \mathcal{S}'_{k+1}$; 
        that is, set
        $\hat{\ell}^k_{v} := \mathcal{K}(\hat{\ell}^{k+1}_{v}(\bar{Z}_{k+1}),\bar{Z}_k,([n] \setminus I_v) \cap J_{k+1} )$. 
        \label{step: sequential regression}
    \EndFor
    \State Set $\widehat{\boldsymbol \ell}_{v} := (\hat{\ell}^k_{v})_{k=1}^{K-1}$.
    \State Compute an estimator of $r_*$ for fold $v$:
    \begin{equation}
        \hat{r}_{v} := \frac{1}{|I_v|} \sum_{i \in I_v} \mathcal{T}(\widehat{\boldsymbol \ell}_{v}, \widehat{\boldsymbol \theta}_{v},\widehat{\boldsymbol \pi}_{v})(O_i) \label{eq: independence foldwise estimator}
    \end{equation}
\EndFor
\State Compute the cross-fit estimator combining estimators $\hat{r}_{v}$ from all folds: $\hat{r} := \frac{1}{n} \sum_{v=1}^V |I_v| \hat{r}_{v}$.
\end{algorithmic}
\end{algorithm2}

\begin{remark}[Estimation of marginal probabilities $\boldsymbol \pi_*$]
It is viable to replace the in-fold estimator $\widehat{\boldsymbol \pi}_{v}$ of $\pi_*$ with an out-of-fold estimator in Algorithms~\ref{alg: general independence estimator} and \ref{alg: independence estimator}.
These two approaches have the same theoretical properties that we will show next, and similar empirical performance.
\end{remark}

We next present sufficient conditions 
for the asymptotic efficiency and multiple robustness of the estimator $\hat{r}$.
In the following analyses without assuming Condition~\ref{DScondition: general independence}, we assume that nuisance estimators $\widehat{\boldsymbol \ell}_{v}$
can be evaluated at any point in the space $\mathcal{Z}$ containing the observation, even if that point is outside the support under $P_*$.
For illustration, we assume that, for each $k \in [2:K]$, $\| \hat{\ell}^k_v - \ell^k_\infty \|_{L^2(\nu_k)} \overset{p}{\to} 0$ for some function $\ell^k_\infty$, where $\nu_k$ denotes the distribution of $\bar{Z}_k \mid A \in \mathcal{S}'_k$ under $P^0$.
Define the \emph{oracle estimator} $h^{k-1}_{v}$ of 
$\ell^{k-1}_*$
based on $\hat{\ell}^k_{v}$, evaluated under the true distribution $P_*$, as\footnote{In all expectations involving nuisance estimators such as $(\widehat{\boldsymbol \ell}_{v}, \widehat{\boldsymbol \theta}_{v}, \widehat{\boldsymbol \pi}_{v})$, these estimators are treated as fixed and the expectation integrates over the randomness in a data point $O=(Z,A)$. For example, $\expect_{P_*}[\hat{\ell}^1_v(Z_1) \mid A=0] = \int \hat{\ell}^1_v \intd \mu$ where $\mu$ is the distribution of $Z_1 \mid A=0$ under $P_*$, and this expectation is itself random due to the randomness in $\hat{\ell}^1_v$.}
$$h^{k-1}_{v}: \bar{z}_{k-1} \mapsto \expect_{P_*}[\hat{\ell}^k_{v}(\bar{Z}_k) \mid \bar{Z}_{k-1}=\bar{z}_{k-1}, A \in \mathcal{S}'_k],$$
and the product bias term
$B_{k,v}$ as
\begin{align}
    & \frac{\sum_{a \in \mathcal{S}'_k} \pi^{a}_*}{\sum_{a \in \mathcal{S}'_k} \hat{\pi}^{a}_{v}} \expect_{P_*} \Bigg[ \left\{ \hat{\lambda}^{k-1}_{v}(\bar{Z}_{k-1}) - \lambda^{k-1}_*(\bar{Z}_{k-1}) \right\} \label{bk}\\
    &\qquad\qquad\times \left\{ h^{k-1}_{v}(\bar{Z}_{k-1}) - \hat{\ell}^{k-1}_{v}(\bar{Z}_{k-1}) \right\} \mid A \in \mathcal{S}'_k \Bigg] \nonumber \\
    &\quad+ \left\{ \frac{\sum_{a \in \mathcal{S}'_k} \pi^{a}_*}{\sum_{a \in \mathcal{S}'_k} \hat{\pi}^{a}_{v}} - \frac{\sum_{a \in \mathcal{S}'_1} \pi^a_*}{\sum_{a \in \mathcal{S}'_1} \hat{\pi}^a_v} \right\} \expect_Q \left[ h^{k-1}_{v}(\bar{Z}_{k-1}) - \hat{\ell}^{k-1}_{v}(\bar{Z}_{k-1}) \right].
    \nonumber
\end{align}
for Condition~\ref{DScondition: general independence} and Alg.~\ref{alg: general independence estimator}, which reduces to
\begin{align}
    & \frac{\sum_{a \in \mathcal{S}'_k} \pi^{a}_*}{\sum_{a \in \mathcal{S}'_k} \hat{\pi}^{a}_{v}} \expect_{P_*} \Bigg[ \left\{ \frac{\sum_{a \in \mathcal{S}'_k} \hat{\pi}^{a}_{v}}{\hat{\pi}^{0}_{v} (1+\hat{\theta}^{k-1}_{v}(\bar{Z}_{k-1}))} - \frac{\sum_{a \in \mathcal{S}'_k} \pi^{a}_*}{\pi^{0}_* (1+\theta^{k-1}_{*}(\bar{Z}_{k-1}))} \right\} \label{bkf}\\
    &\hspace{0.7in}\times \left\{ h^{k-1}_v(\bar{Z}_{k-1}) - \hat{\ell}^{k-1}_{v}(\bar{Z}_{k-1}) \right\} \mid A \in \mathcal{S}'_k \Bigg] \nonumber\\
    &\quad+ \left( \frac{\sum_{a \in \mathcal{S}'_k} \pi^{a}_*}{\sum_{a \in \mathcal{S}'_k} \hat{\pi}^{a}_{v}} - \frac{\sum_{a \in \mathcal{S}'_1} \pi^a_*}{\sum_{a \in \mathcal{S}'_1} \hat{\pi}^a_v} \right) \expect_{P_*} \left[ h^{k-1}_{v}(\bar{Z}_{k-1}) - \hat{\ell}^{k-1}_{v}(\bar{Z}_{k-1}) \mid A=0 \right].\nonumber
\end{align}
for Condition~\ref{DScondition: independence} and Alg.~\ref{alg: independence estimator} when $\hat{\lambda}^{k-1}_{v}$ is transformed from $\hat{\theta^{k-1}_{v}}$ and $\hat{\pi}_v$ as in \eqref{eq: independence density ratio}.

\begin{STcondition} \label{STcondition: independence eff}
    For every fold $v \in [V]$,
    \begin{compactenum}
        \item the following term is $\smallo_p(n^{-1/2})$ for Condition~\ref{DScondition: general independence} and Alg.~\ref{alg: general independence estimator}, or for Condition~\ref{DScondition: independence} and Alg.~\ref{alg: independence estimator}:
        \begin{equation}
        \sum_{k=2}^K B_{k,v}; \label{eq: independence remainder}
        \end{equation}

        \item the following term is $\smallo_p(1)$ for Condition~\ref{DScondition: general independence} and Alg.~\ref{alg: general independence estimator}, or for Condition~\ref{DScondition: independence} and Alg.~\ref{alg: independence estimator}, respectively:
        \begin{align}
            & \left\| \left( \sum_{a \in \mathcal{S}'_1} \hat\pi^{a}_{v} \right) \widetilde{\mathcal{T}}(\widehat{\boldsymbol \ell}_{v},\widehat{\boldsymbol \lambda}_{v},\widehat{\boldsymbol \pi}_{v}) - \left( \sum_{a \in \mathcal{S}'_1} \pi^{a}_* \right) \widetilde{\mathcal{T}}(\boldsymbol \ell_*,\boldsymbol \lambda_*,\boldsymbol \pi_*) \right\|_{L^2(P_*)} \label{eq: general independence EIF distance} \\
            \text{or} \quad & \left\| \left( \sum_{a \in \mathcal{S}'_1} \hat{\pi}^{a}_{v} \right) \mathcal{T}(\widehat{\boldsymbol \ell}_{v},\widehat{\boldsymbol \theta}_{v},\widehat{\boldsymbol \pi}_{v}) - \left( \sum_{a \in \mathcal{S}'_1} \pi^{a}_* \right) \mathcal{T}(\boldsymbol \ell_*,\boldsymbol \theta_*,\boldsymbol \pi_*) \right\|_{L^2(P_*)}. \label{eq: independence EIF distance}
        \end{align}
    \end{compactenum}
\end{STcondition}

In each part of the condition, the requirement for the more general condition \ref{DScondition: general independence} reduces to that for the more restrictive condition \ref{DScondition: independence} with $\hat{\lambda}^{k-1}_{v} = \sum_{a \in \mathcal{S}'_k} \hat{\pi}^a_v/\{\hat{\pi}^0_{v} (1+\hat{\theta}^{k-1}_{v})\}$.
To illustrate Condition~\ref{STcondition: independence eff}, define the \emph{limiting oracle estimator} $u^{k-1}$ of 
$\ell^{k-1}_*$
based on $\ell^k_\infty$ as
$$u^{k-1}: \bar{z}_{k-1} \mapsto \expect_{P_*}[\ell^k_\infty(\bar{Z}_k) \mid \bar{Z}_{k-1}=\bar{z}_{k-1}, A \in \mathcal{S}'_k].$$
By the definitions of $h^{k-1}_{v}$ and $u^{k-1}$, we have that
\begin{equation}
    h^{k-1}_{v}(\bar{Z}_{k-1}) - \hat{\ell}^{k-1}_{v}(\bar{Z}_{k-1}) = \expect_{P_*} \left[ \hat{\ell}^k_v(\bar{Z}_k) - \ell^k_\infty(\bar{Z}_k) \mid \bar{Z}_{k-1}, A \in \mathcal{S}'_k \right] + u^{k-1}(\bar{Z}_{k-1}) - \hat{\ell}^{k-1}_v(\bar{Z}_{k-1}). \label{eq: h-l in limit}
\end{equation}
Thus,
Condition~\ref{STcondition: independence eff} would hold if the nuisance estimator $(\widehat{\boldsymbol \ell}_{v},\widehat{\boldsymbol \theta}_{v})$ converge to the truth $(\boldsymbol \ell_*,\boldsymbol \theta_*)$ sufficiently fast. 
Under Condition~\ref{DScondition: general independence} or \ref{DScondition: independence}, part~2 is a consistency condition that is often mild; we discuss the other case in Section~\ref{sec: independence bias test}.
We next focus on part~1 and consider Condition~\ref{DScondition: independence} and Alg.~\ref{alg: independence estimator} first.
The term in \eqref{eq: independence remainder} is a drift term characterizing the bias of the estimated pseudo-loss $\mathcal{T}(\widehat{\boldsymbol \ell}_{v}, \widehat{\boldsymbol \theta}_{v},\widehat{\boldsymbol \pi}_{v})$ due to estimating nuisance functions. Conditions requiring such terms to be $\smallo_p(n^{-1/2})$ are prevalent in the literature on inference under nonparametric or semiparametric models and are often necessary to achieve efficiency \citep[see, e.g.,][]{Newey1994,Chen2015,Chernozhukov2017,VanderLaan2018}.
\citet{Balakrishnan2023} suggest that such $\smallo_p(n^{-1/2})$ conditions might be necessary without additional assumptions such as smoothness or sparsity on $\boldsymbol \ell_*$ or $\boldsymbol \theta_*$.
Since $\hat{\pi}^{a}_{v}$ is root-$n$ consistent for $\pi^{a}_*$, the second term in $B_k$ is $\smallo_p(n^{-1/2})$ under the mild consistency condition that all $\hat{\ell}^{k-1}_v$ are consistent for $\ell^{k-1}_*$ ($k \in [2:K]$).

By Jensen's inequality and the Cauchy-Schwarz inequality, we have that,
for each $k \in [2:K]$, the first term in $B_k$ is $\smallo_p(n^{-1/2})$ if both $\hat{\ell}^{k-1}_{v} - h^{k-1}_{v}$ and $1/(1+\hat{\theta}^{k-1}_{v}) - 1/(1+\theta^{k-1}_*)$ converge to zero in probability at rates faster than $n^{-1/4}$.
We illustrate this rate requirement for  $\hat{\ell}^{k-1}_{v} - h^{k-1}_{v}$. With $\hat{\ell}^{k-1}_{v}$ estimated by highly adaptive lasso with squared error loss, if $h^{k-1}_{v}$ has finite variation norm and $\hat{\ell}^k_{v}$ is bounded, then $\hat{\ell}^{k-1}_{v} - h^{k-1}_{v}$ diminishes at a rate faster than $n^{-1/4}$ \citep{Benkeser2016,VanderLaan2017}.
We formally present more interpretable sufficient conditions for part~1 of Condition~\ref{STcondition: independence eff} along with some examples of regression methods in Supplement~\ref{sec: sufficient STcondition: independence eff}.
We also allow one difference to converge slower as long as the other converges fast enough to compensate.
In principle, it is also possible to empirically check whether the magnitude of the term in \eqref{eq: independence remainder} is sufficiently small under certain conditions by using methods proposed by \citet{Liu2020,Liu2023falsifyDR}.
We do not pursue this direction in this paper as it is beyond the scope.
For Condition~\ref{DScondition: general independence} and Alg.~\ref{alg: general independence estimator}, sufficient conditions for part~1 of Condition~\ref{STcondition: independence eff} depend on how the Radon-Nikodym derivatives $\boldsymbol{\lambda}_*$ are estimated. If the estimators $\widehat{\boldsymbol{\lambda}}_{v}$ are based on kernel density estimators, part~1 would require strong smoothness conditions; on the other hand, if it is based on a classifier as in \eqref{eq: independence density ratio}, which is applicable for the special case of covariate shift without labeled target population data, the discussion above still applies.

\begin{STcondition} \label{STcondition: independence DR}
    For each fold $v \in [V]$, Condition~\ref{STcondition: independence eff} holds with the $\smallo_p(n^{-1/2})$ in part~1 replaced by $\smallo_p(1)$ and the $\smallo_p(1)$ in part~2 replaced by $\bigO_p(1)$.
\end{STcondition}

Condition~\ref{STcondition: independence DR} is much weaker than Condition~\ref{STcondition: independence eff}.
We illustrate this for Condition~\ref{DScondition: independence} and Alg.~\ref{alg: independence estimator}.
By \eqref{eq: h-l in limit} and the assumption that $\hat{\ell}^k_v$ converges to $\ell^k_\infty$ in probability, Condition~\ref{STcondition: independence DR} holds if, for each $k \in [2:K]$, either $1/(1+\hat{\theta}^{k-1}_{v})$ is consistent for $1/(1+\theta^{k-1}_*)$ or $\hat{\ell}^{k-1}_{v}$ is consistent for $u^{k-1}$.
Thus, for each fold $v$, there are $2^{K-1}$ possible ways for some of nuisance function estimators $\hat{\theta}^{k-1}_{v}$ and $\hat{\ell}^{k-1}_{v}$ ($k \in [2:K]$) to be inconsistent while Condition~\ref{STcondition: independence DR} still holds.
The same multiple allowance of inconsistent estimation applies to Condition~\ref{DScondition: general independence} and Alg.~\ref{alg: general independence estimator}.

\begin{remark} \label{rmk: l diverges}
    Conditions~\ref{STcondition: independence eff} and \ref{STcondition: independence DR} can hold even if the limit $\ell^k_\infty$ of the nuisance function estimator $\hat{\ell}^k_v$ does not exist. This case could happen if the support of $\bar{Z}_k \mid A \in \mathcal{S}_k$ is larger than $\mathcal{Z}_k$. Our results allow for such cases; we introduced the limit $\ell^k_\infty$ to illustrate Conditions~\ref{STcondition: independence eff} and \ref{STcondition: independence DR}.
\end{remark}

These two conditions are used in the next result on $\hat{r}$.

\begin{theorem} \label{thm: independence eff and DR}
    With nuisance estimators $\widehat{\boldsymbol \ell}_{v}$, $\widehat{\boldsymbol \lambda}_{v}$, $\widehat{\boldsymbol \theta}_{v}$ and $\widehat{\boldsymbol \pi}_{v}$ 
        in Algorithms~\ref{alg: general independence estimator} and \ref{alg: independence estimator}, define\footnote{The denominator $\sum_{a \in \mathcal{S}'_1} \hat\pi^{a}_{v}$ is nonzero with probability tending to one exponentially.}
        \begin{equation}
            \Delta_{v} := 
        \frac{\sum_{a \in \mathcal{S}'_1} \pi^{a}_{*}}{\sum_{a \in \mathcal{S}'_1} \hat\pi^{a}_{v}} \sum_{k=1}^K \expect_{P_*} \left[ h^{k-1}_{v}(\bar{Z}_{k-1}) - \hat{\ell}^k_{v}(\bar{Z}_k) \mid A=0 \right] \label{eq: Delta v}
        \end{equation}
        for every fold $v \in [V]$ and
        $\Delta := {n}^{-1} \sum_{v=1}^V |I_v| \Delta_{v}$.
        The following finite-sample expansion of $\hat{r}$ holds:
        \begin{align}
            & \hat{r} - \Delta -r_* \nonumber \\
            &= \sum_{v \in [V]} \frac{|I_v|}{n \sum_{a \in \mathcal{S}'_1} \hat{\pi}^a_v} (P^{n,v}-P_*) \left\{ \left( \sum_{a \in \mathcal{S}'_1} \hat\pi^{a}_{v} \right) \widetilde{\mathcal{T}}(\widehat{\boldsymbol \ell}_{v}, \widehat{\boldsymbol \lambda}_{v}, \widehat{\boldsymbol \pi}_{v}) - \left( \sum_{a \in \mathcal{S}'_1} \pi^{a}_* \right) \widetilde{\mathcal{T}}(\boldsymbol \ell_*,\boldsymbol \lambda_*,\boldsymbol \pi_*) \right\} \nonumber \\
            &\quad+ \sum_{v \in [V]}  \frac{|I_v| \sum_{a \in \mathcal{S}'_1} \hat{\pi}^a_v}{n} \sum_{k=2}^K B_{k,v} + \sum_{v \in [V]} \frac{|I_v| \sum_{a \in \mathcal{S}'_1} \pi^{a}_*}{n \sum_{a \in \mathcal{S}'_1} \pi^{a}_v} (P^{n,v}-P_*) D_\generalindependence(\boldsymbol \ell_*,\boldsymbol \lambda_*,\boldsymbol \pi_*,r_*). \label{eq: r hat expansion}
        \end{align}
        
        Moreover, if for all $n,k,v$,
        \begin{compactenum}
            \item $\Prob( \| \hat{\lambda}^{k-1}_v - \lambda^{k-1}_* \|_{L^2(P_*)} > a_{n,k,v}) \leq c_{n,k,v}$ and $\Prob( \| \hat{\ell}^{k-1}_v - h^{k-1}_v \|_{L^2(P_*)} > b_{n,k,v}) \leq d_{n,k,v}$ for some positive numbers $a_{n,k,v}$, $b_{n,k,v}$, $c_{n,k,v}$ and $d_{n,k,v}$, 
            \item  $\hat{\lambda}^{k-1}_v$ are bounded for all $k$ and $v$, and \item  $\expect_{P_*} |D_\generalindependence(\boldsymbol \ell_*,\boldsymbol \lambda_*,\boldsymbol \pi_*,r_*)(O)|^3 < \infty$,
        \end{compactenum}
        then for any $\epsilon>0$, there exist 
        quantities $\const_1,\const_2>0$ that may only depend on $\epsilon$, $P_*$ and the bound on $\hat{\lambda}^{k-1}_v$ only such that, for any $t > \const_1 \sum_{v \in [V]} \sum_{k=2}^K (a_{n,k,v} b_{n,k,v} + a_{n,k,v} + b_{n,k,v} + n^{-1})$, the following finite-sample confidence guarantee holds:
        \begin{align}
            \Prob(|\hat{r}-\Delta-r_*| > t) &\leq 2 \Phi \left( -\sqrt{n} \frac{ t - \const_1 \sum_{v \in [V]} \sum_{k=2}^K (a_{n,k,v} b_{n,k,v} + n^{-1/2} (a_{n,k,v} + b_{n,k,v}) + n^{-1}) }{\sigma_{*,\generalindependence}} \right) \nonumber \\
            &\quad+ \frac{\const_2}{\sqrt{n}} + \sum_{v \in [V]} \sum_{k=2}^K (c_{n,k,v} + d_{n,k,v}) + \epsilon \label{eq: r hat finite sample confidence}
        \end{align}
        where $\sigma_{*,\generalindependence}$ is defined in \eqref{eq: general independence variance} and $\Phi$ denotes the cumulative distribution function of the standard normal distribution.
    \begin{compactenum}
        \item\textnormal{Efficiency:} Under Condition~\ref{STcondition: independence eff}, with
        $\hat{r}$ in Line~9 of Alg.~\ref{alg: general independence estimator}, and $D_\generalindependence$ in \eqref{eq: general independence EIF},
        \begin{equation}
            \hat{r} - \Delta = r_* + \frac{1}{n} \sum_{i=1}^n D_\generalindependence(\boldsymbol \ell_*,\boldsymbol \lambda_*,\boldsymbol \pi_*,r_*)(O_i) + \smallo_p(n^{-1/2}). \label{eq: general independence eff}
        \end{equation}
        For $\hat{r}$ in Line~9 of Alg.~\ref{alg: independence estimator}, with $D_\independence$ in \eqref{eq: independence EIF}, this specializes to
        \begin{equation}
            \hat{r} - \Delta = r_* + \frac{1}{n} \sum_{i=1}^n D_\independence(\boldsymbol \ell_*,\boldsymbol \theta_*,\boldsymbol \pi_*,r_*)(O_i) + \smallo_p(n^{-1/2}). \label{eq: independence eff}
        \end{equation}

        \item \textnormal{Multiply robust consistency:} 
        Under Condition~\ref{STcondition: independence DR}, $\hat{r} - \Delta \overset{p}{\to} r_*$ as $n \to \infty$. %
    \end{compactenum}
    Additionally under Condition~\ref{DScondition: general independence}, $\Delta=0$ and thus $\hat{r}$ is RAL and efficient under Condition~\ref{STcondition: independence eff} and is consistent for $r_*$ under Condition~\ref{STcondition: independence DR}.
\end{theorem}

We have dropped the dependence of $\Delta_{v}$ and $\Delta$ on the sample size $n$, the nuisance estimators $(\widehat{\boldsymbol \ell}_{v}, \widehat{\boldsymbol \lambda}_{v},\widehat{\boldsymbol \theta}_{v},\widehat{\boldsymbol \pi}_{v})$ and the true distribution $P_*$ from the notation for conciseness.
To illustrate the finite-sample confidence guarantee, suppose that $\| \hat{\lambda}^{k-1}_v - \lambda^{k-1}_* \|_{L^2(P_*)} = \bigO_p(a_n)$ and $\| \hat{\ell}^{k-1}_v - h^{k-1}_v \|_{L^2(P_*)} = \bigO_p(b_n)$. Such convergence rates have been established for many flexible regression or classification methods.
        For example, for Condition~\ref{DScondition: independence}, with $d$ denoting the dimension of $Z_{K-1}$ and highly adaptive lasso used to obtain $\widehat{\boldsymbol{\ell}}_v$ and $\widehat{\boldsymbol{\theta}}_v$, one has $a_n=b_n=n^{-1/4-1/\{8(d+1)\}}$ \citep{Benkeser2016} if the nuisance functions have bounded variation norm. 
        See Supplement~\ref{sec: sufficient STcondition: independence eff} for more examples. Then, taking $c_n=d_n=\epsilon/\{2V(K-1)\}$, the finite-sample guarantee becomes
$$\Prob(|\hat{r}-\Delta-r_*| > t) \leq 2 \Phi \left( -\sqrt{n} \frac{ t - \const_1 (a_n b_n + n^{-1/2} (a_n + b_n) + n^{-1}) }{\sigma_{*,\generalindependence}} \right) + \frac{\const_2}{\sqrt{n}} + 2 \epsilon$$
for a different absolute constant $\const_1$. 
Thus, when the highly adaptive lasso is used,
the above bound equals
$2 \Phi ( -[\sqrt{n} t-\const_1 n^{-1/\{4(d+1)\}}]/\sigma_{*,\generalindependence} ) + \const_2/\sqrt{n} + 2 \epsilon$.
This leads to a non-trivial probability bound for any $t>\const_1 n^{-1/2-1/\{4(d+1)\}}$ such as $t \propto n^{-1/2}$.
Under Conditions~\ref{DScondition: independence} and \ref{STcondition: independence eff}, 
statistical inference about $r_*$ can be performed based on $\hat{r}$ and a consistent estimator of its influence function $D_\independence(\boldsymbol \ell_*,\boldsymbol \theta_*,\boldsymbol \pi_*,r_*)$; 
here, a consistent estimator of the asymptotic variance of $\hat{r}$ is
$\frac{1}{n} \sum_{v \in [V]} \sum_{i \in I_v} D_\independence(\widehat{\boldsymbol \ell}_v,\widehat{\boldsymbol \theta}_v,\widehat{\boldsymbol \pi}_v,\hat{r}_v)(O_i)^2$.
Inference about $r_*$ under the more general condition~\ref{DScondition: general independence} can be conducted similarly.
The results in Theorem~\ref{thm: independence eff and DR} under Condition~\ref{DScondition: general independence} or \ref{DScondition: independence} can be shown using standard approaches to analyzing Z-estimators \citep[see, e.g., Section~3.3 in][]{vandervaart1996}. 
However, to study the behavior of our estimator $\hat{r}$ without Condition~\ref{DScondition: general independence}, we need to carefully study the expansion of the mean of the estimating function $D_\generalindependence$ to identify the bias term $\Delta$ due to failure of Condition~\ref{DScondition: general independence}.
The proof of Theorem~\ref{thm: independence eff and DR} can be found in Supplement~\ref{sec: independence proof}.

We also have the following immediate corollary of Theorem~\ref{thm: independence eff and DR} for model comparison and model selection. Let $\ell^{(1)}$ and $\ell^{(2)}$ be two given losses. For example, for each $j \in \{1,2\}$, $\ell^{(j)}$ may be the squared error loss $z=(x,y) \mapsto (y-f^{(j)}(x))^2$ for a given predictor $f^{(j)}$. With superscript $(1)$ and $(2)$ denoting quantities or functions corresponding to these two predictors, the contrast of their risk $r_*^{(1)} - r_*^{(2)}$ informs the difference between the two models' performance. It is natural to select the model with the smaller estimated risk.
\begin{corollary} \label{corollary: model comparison}
    Under Condition~\ref{DScondition: general independence}, $\hat{r}^{(1)}-\hat{r}^{(2)}$ equals
    $$r_*^{(1)} - r_*^{(2)} + \frac{1}{n} \sum_{i=1}^n \{ D_\generalindependence(\boldsymbol \ell_*^{(1)},\boldsymbol \lambda_*,\boldsymbol \pi_*,r_*^{(1)})(O_i) - D_\generalindependence(\boldsymbol \ell_*^{(2)},\boldsymbol \lambda_*,\boldsymbol \pi_*,r_*^{(2)})(O_i) \} + \smallo_p(n^{-1/2})$$ if Condition~\ref{STcondition: independence eff} holds for both losses, and $\hat{r}^{(1)}-\hat{r}^{(2)} \overset{p}{\to} r_*^{(1)} - r_*^{(2)}$ if Condition~\ref{STcondition: independence DR} holds for both losses.
    Moreover, if $r_*^{(1)} - r_*^{(2)} = -C/\sqrt{n}$ for a constant $C > 0$ and Condition~\ref{STcondition: independence eff} holds for both losses, then the probability for the estimated sign of $r_*^{(1)} - r_*^{(2)}$ to be correct $\Prob(\hat{r}^{(1)} < \hat{r}^{(2)}) \to \Phi(C/\chi)$ (as $n \to \infty$) where $\chi := (\expect_P [\{ D_\generalindependence(\boldsymbol \ell_*^{(1)},\boldsymbol \lambda_*,\boldsymbol \pi_*,r_*^{(1)})(O) - D_\generalindependence(\boldsymbol \ell_*^{(2)},\boldsymbol \lambda_*,\boldsymbol \pi_*,r_*^{(2)})(O) \}^2])^{1/2}$ is the asymptotic standard deviation of the estimator $\hat{r}^{(1)} - \hat{r}^{(2)}$.
\end{corollary}
The nuisances $\boldsymbol \lambda_*$ and $\boldsymbol \pi_*$ are invariant with respect to the loss and so need not be estimated twice for the two losses.
Under Condition~\ref{STcondition: independence eff}, inference about the risk contrast $r_*^{(1)} - r_*^{(2)}$ can be conducted based on $\hat{r}^{(1)}-\hat{r}^{(2)}$ and a consistent estimator of its asymptotic variance, $\frac{1}{n} \sum_{v \in [V]} \sum_{i \in I_v} \{ D_\generalindependence(\boldsymbol \ell_*^{(1)},\boldsymbol \lambda_*,\boldsymbol \pi_*,r_*^{(1)})(O_i) - D_\generalindependence(\boldsymbol \ell_*^{(2)},\boldsymbol \lambda_*,\boldsymbol \pi_*,r_*^{(2)})(O_i) \}^2$. 
In the above example of model comparison, rejection of the null hypothesis $r_*^{(1)} - r_*^{(2)}=0$ indicates a significant difference between the two predictors' performance in the target population; otherwise, the two predictors might perform similarly.

In the example of selecting the model with a smaller risk based on the estimated sign of risk difference $r_*^{(1)} - r_*^{(2)}$, consider the case where the true risk difference is small even in relatively large samples, namely the asymptotic regime where $r_*^{(1)} - r_*^{(2)} = -C/\sqrt{n}$ for a constant $C$ as in Corollary~\ref{corollary: model comparison}. Without loss of generality, let $C>0$. If the risk difference is estimated with the nonparametric estimator, then the asymptotic probability of estimating the correct sign is asymptotically $\Phi(C/\chi_\nonparametric)$ where $\chi_\nonparametric := (\expect_P[\{ D_\nonparametric(\rho,r_*^{(1)})(O)-D_\nonparametric(\rho,r_*^{(2)})(O) \} ^2])^{1/2}$ is the asymptotic standard deviation of the nonparametric estimator. Because our proposed estimator is asymptotically efficient, $\chi_\nonparametric^2 \geq \chi^2$ and our proposed estimator results in a greater probability of selecting the better-performing model than using the nonparametric estimator.
Similarly, consider $J+1$ risks $r_*^{(1)},\ldots,r_*^{(J+1)}$ such that $r_*^{(1)} - r_*^{(j)} = -C_{j-1}/\sqrt{n}$ ($C_{j-1}>0$, $j=2,\ldots,J+1$). Suppose that the asymptotic covariance matrices of $(\hat{r}^{(1)}-\hat{r}^{(2)},\ldots,\hat{r}^{(1)}-\hat{r}^{(J+1)})$ and $(\hat{r}_\nonparametric^{(1)}-\hat{r}_\nonparametric^{(2)},\ldots,\hat{r}_\nonparametric^{(1)}-\hat{r}_\nonparametric^{(J+1)})$ are $\Sigma$ and $\Sigma_\nonparametric$, respectively. Then the asymptotic probabilities of estimating the correct smallest risk based on these estimators are $\Prob(\mathrm{N}_J(0,\Sigma) < (C_1,\ldots,C_J))$ and $\Prob(\mathrm{N}_J(0,\Sigma_\nonparametric) < (C_1,\ldots,C_J))$, respectively, where $\mathrm{N}_J$ denotes the $J$-dimensional normal distribution, and the events in these probabilities mean that a $J$-dimensional normal random vector is less than the vector $(C_1,\ldots,C_J)$ entrywise; the former asymptotic probability is greater than or equal to the latter because $\Sigma \preceq \Sigma_\nonparametric$.

\subsection{Discussion on estimation of conditional mean loss function}

In contrast to our approach 
in
Alg.~\ref{alg: independence estimator}, a
\emph{direct regression method} for 
estimating nuisance functions $\boldsymbol \ell_*$ is to regress $\ell(Z)$ on covariates in the subsample with $A=0$, since $\ell^k_*(\bar{z}_k) = \expect_{Q}[\ell(Z) \mid \bar{Z}_k=\bar{z}_k]$ for $\bar{z}_k \in \mathcal{Z}_k$ under Condition~\ref{DScondition: general independence}.
Direct regression aims to estimate nuisance functions $\ell^k_*$ rather than $h^k_*$, and can also achieve efficiency under Condition~\ref{DScondition: independence} when all nuisance functions are estimated consistently at sufficient rates to satisfy Condition~\ref{STcondition: independence eff}.

However, our sequential regression approach is advantageous in achieving multiple robustness under less stringent conditions on nuisance function estimators.
We illustrate this advantage for Condition~\ref{DScondition: independence} and Alg.~\ref{alg: independence estimator} as an example.
Note that $B_k$ from \eqref{bkf} appearing in the sum \eqref{eq: independence remainder} involves the difference between the nuisance estimator $\hat{\ell}^{k-1}_v$ and the oracle estimator $h^{k-1}_v$, which depends on the nuisance estimator $\hat{\ell}^k_v$ in the previous step. 
Each nuisance function estimator $\hat{\ell}^k_v$ (except those with indices $k=K$ and $k=2$) appears in both $B_k$ and $B_{k-1}$ in \eqref{eq: independence remainder}.

In the direct regression method, we might wish to achieve $2^{K-1}$-robustness in the sense that the final risk estimator is consistent for $r_*$ if, for each $k \in [K-1]$, either $\hat{\ell}^k_v$ or $\hat{\theta}^k_v$ is consistent. 
Consider a fixed index $j \in [K-1]$.
If all nuisance estimators in $\widehat{\boldsymbol \ell}_v$ except $\hat{\ell}^j_v$ are consistent, then neither of $h^j_v-\hat{\ell}^j_v$ and $h^{j-1}_v-\hat{\ell}^{j-1}_v$ would converge to zero.
If $\hat{\theta}^j_v$ is consistent but $\hat{\theta}^{j-1}_v$ is not, then, in general, \eqref{eq: independence remainder} would not converge to zero, and thus the risk estimator is inconsistent.
In other words, the above approach might not achieve the desired $2^{K-1}$-robustness property.

In contrast, our sequential regression approach directly aims to estimate the oracle regression functions $h^k_v$ and the estimator $\hat{\ell}^k_v$ inherits the potential bias in $\hat{\ell}^{k+1}_v$. 
For example, if all differences $h^k_v-\hat{\ell}^k_v$ except $h^j_v-\hat{\ell}^j_v$ converge to zero, in order to make \eqref{eq: independence remainder} small, it is indeed sufficient for $1/(1+\hat{\theta}^j_v)$ to be consistent for $1/(1+\theta^j_*)$. 
It seems that such a multiple robustness property could only hold for the  direct regression method under stringent or even implausible conditions on the nuisance function estimators $\hat{\ell}^k_v$.

\subsection{Consequences of the failure of Condition~\ref{DScondition: general independence}} \label{sec: independence bias test}

In this section, we focus on Condition~\ref{DScondition: independence} and Alg.~\ref{alg: independence estimator} for illustration. All results apply to the more general condition \ref{DScondition: general independence} unless otherwise stated.
We first show the intriguing fact that, 
when Condition~\ref{DScondition: general independence} fails, the conditional mean loss $\ell^k_*$ might not be uniquely defined even in $\mathcal{Z}_k$. 
The reason is that the supports of $\bar{Z}_k \mid A=a$
may be potentially mismatched
 across $Q$ and $a \in \mathcal{S}'_k$ for $k \in [K]$, 
 and moreover that for $j>k$, $\ell^j_*$ may not be uniquely defined. 
 
Consider the following simple example with $K=3$.
Suppose that the support of $\bar{Z}_2 \mid A \in \mathcal{S}'_3$ is one point $\{(0,0)\}$, while the supports of $\bar{Z}_{2} \mid A \in \mathcal{S}'_{2}$ and $\bar{Z}_2 \mid A \in \mathcal{S}'_{1}$ are both two points $\{(0,0),(0,1)\}$.
Suppose that $\mathcal{Z}_2=\{(0,0)\}$ and therefore $\mathcal{Z}_1=\{0\}$.
Then, $\ell^{2}_*$ is non-uniquely defined at $(0,1)$.
This non-unique definition is allowed for under Condition~\ref{DScondition: independence}.
However, since $\ell^{1}_*(\bar{Z}_{1}) = \expect_{P_*}[\ell^{2}_*(\bar{Z}_{2}) \mid \bar{Z}_{1},A \in \mathcal{S}'_{2}]$, when Condition~\ref{DScondition: independence} fails, the value of $\ell^{1}_*$ at $0 \in \mathcal{Z}_{1}$ depends on the value of $\ell^{2}_*$ at $(0,1)$, which is not uniquely defined.
In other words, $\ell^{1}_*$ is non-uniquely defined even in the support $\mathcal{Z}_{1}$ of $\bar{Z}_{1} \mid A=0$.

We note that this dependence of $\ell^{1}_*$ on $\ell^{2}_*(\bar{z}_{2})$ for $\bar{z}_{2} = (0,1) \notin \mathcal{Z}_2$ is excluded under Condition~\ref{DScondition: independence}: $(0,1)$ cannot be in the support of $\bar{Z}_2 \mid A \in \mathcal{S}'_2$ by the assumption that $\bar{Z}_2 \mid Z_1=z_1,A=0$ and $\bar{Z}_2 \mid Z_1=z_1,A \in \mathcal{S}'_2$ are identically distributed for $z_1 \in \mathcal{Z}_1$.
This non-unique definition might also be reflected in the corresponding nuisance estimator $\widehat{\boldsymbol \ell}_v$: $\hat{\ell}^k_v$ might be (unintentionally) extrapolated to outside the support of $\bar{Z}_k \mid A \in \mathcal{S}'_k$ in order to obtain the estimator $\hat{\ell}^{k-1}_{v}$ in Line~\ref{step: sequential regression}, Alg.~\ref{alg: independence estimator}.
This support issue might go undetected in the estimation procedure.
The oracle estimator $h^{k-1}_{v}$ would also depend on how $\hat{\ell}^k_{v}$ is extrapolated.

Nevertheless, without Condition~\ref{DScondition: general independence}, our results in Section~\ref{sec: independence estimator} remain valid as long as Condition~\ref{STcondition: independence eff} or \ref{STcondition: independence DR} holds for one version of the collection of true conditional mean loss functions $\boldsymbol \ell_*$.
For example, part~2 of Condition~\ref{STcondition: independence eff} would require the consistency of $\widehat{\boldsymbol \ell}_v$ for some version of $\boldsymbol \ell_*$, and $D_\independence(\boldsymbol \ell_*,\boldsymbol \theta_*,\boldsymbol \pi_*,r_*)$ in \eqref{eq: independence eff} would depend on the particular adopted version of $\boldsymbol \ell_*$. The appropriate choice of $\boldsymbol \ell_*$ often depends on the asymptotic behavior of the nuisance estimator $\widehat{\boldsymbol \ell}_v$, which might heavily depend on the particular choice of the regression technique used in the sequential regression (Line~\ref{step: sequential regression}, Alg.~\ref{alg: independence estimator}).

The choice of regression techniques can further affect the bias term $\Delta$ when Condition~\ref{DScondition: independence} fails.
Because of the potential extrapolation when evaluating and estimating $\widehat{\boldsymbol \ell}_v$, the bias term $\Delta$ can have drastically different behavior for different estimators $\widehat{\boldsymbol \ell}_v$, even if these estimators are all consistent for some $\boldsymbol \ell_*$ when restricted to $\mathcal{Z}_k$.
Consequently, $\Delta$ might not have a probabilistic limit, and so the estimator $\hat{r}$ can diverge.

We illustrate the behavior of $\Delta$ in the following example of concept shift in the features (\ref{DScondition: X con shift}), a special case of Condition~\ref{DScondition: independence}.
Under the setup of this condition,
$$\Delta_{v} \approx \expect_{P_*}[\expect_{P_*}[\ell(X,Y) \mid X,A=0]] - r_* + \expect_{P_*}[\hat{\ell}^1_{v}(X)] - \expect_{P_*}[\hat{\ell}^1_{v}(X) \mid A=0]$$
where we have dropped the estimation error of order $\bigO_p(n^{-1/2})$ in estimating $\pi^{a}_{*}$ with $\hat{\pi}^{a}_{v}$ in this approximation.
If Condition~\ref{DScondition: X con shift} in fact does not hold and the difference $\mathcal{B}$ between the support of $X \mid A=1$ and that of $X \mid A=0$ is nonempty, the asymptotic behavior of the third term $\expect_{P_*}[\hat{\ell}^1_{v}(X)]$ would depend on how the estimator $\hat{\ell}^1_{v}$ behaves asymptotically in $\mathcal{B}$, even if this estimator is known to be consistent for $x \mapsto \expect_{P_*}[\ell(X,Y) \mid X=x,A=0]$ when restricted to the support of $X \mid A=0$.
If $\hat{\ell}^1_{v}$ diverges in the region $\mathcal{B}$ as $n \to \infty$, our estimator $\hat{r}_{v}$ can diverge.
This phenomenon is fundamental and cannot be resolved by, for example, using an assumption lean approach \citep{Vansteelandt2022} because it mirrors the ill-defined nuisance functions $\ell^k_*$ at $\bar{z}_k\notin \mathcal{Z}_k$ ($k \in [K-1]$).

In practice, mismatched supports and extrapolation of the estimator $\widehat{\boldsymbol \ell}_v$ caused by failure of Condition~\ref{DScondition: general independence} might be detected from extreme values or even numerical errors when evaluating $\hat{\ell}^k_v$ at sample points, but such detection is not guaranteed.
When target population data is observed, the above analysis of our estimator $\hat{r}$---which leverages the dataset shift condition \ref{DScondition: independence} to gain efficiency---
motivates the following result that leads to a test of whether $\hat{r}$ is consistent for $r_*$.

\begin{corollary}[Testing root-$n$ consistency of $\hat{r}$] \label{corollary: test}
    Under Conditions~\ref{DScondition: independence} and \ref{STcondition: independence eff}, $\sqrt{n} (\hat{r} - \hat{r}_{\nonparametric})$ is asymptotically distributed as a normal distribution with mean zero and variance 
    $$\expect_{P_*}[\{D_\independence(\boldsymbol \ell_*,\boldsymbol \theta_*,\boldsymbol \pi_*,r_*)(O)-D_\nonparametric(\Pi_*,r_*)(O)\}^2] = \sigma_{*,\nonparametric}^2 - \sigma_{*,\independence}^2,$$
    as $n \to \infty$.
\end{corollary}
This corollary is implied by Theorem~\ref{thm: independence eff and DR} and the orthogonality between $D_\independence(\boldsymbol \ell_*,\boldsymbol \theta_*,\boldsymbol \pi_*,r_*)$ -$D_\nonparametric(\Pi_*,r_*)$ and $D_\independence(\boldsymbol \ell_*,\boldsymbol \theta_*,\boldsymbol \pi_*,r_*)$ under Condition~\ref{DScondition: independence}.
This result might not hold for the more general condition \ref{DScondition: general independence} due to potential lack of a nonparametric estimator $\hat{r}_\nonparametric$.
A specification test \citep{Hausman1978} of the dataset shift condition \ref{DScondition: independence} can be constructed based on the two estimators $\hat{r}_{\nonparametric}$ and $\hat{r}$ along with their respective standard errors $\mathrm{SE}_1$ and $\mathrm{SE}_2$. 
Under Condition~\ref{STcondition: independence eff} and the null hypothesis Condition~\ref{DScondition: independence}, the test statistic\footnote{It is viable to use a variant statistic with the denominator replaced by an asymptotic variance estimator based on the influence function $D_\independence(\boldsymbol \ell_*,\boldsymbol \theta_*,\boldsymbol \pi_*,r_*)-D_\nonparametric(\Pi_*,r_*)$.}
$$\frac{\hat{r} - \hat{r}_{\nonparametric}}{(\mathrm{SE}_1^2 - \mathrm{SE}_2^2)^{1/2}}$$
is approximately distributed as $\mathrm{N}(0,1)$ in large samples; in contrast, if Condition~\ref{DScondition: independence} does not hold, $\hat{r}$ is generally inconsistent for $r_*$ and thus the test statistic diverges as $n \to \infty$.

The aforementioned test of Condition~\ref{DScondition: independence} may be underpowered because only one loss function $\ell$ is considered.
For example, it is possible to construct a scenario where Condition~\ref{DScondition: independence} fails, while the loss function $\ell$ and the nuisance function estimators $\widehat{\boldsymbol \ell}_v$ are chosen such that $\Delta = \smallo_p(n^{-1/2})$.
In this case, the asymptotic power of the aforementioned test is no greater than the asymptotic type I error rate.
A somewhat contrived construction is to set $\widehat{\boldsymbol \ell}_v$ as one version of the true conditional mean loss $\boldsymbol \ell_*$ and choose $P_*$ such that  $\expect_{P_*}[\ell^k_*(\bar{Z}_k) \mid \bar{Z}_{k-1},A=0] = \ell^{k-1}_*(\bar{Z}_{k-1})$ holds.
This implies that $\Delta=0$, while Condition~\ref{DScondition: independence} can fail, for example, due to the heteroskedasticity of the residuals $\ell^k_*(\bar{Z}_k) - \ell^{k-1}_*(\bar{Z}_{k-1})$.
Such phenomena have also been found in specification tests for generalized method of moments \citep{Newey1985}.
More powerful tests of conditional independence that do not suffer from the above phenomenon have been proposed in other settings \citep[see, e.g.][etc.]{Doran2014,Hu2020,Shah2020,Zhang2011condindtest}.

However, since $\hat{r}$ might still be root-$n$ consistent for $r_*$ even if Condition~\ref{DScondition: independence} does not hold,
the above test should be interpreted as a test of the null hypothesis that $\hat{r}$ is root-$n$ consistent for $r_*$, a weaker null hypothesis than conditional independence (Condition~\ref{DScondition: independence}). Nevertheless, this weaker null hypothesis is meaningful when the risk $r_*$ is the estimand of interest.

\begin{remark}
    It is possible to adopt our proposed estimator when the relevant source populations $\mathcal{S}_k$ in Condition~\ref{DScondition: general independence} or \ref{DScondition: independence} are not known \textit{a priori} but can be selected based on data.
    In this case, the user can split the data into two folds, use fold 1 to select $\mathcal{S}_k$, and finally compute our estimator on fold 2.
    We leave a more thorough study of such approaches to future work.
\end{remark}

For the following Sections~\ref{sec: X con shift} and \ref{sec: cov shift}, we focus on risk estimation under one of the four popular dataset shift conditions \ref{DScondition: X con shift}--\ref{DScondition: label shift}. 
The special structures of these conditions will be further exploited in the estimation procedure, leading to additional simplifications, more flexibility in estimation, and potentially more robustness.
For example, for Conditions~\ref{DScondition: X con shift} and \ref{DScondition: Y con shift}, $K=2$ and $\mathcal{S}_2$ is known to be empty, and we will show that estimators with better robustness properties than part~2 of Theorem~\ref{thm: independence eff and DR} can be constructed; for Conditions~\ref{DScondition: cov shift} and \ref{DScondition: label shift}, $K=2$ and $\mathcal{S}_1$ is known to be empty, and the estimation procedure described in Alg.~\ref{alg: independence estimator} can be simplified.

Conditions~\ref{DScondition: X con shift} and \ref{DScondition: Y con shift} are identical up to switching the roles of $X$ and $Y$; the same holds for the other two conditions \ref{DScondition: cov shift} and \ref{DScondition: label shift}. 
Therefore, we study Conditions~\ref{DScondition: X con shift} and \ref{DScondition: cov shift} in the main body and present results for Conditions~\ref{DScondition: Y con shift} and \ref{DScondition: label shift} in Supplement~\ref{sec: label shift and Y con shift}.

\subsection{More general dataset shift condition} \label{section: weakly aligned}

\citet{Li2023weak} considered a condition 
    more general than \ref{DScondition: general independence}. 
    In addition to Condition~\ref{DScondition: general independence}, 
    for each $k\in[K]$,
    their setting also allows for a population index subset $\mathcal{V}_k \subset \mathcal{A}$ such that, for each $a \in \mathcal{V}_k$, and all $z_k, \bar{z}_{k-1}$,
    $$\intd P_{Z_k \mid \bar{Z}_{k-1},A=a} / \intd Q_{Z_k \mid \bar{Z}_{k-1}} (z_k \mid \bar{z}_{k-1}) \propto w_{k,a}(\bar{z}_k; \beta_{k,a})$$ 
    for some \emph{unknown} finite-dimensional parameter $\beta_{k,a}$ and a known tilting function $w_{k,a}$.
    This allows for dataset shift up to unknown finite-dimensional parameters.
    For instance, this includes the following cases:

    \begin{compactenum} 
        \item {\bf Example 1. Covariate shift up to exponential tilting.} We can consider a generalization of Condition~\ref{DScondition: cov shift}, where the density of $Y \mid X=x,A=1$ is proportional to the density of $Y \mid X=x,A=0$ multiplied by $w(x,y;\beta)=\exp([x,y]^\top \beta)$, for all $x,y$.
        \item {\bf Example 2. Covariate shift with truncation.}  Another generalization of Condition~\ref{DScondition: cov shift} allows $Y$ in the source data to be truncated \citep{Jewell1985,Bickel1993,Bhattacharya2007}; for example, $Y$ is only observed when it is above an unknown threshold $\beta$. In this case, the density of $Y \mid X=x,A=1$ is proportional to the density of $Y \mid X=x,A=0$ multiplied by $w(x,y;\beta)=\ind(y \geq \beta)$, for all $x,y$.
        \item {\bf Example 3. Covariate shift with clipping.} Suppose that $Y$ is integer-valued, such as a count variable. Condition~\ref{DScondition: cov shift} can be generalized to allow $Y$ in the source data to be clipped at a threshold $B$; that is, if the true outcome is above $B$, the observed $Y$ equals $B$. In this case, the density (with respect to the counting measure) of $Y \mid X=x,A=1$ is proportional to the density of $Y \mid X=x,A=0$ multiplied by $w(x,y;\beta)=\ind(y<B) + \ind(y=B) \exp(\beta)$ for an unknown normalizing parameter $\beta$,  for all $x,y$.
    \end{compactenum}

    Concept shift and label shift can be extended similarly; and these constructions also clearly apply to more the general sequential conditionals setting. 
    For this more general condition with unknown finite-dimensional parameters, \citet{Li2023weak} derived influence functions that have smaller variances those not using the source data from $\mathcal{V}_k$.

Let $D^*$ be such an influence function with a reduced variance. Similarly to the cross-fitting strategy in Alg.~\ref{alg: independence estimator}, for each fold $v\in[V]$, let  $\hat{D}_v$ denote an estimator of $D_*$ based on data out of fold $v$, which involves an estimator $\hat{\ell}^1_v$ of $\ell^1_*$. 
    With data split into $V$ folds and $P^{n,v}$ denoting the empirical distribution of data in fold $v$ ($v \in [V]$), 
    consider a cross-fit one-step estimator 
$$n^{-1} \sum_{v \in [V]} |I_v| \{ P^{n,v}_{Z_1 \mid A \in \mathcal{S}'_1} \hat{\ell}^1_v + P^{n,v} \hat{D}_v \}.$$

    If all nuisance functions are estimated well in the sense that $\| \hat{D}_v - D_* \|_{L^2(P_0)} = \smallo_p(1)$ and $\mathscr{R}_v := P^{n,v}_{Z_1 \mid A \in \mathcal{S}'_1} \hat{\ell}^1_v - r_* + P_* \hat{D}_v = \smallo_p(n^{-1/2})$ for all $v \in [V]$, then the one-step estimator $P^{n,v}_{Z_1 \mid A \in \mathcal{S}'_1} \hat{\ell}^1_v + P^{n,v} \hat{D}_v$ for fold $v$ equals
$$r_* + P^{n,v} D_* + (P^{n,v} - P_*) (\hat{D}_v - D_*) + \mathscr{R}_v = r_* + P^{n,v} D_* + \smallo_p(n^{-1/2}).$$ 
With $P^n$ denoting the empirical distribution of the entire data, the cross-fit one-step estimator equals $r_* + P^n D_* + \smallo_p(n^{-1/2})$, that is, it is asymptotically normal with 
    a \emph{reduced asymptotic variance compared to the nonparametric estimator}. 
    We leave studying multiple robustness for future work.

\section{Concept shift in the features} \label{sec: X con shift}

\subsection{Efficiency bound} \label{sec: X con shift EIF}

We first present the efficient influence function for the risk $r_*$ under concept shift in the features, where $X \independent A$ (\ref{DScondition: X con shift}).
To do so, define $\mathcal{E}_*: x \mapsto \expect_{P_*}[\ell(X,Y) \mid X=x,A=0]$, the conditional risk function in the target population. Recall that $\rho_*$ denotes $P_*(A=0)$. For scalars $\rho \in (0,1)$, $r \in \real$, and a function $\mathcal{E}:\mX\to \R$, 
define
\beq\label{eq: X con shift EIF}
D_\xconshift(\rho,\mathcal{E},r): o=(x,y,a) \mapsto \frac{1-a}{\rho} \{ \ell(x,y) - \mathcal{E}(x) \} + \mathcal{E}(x) - r.
\eeq
We next present the efficient influence function, which is implied by the efficiency bound under Condition~\ref{DScondition: independence} from \eqref{eq: independence EIF}, along with the efficiency gain of an efficient estimator.

\begin{corollary} \label{corollary: X con shift EIF}
    Under Condition~\ref{DScondition: X con shift}, the efficient influence function for the risk $r_*$ from \eqref{eq: def risk} is $D_\xconshift(\rho_*,\mathcal{E}_*,r_*)$ with 
    $D_\xconshift$
    from \eqref{eq: X con shift EIF}. 
    Thus, the smallest possible normalized limiting variance of a sequence of RAL estimators is $\sigma_{*,\xconshift}^2 := \expect_{P_*} [D_\xconshift(\rho_*,\mathcal{E}_*,r_*)(O)^2]$.
\end{corollary}

\begin{corollary} \label{corollary: X con shift eff gain}
Under conditions of Corollary~\ref{corollary: X con shift EIF}, the relative efficiency gain from using an efficient estimator is
$$1-\frac{\sigma_{*,\xconshift}^2}{\sigma_{*,\nonparametric}^2} = \frac{ (1-\rho_*) \expect_{P_*} \left[ (\mathcal{E}_*(X) - r_*)^2 \right] }{ \expect_{P_*} \left[ \expect_{P_*} \left[ \{ \ell(X,Y) - \mathcal{E}_*(X) \}^2 \mid A=0,X \right] \right] + \expect_{P_*} \left[ \{ \mathcal{E}_*(X) - r_* \}^2 \right] }.$$
\end{corollary}

Since 
$\mathcal{E}_*(X)=\expect_{P_*}[\ell(X,Y) \mid X,A=0]$, 
conditioning on $A=0$ throughout, recall the tower rule decomposition of the variance of loss $\ell(X,Y)$:
$$\expect_{P_*}[\{\ell(X,Y)-r_*\}^2 ] = \underbrace{\expect_{P_*} \left[ \expect_{P_*} \left[ \{ \ell(X,Y) - \mathcal{E}_*(X) \}^2 \mid X\right] \right]}_{\text{variability not just due to $X$}} + \underbrace{\expect_{P_*} [\{ \mathcal{E}_*(X) - r_* \}^2]}_{\text{variability due to $X$ alone}}.$$
By Corollary~\ref{corollary: X con shift eff gain}, more relative efficiency gain is achieved for estimating the true risk $r_*$ 
at a data-generating distribution $P_*$ with the following properties:
\begin{compactenum}
    \item The proportion $\rho_*$ of target population data is small.
    \item In the target population, the proportion of variance of $\ell(X,Y)$ due to $X$ alone is large compared to that not just due to $X$ but rather also due to $Y$.
\end{compactenum}

\begin{remark} \label{rmk: X con shift eff gain}
    We illustrate the second property in an example with squared error loss $\ell(x,y)=(y-f(x))^2$ for a given predictor $f$. We consider the target population and condition on $A=0$ throughout.
    Let $\mu_*: x \mapsto \expect_{P_*}[Y \mid X=x]$ be the oracle predictor and suppose that $Y=\mu_*(X)+\epsilon$ for independent noise $\epsilon \independent X$. 
    In this case, the variance of $\ell(X, Y)$ not due to $X$ is determined by the random noise $\epsilon$, while that due to $X$ is determined by the bias $f-\mu_*$. Therefore, the proportion of variance of $\ell(X,Y)$ not due to $X$ would be large if the given predictor $f$ is far from the oracle predictor $\mu_*$ heterogeneously. 
    In a related paper, \cite{azriel2021semi} showed that, for linear regression under semi-supervised learning, namely concept shift in the features, 
    there is efficiency gain only if the linear model is misspecified. 
    Our observation is an extension to more general risk estimation problems.
\end{remark}

\subsection{Cross-fit risk estimator} \label{sec: X con shift estimator}

In this section, we present our proposed estimator of the risk $r_*$, along with its theoretical properties. This estimator is described in Alg.~\ref{alg: X con shift estimator}.
This algorithm is the special case of Alg.~\ref{alg: independence estimator} with simplifications implied by Condition~\ref{DScondition: X con shift}.

\begin{algorithm}
\caption{Cross-fit estimator of risk $r_*$ under Condition~\ref{DScondition: X con shift}, concept shift in the features} \label{alg: X con shift estimator}
\begin{algorithmic}[1]
\Require{Data $\{O_i=(Z_i,A_i)\}_{i=1}^n$, number $V$ of folds, regression estimation method for $\mathcal{E}_*$}
\State Randomly split all data (from both populations) into $V$ folds. Let $I_v$ be the indices of data points in fold $v$.
\IFor{$v \in [V]$}
    Estimate $\mathcal{E}_*$ by $\hat{\mathcal{E}}^{-v}$ using data out of fold $v$.
\EndIFor
\For{$v \in [V]$} \Comment{(Obtain an estimating-equation-based estimator for fold $v$)}
    \State With $\hat{\rho}^v := |I_v|^{-1} \sum_{i \in I_v} \ind(A_i=0)$, set
    \begin{equation}\label{eq: X con shift foldwise estimator}
        \hat{r}_{\xconshift}^v := \frac{1}{|I_v|}\sum_{i \in I_v} \left\{ \frac{\ind(A_i=0)}{\hat{\rho}^v} [\ell(X_i,Y_i) - \hat{\mathcal{E}}^{-v}(X_i)] + \hat{\mathcal{E}}^{-v}(X_i) \right\}.
    \end{equation}
\EndFor
\State Obtain the cross-fit estimator: $\hat{r}_{\xconshift} := \frac{1}{n} \sum_{v=1}^V |I_v| \hat{r}_{\xconshift}^v$.
\end{algorithmic}
\end{algorithm}

We next present the efficiency of the estimator $\hat{r}_{\xconshift}$, along with its fully robust asymptotic linearity: $\hat{r}_{\xconshift}$ is asymptotically linear even if the nuisance function $\mathcal{E}_*$ is estimated inconsistently. This robustness property is stronger and more desirable than that stated in part~2 of Theorem~\ref{thm: independence eff and DR}, a multiply robust consistency.
Moreover, the efficiency of $\hat{r}_{\xconshift}$ only relies on the consistency of the nuisance estimator $\hat{\mathcal{E}}^{-v}$ with no requirement on its convergence rate. This condition is also weaker than Condition~\ref{STcondition: independence eff}, which is required by part~1 of Theorem~\ref{thm: independence eff and DR}.
The proof of this result can be found in Supplement~\ref{sec: X con shift proof}.

\begin{theorem}[Efficiency and fully robust asymptotic linearity of $\hat{r}_{\xconshift}$] \label{thm: X con shift efficiency robust}
    Suppose that there exists a function $\mathcal{E}_\infty \in L^2(P_*)$ such that $\max_{v \in [V]} \| \hat{\mathcal{E}}^{-v} - \mathcal{E}_\infty \|_{L^2(P_*)} = \smallo_p(1)$.
    Under Condition~\ref{DScondition: X con shift}, the sequence of estimators $\hat{r}_{\xconshift}$ in Line 5 of Algorithm \ref{alg: X con shift estimator} is RAL: with $r_*$ from \eqref{eq: def risk} and $D_\xconshift$ from \eqref{eq: X con shift EIF},
    \begin{align}
        & \hat{r}_{\xconshift} = r_* + \frac{1}{n} \sum_{i=1}^n \left\{ D_{\xconshift}(\rho_*,\mathcal{E}_\infty,r_*)(O_i) + \frac{\expect_{P_*} \left[ \mathcal{E}_\infty(X) \right] - r_*}{\rho_*} (1 - A_i - \rho_*) \right\} + \mathcal{B}, \label{eq: X con shift robust AL} \\
        \text{where } & \mathcal{B} := \sum_{v \in [V]} \frac{|I_v|}{n} \Bigg\{ \frac{\hat{\rho}^v-\rho_*}{\hat{\rho}^v} P_* (\hat{\mathcal{E}}^{-v} - \mathcal{E}_\infty) \nonumber \\
        &\qquad\qquad+ (P^{n,v}-P_*) \left\{ D_\xconshift(\hat{\rho}^{-v},\hat{\mathcal{E}}^{-v},\hat{r}_{\xconshift}^v) - D_\xconshift(\rho_*,\mathcal{E}_\infty,r_*) \right\} \Bigg\} = \smallo_p(n^{-1/2}). \nonumber
    \end{align}
    Moreover, if $\mathcal{E}_*$ is estimated consistently, namely $\mathcal{E}_\infty=\mathcal{E}_*$, then $\hat{r}_{\xconshift}$ is efficient:
    \begin{equation}
        \hat{r}_{\xconshift} = r_* + \frac{1}{n} \sum_{i=1}^n D_{\xconshift}(\rho_*,\mathcal{E}_*,r_*)(O_i) + \smallo_p(n^{-1/2}). \label{eq: X con shift eff}
    \end{equation}
\end{theorem}

\begin{remark}[Estimation of $\rho_*$]
It is possible to replace the in-fold estimator $\hat{\rho}^v$ of $\rho_*$ with an out-of-fold estimator in Alg.~\ref{alg: X con shift estimator}.
Unlike Alg.~\ref{alg: independence estimator}, this would lead to a different influence function when the nuisance function $\mathcal{E}_*$ is estimated inconsistently in Theorem~\ref{thm: X con shift efficiency robust}. 
In this case, the influence function of $\hat{r}_{\xconshift}$ from Theorem~\ref{thm: X con shift efficiency robust} cannot be used to construct asymptotically valid confidence intervals if $\mathcal{E}_*$ is estimated inconsistently.
\end{remark}

\begin{remark}[A semiparametric perspective on prediction-powered inference]
Our proposed estimator is distantly related to the work of \citet{Angelopoulos2023}. 
\citet{Angelopoulos2023} studied the estimation of and inferences about a risk minimizer with the aid of an arbitrary predictor under concept shift (Condition~\ref{DScondition: X con shift}). Their proposed risk estimator is essentially a special variant of Alg.~\ref{alg: X con shift estimator} without cross-fitting and with a fixed given estimator of $\mathcal{E}_*$.
Theorem~\ref{thm: X con shift efficiency robust} provides another perspective on why their proposed method is valid for an \emph{arbitrary} given nuisance estimator of the true conditional mean risk $\mathcal{E}_*$ 
and improves efficiency when the given estimator is close to the truth $\mathcal{E}_*$.
\end{remark}

\subsection{Simulation} \label{sec: X con shift sim}

We illustrate Theorem~\ref{thm: X con shift efficiency robust} and Corollary~\ref{corollary: X con shift eff gain} in a simulation study. We consider the application of estimating the mean squared error (MSE) of a given predictor $f$ that predicts the outcome $Y$ given input covariate $X$; that is, we take $\ell(x,y)=(y-f(x))^2$. We consider five scenarios:
\begin{compactenum}[(A)]
    \item The predictor $f$ is identical to the oracle predictor and the additive noise is homoskedastic. According to Corollary~\ref{corollary: X con shift eff gain} and Remark~\ref{rmk: X con shift eff gain}, there should be no efficiency gain from using our proposed estimator $\hat{r}_{\xconshift}$ compared to the nonparametric estimator $\hat{r}_{\nonparametric}$. \label{cov shift scenario: large gain}
    \item The predictor $f$ is a good linear approximation to the oracle predictor. This scenario may occur if the given predictor is fairly close to the truth. \label{cov shift scenario: medium gain}
    \item The predictor $f$ substantially differs from the oracle predictor. This scenario may occur if the given predictor has poor predictive power, possibly because of inaccurate tuning or using inappropriate domain knowledge in the training process. \label{cov shift scenario: small gain}
    \item The predictor $f$ substantially differs from the oracle predictor and the outcome $Y$ is deterministic given $X$. According to Corollary~\ref{corollary: X con shift eff gain}, there is a large efficiency gain from using our proposed estimator $\hat{r}_{\xconshift}$ compared to the nonparametric estimator $\hat{r}_{\nonparametric}$. \label{cov shift scenario: no gain}
    \item Condition~\ref{DScondition: X con shift} does not hold. \label{cov shift scenario: invalid}
\end{compactenum}
Scenarios~\ref{cov shift scenario: large gain} and \ref{cov shift scenario: no gain} are extreme cases designed for sanity checks, 
while Scenarios~\ref{cov shift scenario: medium gain} and \ref{cov shift scenario: small gain} are intermediate and more realistic. Scenario~\ref{cov shift scenario: invalid} is a relatively realistic case where the assumed dataset shift condition fails and is designed to check the robustness against assuming the wrong dataset shift condition.

More specifically, the data is generated as follows. 
We first generate the covariate $X=(X_1,X_2,X_3)$ from a trivariate normal distribution with mean zero and identity covariance matrix. For Scenarios~\ref{cov shift scenario: large gain}--\ref{cov shift scenario: no gain}, where Condition~\ref{DScondition: X con shift} holds, we generate the population indicator $A$ from $\mathrm{Bernoulli}(0.9)$ independent of $X$. In other words, $\rho_*=10\%$ of data points are from the target population and the other 90\% of data points are from the source population. The label in the source population, namely with $A=1$, is treated as missing as it is not assumed to contain any information about the target population. The label $Y$ in the target population, namely with $A=0$, is generated depending on the scenario as follows:
(A) $Y \mid X=x,A=0 \sim \mathrm{N}(\mu_*(x),5^2)$;
(B) \& (C): $Y \mid X=x,A=0 \sim \mathrm{N}(\mu_*(x),1)$; 
(D): $Y=\mu_*(X)$,
where
\begin{equation}
    \mu_*(x) = x_1 + x_2 + x_3 + 0.4 x_1 x_3 - 0.5 x_2 x_3 + \sin(x_1 + x_3). \label{eq: sim mu0}
\end{equation}
We set different predictors $f$ for these scenarios: 
\begin{compactenum}[(A)]
    \item $f$ is the truth $\mu_*$;
    \item $f$ is a linear function close to the best linear approximation to $\mu_*$ in $L^2(P_*)$-sense: $f(x)= 1.4 x_1 + x_2 + 1.4 x_3$;
    \item $f$ substantially differs from $\mu_*$: $f(x) = -1 - 3 x_1 + 0.5 x_3$;
    \item $f$ substantially differs from $\mu_*$: $f(x) = x_1$.
\end{compactenum}
For Scenario~\ref{cov shift scenario: invalid}, where Condition~\ref{DScondition: X con shift} does not hold, we include dependence of $A$ on $X$ by generating $A$ as
$$A \mid X=x \sim \mathrm{Bernoulli}(\expit\{ \cos(x_1 + x_2 x_3) + 2 x_1^2 x_2^2 + 3 |x_1 x_3| + |x_2| (0.5 - x_3) \}).$$
The resulting proportion $\rho_*$ of target population data is around 10\%, similar to the other scenarios. The outcome $Y$ is generated in the same way as Scenarios~\ref{cov shift scenario: medium gain} and \ref{cov shift scenario: small gain}. We set the fixed predictor $f$ to be the same as in Scenario~\ref{cov shift scenario: medium gain}.

We consider the following three estimators:
\texttt{np}: the nonparametric estimator $\hat{r}_{\nonparametric}$ in \eqref{eq: np estimator};
\texttt{Xconshift}: our estimator  $\hat{r}_{\xconshift}$ from Line 5 of Algorithm \ref{alg: X con shift estimator} with a consistent estimator of $\mathcal{E}_*$;
\texttt{Xconshift,mis.E}: $\hat{r}_{\xconshift}$ with an inconsistent estimator of $\mathcal{E}_*$.
To estimate nuisance functions consistently, 
we use Super Learner \citep{VanderLaan2007} whose library consists of generalized linear model, generalized additive model \citep{Hastie1990}, generalized linear model with lasso penalty \citep{hastie1995penalized,tibshirani1996regression}, and gradient boosting \citep{Mason1999,Friedman2001,Friedman2002,Chen2016} with various combinations of tuning parameters. This library contains highly flexible machine learning methods and is likely to yield consistent estimators of the nuisance function.
Super Learner is an ensemble learner that performs almost as well as the best learner in the library.
To estimate nuisance functions inconsistently,  we take the estimator as a fixed function that differs from the truth for the two extreme scenarios (\ref{cov shift scenario: large gain} and \ref{cov shift scenario: no gain}) for a sanity check, and drop gradient boosting from the above library for the other two relatively realistic scenarios (\ref{cov shift scenario: medium gain} and \ref{cov shift scenario: small gain}). 
Since neither of generalized linear model (with or without lasso penalty) and generalized additive model is capable of capturing interactions, dropping gradient boosting would yield inconsistent estimators of the nuisance function $\mathcal{E}_*$.
We consider sample sizes $n \in \{500, 1000, 2000\}$ and run 200 Monte Carlo experiments for each combination of the sample size and the scenario.

\begin{figure}
    \centering
    \includegraphics[scale=0.8]{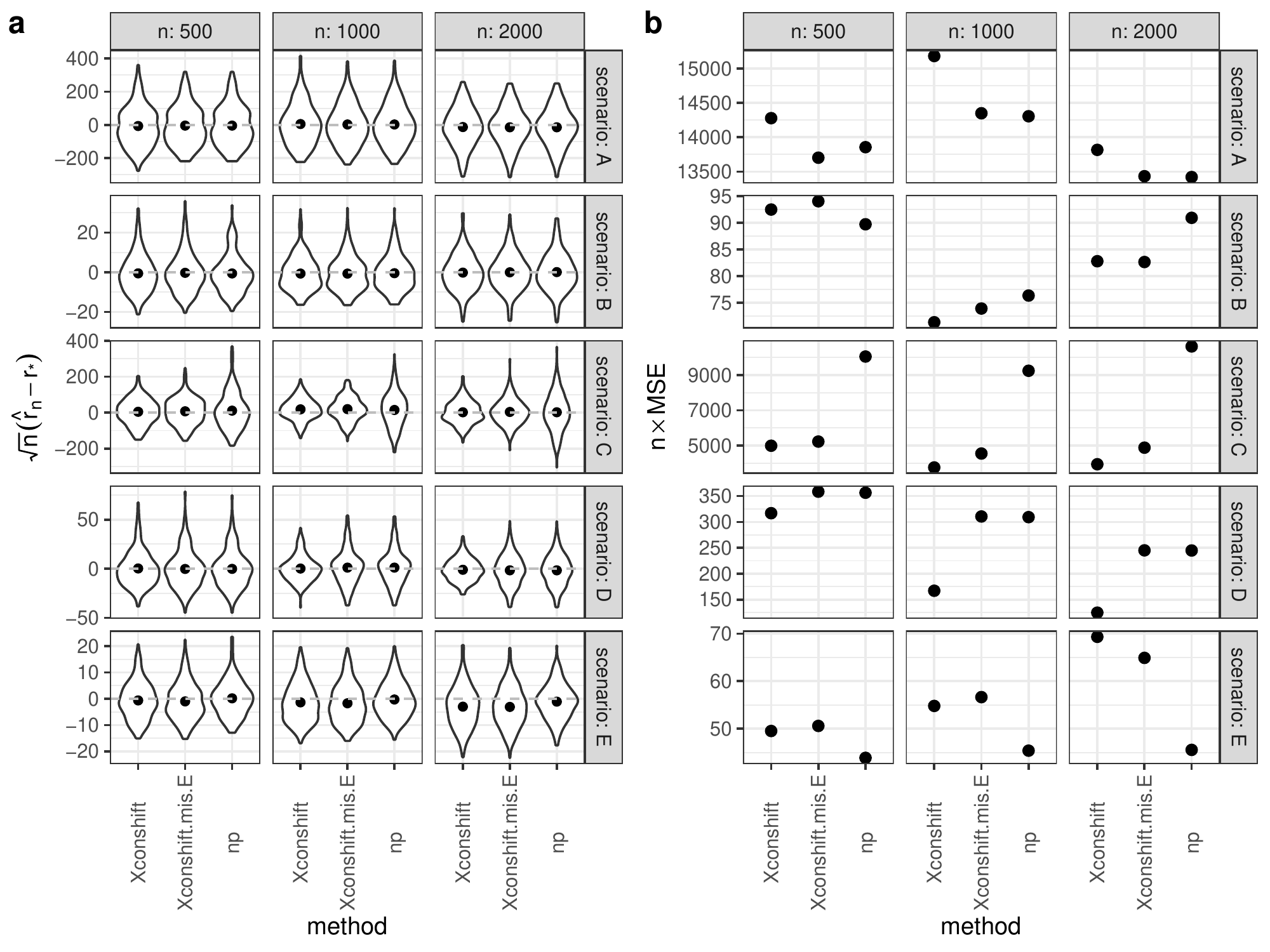}
    \caption{(a) Sampling distribution of the scaled difference between MSE estimators and the true MSE in the four scenarios under concept shift in the features. The point stands for the empirical average in Monte Carlo simulations. (b) Monte Carlo estimate of the scaled mean squared error of the estimators.}
    \label{fig: X con shift}
\end{figure}

Figure~\ref{fig: X con shift} presents the sampling distribution of the scaled difference between the three estimators of the MSE and the true MSE.
When Condition~\ref{DScondition: X con shift} holds, all three estimators appear close to normal and centered around the truth, demonstrating Theorem~\ref{thm: X con shift efficiency robust}. The variance of \texttt{Xconshift} is much smaller than that of \texttt{np} in both Scenarios~\ref{cov shift scenario: small gain} and \ref{cov shift scenario: no gain} for sample sizes 1000 and 2000, indicating a large efficiency gain; in the other two scenarios, the variance of these two estimators is comparable. These results are consistent with Corollary~\ref{corollary: X con shift eff gain} and Remark~\ref{rmk: X con shift eff gain}.
When Condition~\ref{DScondition: X con shift} does not hold (Scenario~\ref{cov shift scenario: invalid}), our proposed estimator appears consistent, indicating that $\hat{r}_{\xconshift}$ might be robust against moderate failures of the concept shift condition \ref{DScondition: X con shift}.

\section{Full-data covariate shift} \label{sec: cov shift}

\subsection{Efficiency bound} \label{sec: cov shift EIF}

Let $g_*: x \mapsto P_*(A=0 \mid X=x)$ be the propensity score function for the target population and $\mathcal{L}_*: x \mapsto \expect_{P_*}[\ell(X,Y) \mid X=x]$ be the conditional risk function. Under Condition~\ref{DScondition: cov shift}, $\mathcal{L}_*(x) = \expect_{P_*}[\ell(X,Y) \mid X=x,A=0] = \expect_{P_*}[\ell(X,Y) \mid X=x,A=1]$. 
For scalars $\rho \in (0,1)$, $r \in \real$
and functions $g, \mathcal{L}: \mathcal{X} \to \real$, define
\begin{equation}\label{dcov}
    D_\covshift(\rho,g,\mathcal{L},r): o=(x,y,a) \mapsto \frac{g(x)}{\rho} \{ \ell(x,y) - \mathcal{L}(x)\}
    + \frac{1-a}{\rho} \{ \mathcal{L}(x)- r \}.
\end{equation}

We have the following efficient influence function---implied by the efficiency bound under Condition~\ref{DScondition: independence} from \eqref{eq: independence EIF}---as well as the associated relative efficiency gain.

\begin{corollary} \label{corollary: cov shift EIF}
Under Condition~\ref{DScondition: cov shift}, if $g_*$ is bounded away from zero for almost every $x$ in the support of $X \mid A=0$, the efficient influence function for the risk $r_*$ is $D_\covshift(\rho_*, g_*, \mathcal{L}_*,r_*)$. 
Thus, the smallest possible limiting normalized asymptotic variance for a sequence of RAL estimators is $\sigma_{*,\covshift}^2 
:= \expect_{P_*} [D_\covshift$ $(\rho_*,g_*$, $\mathcal{L}_*,r_*)(O)^2]$.
\end{corollary}

\begin{corollary} \label{corollary: cov shift eff gain}
Under conditions of Corollary~\ref{corollary: cov shift EIF}, the relative efficiency gain from using an efficient estimator is
$$1-\frac{\sigma_{*,\covshift}^2}{\sigma_{*,\nonparametric}^2} = \frac{ \expect \left[ g_*(X) (1-g_*(X)) \expect_{P_*} \left[ \{ \ell(X,Y) - \mathcal{L}_*(X) \}^2 \mid X \right] \right] }{ \expect_{P_*} \left[ g_*(X) \expect_{P_*} [\{ \ell(X,Y) - \mathcal{L}_*(X) \}^2 \mid X] \right] + \expect_{P_*} \left[ g_*(X) \{ \mathcal{L}_*(X) - r_* \}^2 \right] }.$$
\end{corollary}

By Corollary~\ref{corollary: cov shift eff gain}, more efficiency gain is achieved for $P_*$ when the followings hold:
\begin{compactenum}
    \item The propensity score $g_*$ is close to zero. 
    If the covariate distribution in the source population covers that in the target population, then the proportion of source data is large.
    \item The proportion of variance of $\ell(X,Y)$ not due to $X$ is large compared to that due to $X$.
\end{compactenum}
The second property is the opposite of the implication of Corollary~\ref{corollary: X con shift eff gain} under concept shift in the features. Therefore, in the illustrating example in Remark~\ref{rmk: X con shift eff gain}, the efficiency gain under covariate shift would be large if the given predictor $f$ is close to the oracle predictor $\mu_*$.

\subsection{Cross-fit risk estimator} \label{sec: cov shift estimator}

We propose to use a cross-fit estimator based on estimating equations, as described in Alg.~\ref{alg: cov shift estimator}.
This algorithm is the special case of Alg.~\ref{alg: independence estimator} with simplifications implied by Condition~\ref{DScondition: cov shift}.

\begin{algorithm} 
\caption{Cross-fit estimator of risk $r_*$ under full-data covariate shift condition \ref{DScondition: cov shift}} \label{alg: cov shift estimator}
\begin{algorithmic}[1]
\Require{Data $\{O_i=(Z_i,A_i)\}_{i=1}^n$, number $V$ of folds, classifier to estimate $g_*$, regression estimation method for $\mathcal{L}_*$}
\State Randomly split all data (from both populations) into $V$ folds. Let $I_v$ be the indices of for fold $v$.
\IFor{$v \in [V]$}
    Estimate $(g_*,\mathcal{L}_*)$ by $(\hat{g}^{-v},\hat{\mathcal{L}}^{-v})$ using data out of fold $v$.
\EndIFor
\For{$v \in [V]$} \Comment{(Obtain an estimating-equation-based estimator for fold $v$)}
    \State With $\hat{\rho}^v := \frac{1}{|I_v|} \sum_{i \in I_v} \ind(A_i=0)$, set
    \begin{equation}\label{eq: cov shift foldwise estimator}
        \hat{r}_{\covshift}^v := \frac{1}{\hat{\rho}^v |I_v|}\sum_{i \in I_v} \{ \hat{g}^{-v}(X_i) [\ell(X_i,Y_i) - \hat{\mathcal{L}}^{-v}(X_i)] + \ind(A_i=0) \hat{\mathcal{L}}^{-v}(X_i) \}.
    \end{equation}
\EndFor
\State Obtain the cross-fit estimator:
$\hat{r}_{\covshift} := \frac{1}{n} \sum_{v=1}^V |I_v| \hat{r}_{\covshift}^v$.
\end{algorithmic}
\end{algorithm}

Compared to our proposed estimator $\hat{r}_{\xconshift}$ for concept shift, 
the estimator $\hat{r}_{\covshift}$ involves two nuisance functions rather than one. 
As we will show next, in contrast to $\hat{r}_{\xconshift}$, the efficiency of $\hat{r}_{\covshift}$ is based on sufficiently fast convergence rates---rather than consistency alone---of the nuisance function estimators, similarly to Theorem~\ref{thm: independence eff and DR}.

\begin{STcondition}[Sufficient rate of convergence for nuisance estimators] \label{STcondition: cov shift remainder}
It holds that
\begin{align*}
    & \max_{v \in [V]} \left| \int (\hat{g}^{-v} - g_*) (\hat{\mathcal{L}}^{-v} - \mathcal{L}_*) \intd P_* \right| = \smallo_p(n^{-1/2}), \\
    & \max_{v \in [V]} \left\| \hat{g}^{-v} - g_* \right\|_{L^2(P_*)} = \smallo_p(1), \qquad \max_{v \in [V]} \left\| \hat{\mathcal{L}}^{-v} - \mathcal{L}_* \right\|_{L^2(P_*)} = \smallo_p(1).
\end{align*}
\end{STcondition}
This condition implies Condition~\ref{STcondition: independence eff} under Condition~\ref{DScondition: cov shift}; 
but is perhaps a bit more interpretable.
This leads to the efficiency of $\hat{r}_{\covshift}$, whose proof has the same spirit as part~1 of Theorem~\ref{thm: independence eff and DR}.

\begin{corollary}[Efficiency of $\hat{r}_{\covshift}$] \label{corollary: cov shift efficiency}
    Under Condition~\ref{DScondition: cov shift}, with $r_*$ from \eqref{eq: def risk} and $D_{\covshift}$ from \eqref{dcov}, the estimator $\hat{r}_{\covshift}$ in Line 5 of Alg.~\ref{alg: cov shift estimator} has the following finite-sample expansion:
    \begin{align*}
        &\hat{r}_{\covshift} - r_* = \sum_{v \in [V]} \frac{|I_v| \rho_*}{n \hat{\rho}^v} (P^{n,v}-P_*)D_{\covshift}(\rho_*,g_*,\mathcal{L}_*,r_*) \\
        &\quad+ \sum_{v \in [V]} \frac{|I_v|}{n \hat{\rho}^v} (P^{n,v} - P_*) \left\{ \hat{g}^{-v} (\ell - \hat{\mathcal{L}}^{-v}) - g_* (\ell - \mathcal{L}_*) \right\} - \sum_{v \in [V]} \frac{|I_v|}{n} P_*(\hat{g}^{-v} - g_*) (\hat{\mathcal{L}}^{-v} - \mathcal{L}_*)
    \end{align*}
    Additionally under Condition~\ref{STcondition: cov shift remainder}, $\hat{r}_{\covshift}$ is regular and efficient:
    $$\hat{r}_{\covshift} = r_* + \frac{1}{n} \sum_{i=1}^n D_{\covshift}(\rho_*,g_*,\mathcal{L}_*,r_*)(O_i) + \smallo_p(n^{-1/2}).$$
    Therefore, with $\sigma_{*,\covshift}^2$ from \eqref{corollary: cov shift EIF}, $\sqrt{n} (\hat{r}_{\covshift} - r_*) \overset{d}{\to} \mathrm{N}(0, \sigma_{*,\covshift}^2)$.
\end{corollary}

Another difference between $\hat{r}_{\covshift}$ and $\hat{r}_{\xconshift}$ is that $\hat{r}_{\covshift}$ is doubly robust consistent, as implied by part~2 of Theorem~\ref{thm: independence eff and DR}, but not fully robust asymptotically linear.
The doubly robust consistency of $\hat{r}_{\covshift}$ relies on the following condition that corresponds to Condition~\ref{STcondition: independence DR} for sequential conditionals and is weaker than Condition~\ref{STcondition: cov shift remainder}.
The comparison between these conditions is similar to that between Conditions~\ref{STcondition: independence eff} and \ref{STcondition: independence DR}.

\begin{STcondition}[Consistent estimation of one nuisance function] \label{STcondition: cov shift DR consistent}
It holds that
\begin{align*}
    & \max_{v \in [V]} \left| \int (\hat{g}^{-v} - g_*) (\hat{\mathcal{L}}^{-v} - \mathcal{L}_*) \intd P_* \right| = \smallo_p(1), \\
    & \max_{v \in [V]} \left\| \hat{g}^{-v} - g_* \right\|_{L^2(P_*)} = \bigO_p(1), \qquad \max_{v \in [V]} \left\| \hat{\mathcal{L}}^{-v} - \mathcal{L}_* \right\|_{L^2(P_*)} = \bigO_p(1).
\end{align*}

\end{STcondition}

We obtain the doubly robust consistency for $\hat{r}_{\covshift}$, a corollary of part~2 of Theorem~\ref{thm: independence eff and DR}.

\begin{corollary}[Double robustness of $\hat{r}_{\covshift}$] \label{corollary: cov shift DR consistent}
    Under Conditions~\ref{DScondition: cov shift} and \ref{STcondition: cov shift remainder},
    with $r_*$ from \eqref{eq: def risk},
    the sequence of estimators $\hat{r}_{\covshift}$ from Line 5 of Alg.~\ref{alg: cov shift estimator} is consistent for $r_*$.
\end{corollary}

The differences between the asymptotic properties of $\hat{r}_{\xconshift}$ and $\hat{r}_{\covshift}$ are due to the differences between the dataset shift conditions.
Covariate shift \ref{DScondition: cov shift} is independence conditional on covariates, while concept shift \ref{DScondition: X con shift} is marginal independence. 
One aspect in which they differ is that conditional independence is often more difficult to test than marginal independence \citep{Shah2020}.

When covariate shift \ref{DScondition: cov shift} holds, compared to the nonparametric estimator $\hat{r}_{\nonparametric}$, our proposed estimator $\hat{r}_{\covshift}$ has advantages and limitations. 
In terms of efficiency, $\hat{r}_{\covshift}$ may achieve efficiency gains when both nuisance functions are estimated consistently. In terms of robustness, $\hat{r}_{\nonparametric}$ does not require estimating any nuisance function and is therefore fully robust asymptotically linear; in contrast, $\hat{r}_{\covshift}$ is only doubly robust consistent but not fully robust. 
A natural question is whether there exists a regular estimator that is fully robust asymptotically linear and also attains the efficiency bound under reasonable conditions, similarly to $\hat{r}_{\xconshift}$ under concept shift. 
Unfortunately, as the following result shows, such estimators do not exist under the common parametrizations 
$(P_X,P_{A \mid X},P_{Y \mid X})$ and $(P_A,P_{X \mid A},P_{Y \mid X})$ of the distribution $P$.
\begin{lemma} \label{lemma: cov shift no eff and robust}
    Suppose that 
    the covariate shift condition \ref{DScondition: cov shift} holds, but no further assumptions on $P_*$ are made.
        Under the parametrization $(P_X,P_{A \mid X},P_{Y \mid X})$ of a distribution $P$, suppose that for all $P_*$,
         $\IF(P_{*,X},$ $P_{A \mid X},P_{Y \mid X},r_*)$
         is an influence function for estimating $r_*$ at $P_*$, for arbitrary $(P_{A \mid X},P_{Y \mid X})$. Then, we have that
        $$\IF(P_{*,X},P_{A \mid X},P_{Y \mid X},r_*) = D_\nonparametric(\rho_*,r_*).$$
        A similar result holds under the parametrization $(P_A,P_{X \mid A},P_{Y \mid X})$ of $P$.
\end{lemma}
Therefore, if a regular estimator of $r_*$ is fully robust asymptotically linear
under either of the above parametrizations, that is, its influence function satisfies the assumptions for $\IF$ in either case of Lemma~\ref{lemma: cov shift no eff and robust}, then its influence function must be $D_\nonparametric(\rho_*,r_*)$ and cannot attain the efficiency bound. Since full-data covariate shift (\ref{DScondition: cov shift}) under the second parametrization in the above lemma is a special case of Condition~\ref{DScondition: independence} under the common parametrization $(P_A,P_{Z_1 \mid A},\ldots,P_{Z_K \mid \bar{Z}_{K-1},A})$ of distributions $P$ \citep{Li2023}, we also conclude that there is generally no regular and fully robust asymptotically linear estimator of $r_*$ that can attain the efficiency bound for Condition~\ref{DScondition: independence} under this parametrization.

We ran a simulation similar to that in Section~\ref{sec: X con shift sim}. The results, which are presented in Supplement~\ref{sec: cov shift sim}, demonstrate our theoretical results about the efficiency gains and the asymptotic behavior of our proposed estimator $\hat{r}_\covshift$.

\section{Discussion} \label{sec: discussion}

We have developed a general framework for risk estimation under data set shift. It will be interesting to understand how our methods can be used to solve statistical problems such as model training. Due to space limitation, we defer a more detailed discussion to the Supplemental Material.

\section*{Acknowledgements}
This work was supported in part by the NSF DMS 2046874 (CAREER) award, NIH grants R01AI27271, R01CA222147, R01AG065276, R01GM139926, and Analytics at Wharton.
The Africa Health Research Institute’s Demographic Surveillance Information System and Population Intervention Programme is funded by the Wellcome Trust (201433/Z/16/Z) and the South Africa Population Research Infrastructure Network (funded by the South African Department of Science and Technology and hosted by the South African Medical Research Council).

\bibliographystyle{abbrvnat}
\bibliography{ref}

\clearpage

\setcounter{page}{1}
\setcounter{section}{0}
\renewcommand{\thesection}{S\arabic{section}}%
\setcounter{table}{0}
\renewcommand{\thetable}{S\arabic{table}}%
\setcounter{figure}{0}
\renewcommand{\thefigure}{S\arabic{figure}}%
\setcounter{equation}{0}
\renewcommand{\theequation}{S\arabic{equation}}%
\setcounter{DScondition}{0}
\renewcommand{\theDScondition}{DS.S\arabic{DScondition}}%
\setcounter{STcondition}{0}
\renewcommand{\theSTcondition}{ST.S\arabic{STcondition}}%
\setcounter{lemma}{0}
\renewcommand{\thelemma}{S\arabic{lemma}}%
\setcounter{theorem}{0}
\renewcommand{\thetheorem}{S\arabic{theorem}}%
\setcounter{corollary}{0}
\renewcommand{\thecorollary}{S\arabic{corollary}}%
\setcounter{algorithm}{0}
\setcounter{example}{0}
\renewcommand{\theexample}{S\arabic{example}}%

\begin{center}
    \LARGE Supplement to \\``\ourtitle''
\end{center}

\noindent\textbf{Additional notations:} For any two quantities $a$ and $b$, we use $a \lesssim b$ to denote $a \leq C b$ for an absolute constant $C>0$. For any function $f$ and any distribution $P$, 
we may use $P f$ or $P f(X)$ to denote $\int f(x) P(\intd x)$. 
We sometimes use $\tangent_\sharp$ to denote $L_0^2(P_{*,\sharp})$ for a generic set of random variables $\sharp$. 
For generic sets of random variables $\sharp$ and $\natural$, we use $\tangent_{\sharp \mid \natural}$ to denote the function space $\{f \in L^2_0(P_*): \expect_{P_*}[ f(\sharp,\natural) \mid \natural] = 0\}$, and $\tangent_{\sharp \mid A=a}$ to denote the function space $\{f \in L^2_0(P_*): \expect_{P_*}[ f(\sharp) \mid A=a] = 0\}$. We equip the linear space $L^2(P_*)$ with the covariance inner product $\langle f, g \rangle := P_* fg = \int f g \ \intd P_*$. For any linear subspace $\mathcal{S}$ of $L_0^2(P_*)$, we use $\mathcal{S}^\perp$ to denote the orthogonal complement of $\mathcal{S}$, namely $\{f \in L_0^2(P_*): P_* fg = 0 \text{ for all } g \in \mathcal{S}\}$. We use $\closure$ to denote the $L^2(P_*)$-closure of a set. We use $\oplus$ to denote the direct sum of orthogonal linear subspaces. We may use $\ind_k$ to denote the function $a \mapsto \ind(a=k)$ for an index $k$ and $\ind_\mathcal{S}$ to denote the function $a \mapsto \ind(a \in \mathcal{S})$ for an index set $\mathcal{S}$.
For a function $f$ of a transformation $g(O)$ of a data point $O$, we also write $f(g(O))$ as $f(O)$ for convenience, which should cause no confusion; for example, we may write $\ind_0(O)$ to denote $\ind(A=0)$.
For any function $g$ and a function class $\funclass$, we use $g \funclass$ to denote the function class $\{fg: f \in \funclass\}$.

\begin{algorithm}

\caption{Cross-fit estimator of $r_* = \expect_Q[\ell(Z)]$ under Condition~\ref{DScondition: general independence}} \label{alg: general independence estimator}
\begin{algorithmic}[1]
\Require{Data $\{O_i=(Z_i,A_i)\}_{i=1}^n$, nonempty relevant source population sets $\mathcal{S}_k'$ ($k \in [K]$), number $V$ of folds, Radon-Nikodym derivative estimator $\mathcal{W}$, regression estimator $\mathcal{K}$}
\State Split data as in Line~\ref{step: split data}, Alg.~\ref{alg: independence estimator}.
\For{$v \in [V]$}
    \State For all $k=1,\ldots,K-1$, estimate $\lambda^k_*$ by $\hat{\lambda}^k_{v}$ using data out of fold $v$; that is, set $\hat{\lambda}^k_{v} := \mathcal{W}([n] \setminus I_v)$. Set $\widehat{\boldsymbol \lambda}_{v} := (\hat{\lambda}^k_{v})_{k=1}^{K-1}$.
    \State Set $\hat{\pi}^{a}_{v} := |I_v|^{-1} \sum_{i \in I_v} \ind(A_i = a)$ for all $a \in \mathcal{A}$, $\widehat{\boldsymbol \pi}_{v} := (\hat{\pi}^{a}_{v})_{a \in \mathcal{A}}$, and $\hat{\ell}^K_v$ to be $\ell$.
    \For{$k=K-1,\ldots,1$}
        \State Estimate $\ell^k_*$ by $\hat{\ell}^k_{v}$
        as in Line~\ref{step: sequential regression}, Alg.~\ref{alg: independence estimator}.
    \EndFor
    \State Set $\widehat{\boldsymbol \ell}_{v} := (\hat{\ell}^k_{v})_{k=1}^{K-1}$.
    \State Compute the following estimator for fold $v$:
    \begin{equation}
        \hat{r}_{v} := \frac{1}{|I_v|} \sum_{i \in I_v} \widetilde{\mathcal{T}}(\widehat{\boldsymbol \ell}_{v},\widehat{\boldsymbol \lambda}_{v},\widehat{\boldsymbol \pi}_{v})(O_i). \label{eq: general independence foldwise estimator}
    \end{equation}
\EndFor
\State Compute the cross-fit estimator combining estimators $\hat{r}_{v}$ from all folds:
\begin{equation}
    \hat{r} := \frac{1}{n} \sum_{v=1}^V |I_v| \hat{r}_{v}. \label{eq: general independence estimator}
\end{equation}
\end{algorithmic}

\end{algorithm}

\renewcommand{\thealgorithm}{S\arabic{algorithm}}%

\section{Further examples of sequential conditionals} \label{sec: ex}

We provide here some further examples of models satisfying the sequential conditionals condition (Condition~\ref{DScondition: independence}).

\begin{example}[Improving diagnosis with texture source data] \label{ex: lung imaging}
    \citetsupp{Christodoulidis2017} study diagnosing  interstitial lung diseases (ILD) $Y$ based on computed tomography (CT) scans $X$ by training a classifier $f: x \mapsto f(x) \in \mathcal{Y}$.
    One key task in the diagnosis is to classify the texture $W$ revealed in the CT scan image $X$. 
    Because of the high cost of collecting and human-labeling CT scans, \citetsupp{Christodoulidis2017}  leverage external texture source datasets ($A=1$) containing $(X,W)$ to train the predictive model. 
    These source datasets were not built for diagnosing ILD and thus $Y$ is missing.
    
    It may be reasonable to assume 
    that the distribution of textures given the image is the same across the sample of CT scans and the source datasets, namely that 
    $W|X=x,A=0$ and $W|X=x,A=1$ are identically distributed for $x$ in the common support of $X \mid A=0$ and $X \mid A=1$. This assumption corresponds to
    Condition~\ref{DScondition: independence} where $K=3$, $\mathcal{A}=\{0,1\}$, $Z_1=X$, $Z_2=W$, $Z_3=Y$, $\mathcal{S}_2=\{1\}$ and $\mathcal{S}_1=\mathcal{S}_3=\emptyset$.
\end{example}

\begin{example}[Partial covariate shift] \label{ex}
Consider a full-data covariate shift setting 
with data $(X,Y)$
where 
the features $X = (X_1,X_2)$ consist of two components. Suppose that $X_1$ is known to be identically distributed in the two environments $A=0$ and $A=1$. 
For instance, an investigator might be interested in image classification where the environment only partially shifts. Such a partial shift may occur when the distribution of the weather---represented by extracted features $X_2$--- changes, but the distribution of the objects---represented by other extracted features $X_1$---in the images is unchanged \citepsupp[e.g.,][etc.]{robey2020model}.
This assumption corresponds to Condition~\ref{DScondition: independence} with $K=3$, $Z_1=X_1$, $Z_2=X_2$, $Z_3=Y$, $\mathcal{A}=\{0,1\}$, $\mathcal{S}_1=\{1\}$, $\mathcal{S}_2=\emptyset$, and $\mathcal{S}_3=\{1\}$.
\end{example}

\begin{example}[Covariate \& concept shift] \label{ex: cov con shift}
    Suppose that $Z=(X,Y)$ and
    in addition to a fully observed dataset from the target population ($A=0$), two source datasets are available: one dataset ($A=1$) contains \emph{unlabeled data}---that is, $Y$ is missing---from the target population; the other ($A=2$) contains \emph{fully observed data from a relevant source population} satisfying covariate shift.
    For example, an investigator might wish to train a prediction model in a new target population with little data. The investigator might start collecting unlabeled data from the target population, label some target population data, and further leverage a large existing source data set that was assembled for a similar prediction task.
    Condition~\ref{DScondition: independence} reduces to this setup with $K=2$, $\mathcal{A}=\{0,1,2\}$, $Z_1=X$, $Z_2=Y$, $\mathcal{S}_1=\{1\}$ and $\mathcal{S}_2=\{2\}$.
\end{example}

\begin{example}[Multiphase sampling] \label{ex: expensive variables}
    Suppose that we have variables $Z_K$ that are expensive to measure (e.g., personal interview responses, thorough diagnostic evaluation, or manual field measurements) and relatively inexpensive variables $\bar{Z}_{K-1}$ (e.g., mail survey responses, screening instruments, or sensor measurements).
    The investigator might draw a random sample from the target population measuring the inexpensive variables, but only measure $Z_K$ for a random subsample. 
    We label the subsample of full datapoints with $A=0$, 
    and label subsamples of incomplete observations with other indices in $\mathcal{A}$.
    
    It might be reasonable to assume Condition~\ref{DScondition: independence}, 
    stating that the distributions of $Z_K|\bar{Z}_{K-1},$ $A=0$ and $Z_K|\bar{Z}_{K-1},A=a$ are identical for $a\in \mathcal{S}_k$ for some index set $\mathcal{S}_k$.
    This setup can arise from a two-phase sampling design \citepsupp[e.g.,][etc.]{Hansen1946,Shrout1989} for $K=2$, or multiphase sampling \citepsupp[e.g.][etc.]{Srinath1971,Tuominen2006} for larger $K$.
\end{example}

\section{Results for concept shift in the labels and full-data label shift} \label{sec: label shift and Y con shift}

In this section, we present estimators and their theoretical properties for concept shift in the labels (\ref{DScondition: Y con shift}) and full-data label shift (\ref{DScondition: label shift}). Because of their similarity to Conditions~\ref{DScondition: X con shift} and \ref{DScondition: cov shift}, we abbreviate the presentation and omit the proofs of theoretical results.
All theoretical results are numbered in parallel to those in the main text with a superscript $\dagger$.

\subsection{Concept shift in the labels} \label{sec: Y con shift}

Define $\mathcal{E}_{*,Y}: y \mapsto \expect_{P_*}[\ell(X,Y) \mid Y=y,A=0]$. For scalars $\rho \in (0,1)$, $r \in \real$, and a function $\mathcal{Y}: \mathcal{Y} \to \real$, define
$$D_\yconshift(\rho,\mathcal{E}_Y,r): o=(x,y,a) \mapsto \frac{1-a}{\rho} \{ \ell(x,y) - \mathcal{E}_Y(y) \} + \mathcal{E}_Y(y) - r.$$

\begin{corollary2}{corollary: X con shift EIF} \label{thm: Y con shift EIF}
    Under Condition~\ref{DScondition: Y con shift}, the efficient influence function for the risk $r_*$ is $D_\yconshift(\rho_*,$ $\mathcal{E}_{*,Y},r_*)$. 
    Thus, the smallest possible normalized limiting variance for a sequence of RAL estimators is $\sigma_{*,\yconshift}^2 := \expect_{P_*} [D_\yconshift(\rho_*,\mathcal{E}_{*,Y},r_*)(O)^2]$.
\end{corollary2}

\begin{corollary2}{corollary: X con shift eff gain} \label{corollary: Y con shift eff gain}
Under Condition~\ref{DScondition: Y con shift}, the relative efficiency gain from using an efficient estimator is
$$1-\frac{\sigma_{*,\yconshift}^2}{\sigma_{*,\nonparametric}^2} = \frac{ (1-\rho_*) \expect_{P_*} \left[ (\mathcal{E}_{*,Y}(y) - r_*)^2 \right] }{ \expect_{P_*} \left[ \expect_{P_*} \left[ \{ \ell(X,Y) - \mathcal{E}_{*,Y}(Y) \}^2 \mid A=0,X \right] \right] + \expect_{P_*} \left[ \{ \mathcal{E}_{*,Y}(Y) - r_* \}^2 \right] }.$$
\end{corollary2}

Our proposed estimator is described in Alg.~\ref{alg: Y con shift estimator}.

\begin{algorithm2}{alg: X con shift estimator}
\caption{Cross-fit estimator of risk $r_*$ under Condition~\ref{DScondition: Y con shift}, concept shift in the labels} \label{alg: Y con shift estimator}
\begin{algorithmic}[1]
\Require{Data $\{O_i=(Z_i,A_i)\}_{i=1}^n$, number $V$ of folds, Radon-Nikodym derivative estimator $\mathcal{W}$, regression estimation method for $\mathcal{E}_{*,Y}$}
\State Randomly split all data (from both populations) into $V$ folds. Let $I_v$ be the indices of data points in fold $v$.
\IFor{$v \in [V]$}
    Estimate $\mathcal{E}_{*,Y}$ by $\hat{\mathcal{E}}_{Y}^{-v}$ using data out of fold $v$.
\EndIFor
\For{$v \in [V]$} \Comment{(Obtain an estimating-equation-based estimator for fold $v$)}
    \State With $\hat{\rho}^v := \frac{1}{|I_v|} \sum_{i \in I_v} \ind(A_i=0)$, set
    \begin{equation}\label{eq: Y con shift foldwise estimator}
        \hat{r}_{\yconshift}^v := \frac{1}{|I_v|}\sum_{i \in I_v} \left\{ \frac{\ind(A_i=0)}{\hat{\rho}^v} [\ell(X_i,Y_i) - \hat{\mathcal{E}}_{Y}^{-v}(Y_i)] + \hat{\mathcal{E}}_{Y}^{-v}(X_i) \right\}.
    \end{equation}
\EndFor
\State Obtain the cross-fit estimator:
    \begin{equation}\label{eq: Y con shift estimator}
    \hat{r}_{\yconshift} := \frac{1}{n} \sum_{v=1}^V |I_v| \hat{r}_{\yconshift}^v.
    \end{equation}
\end{algorithmic}
\end{algorithm2}

\begin{theorem2}{thm: X con shift efficiency robust}[Efficiency and fully robust asymptotic linearity of $\hat{r}_{\yconshift}$] \label{thm: Y con shift efficiency robust}
    Suppose that there exists a function $\mathcal{E}_{\infty,Y}:\mY\to \R$ such that
    $$\max_{v \in [V]} \left\| \hat{\mathcal{E}}_{Y}^{-v} - \mathcal{E}_{\infty,Y} \right\|_{L^2(P_*)} = \smallo_p(1).$$
    Under Condition~\ref{DScondition: Y con shift}, the estimator $\hat{r}_{\yconshift}$ in \eqref{eq: Y con shift estimator} is RAL:
    $$\hat{r}_{\yconshift} = r_* + \frac{1}{n} \sum_{i=1}^n \left\{ D_{\yconshift}(\rho_*,\mathcal{E}_{\infty,Y},r_*) + \frac{\expect_{P_*} \left[ \mathcal{E}_{\infty,Y}(Y) \right] - r_*}{\rho_*} (1 - A_i - \rho_*) \right\} + \smallo_p(n^{-1/2}).$$
    Moreover, if $\mathcal{E}_{*,Y}$ is estimated consistently, namely $\mathcal{E}_{\infty,Y}=\mathcal{E}_{*,Y}$, then $\hat{r}_{\yconshift}$ is asymptotically efficient:
    $$\hat{r}_{\yconshift} = r_* + \frac{1}{n} \sum_{i=1}^n D_{\yconshift}(\rho_*,\mathcal{E}_{*,Y},r_*)(O_i) + \smallo_p(n^{-1/2}).$$
\end{theorem2}

\subsection{Full-data label shift} \label{sec: label shift}

Define $g_{*,Y}: y \mapsto P_*(A=0 \mid Y=y)$ and $\mathcal{L}_{*,Y}: y \mapsto \expect_{P_*}[\ell(X,Y) \mid Y=y]$. 
For scalars $\rho \in (0,1)$, $r \in \real$
and functions $g_Y, \mathcal{L}_Y: \mathcal{Y} \to \real$,
define
$$D_\labshift(\rho,g_Y,\mathcal{L}_Y,r): o=(x,y,a) \mapsto \frac{g_Y(y)}{\rho} \{ \ell(x,y) - \mathcal{L}_Y(y)\} + \frac{1-a}{\rho} \{ \mathcal{L}_Y(y)- r \}.$$

\begin{corollary2}{corollary: cov shift EIF} \label{thm: label shift EIF}
Under Condition~\ref{DScondition: label shift}, if $g_{*,Y}$ is bounded away from zero for almost every $x$ in the support of $Y \mid A=0$, the efficient influence function for the risk $r_*$ is $D_\labshift(\rho_*, g_{*,Y}, \mathcal{L}_{*,Y},r_*)$. 
Thus, the smallest possible limiting normalized variance for a sequence of RAL estimators is $\sigma_{*,\labshift}^2 := \expect_{P_*} [D_\labshift(\rho_*,g_{*,Y},\mathcal{L}_{*,Y},r_*)(O)^2]$.
\end{corollary2}

\begin{corollary2}{corollary: cov shift eff gain} \label{corollary: label shift eff gain}
Under Condition~\ref{DScondition: label shift}, the relative efficiency gain from using an efficient estimator is
$$1-\frac{\sigma_{*,\labshift}^2}{\sigma_{*,\nonparametric}^2} = \frac{ \expect_{P_*} \left[ g_{*,Y}(Y) (1-g_{*,Y}(Y)) \expect_{P_*} \left[ \{ \ell(X,Y) - \mathcal{L}_{*,Y}(Y) \}^2 \mid Y \right] \right] }{ \expect_{P_*} \left[ g_{*,Y}(Y) \expect_{P_*} [\{ \ell(X,Y) - \mathcal{L}_{*,Y}(Y) \}^2 \mid Y] \right] + \expect_{P_*} \left[ g_{*,Y}(Y) \{ \mathcal{L}_{*,Y}(Y) - r_* \}^2 \right] }.$$
\end{corollary2}

Our proposed estimator is described in Alg.~\ref{alg: label shift estimator}.

\begin{algorithm2}{alg: cov shift estimator}
\caption{Cross-fit estimator of risk $r_*$ under full-data label shift condition \ref{DScondition: label shift}} \label{alg: label shift estimator}
\begin{algorithmic}[1]
\Require{Data $\{O_i=(Z_i,A_i)\}_{i=1}^n$, number $V$ of folds, classifier to estimate $g_{*,Y}$, regression estimation method for $\mathcal{L}_{*,Y}$}
\State Randomly split all data (from both populations) into $V$ folds. Let $I_v$ be the indices of data points in fold $v$.
\IFor{$v \in [V]$}
    Estimate $(g_{*,Y},\mathcal{L}_{*,Y})$ by $(\hat{g}_{Y}^{-v},\hat{\mathcal{L}}_{Y}^{-v})$ using data out of fold $v$.
\EndIFor
\For{$v \in [V]$} \Comment{(Obtain an estimating-equation-based estimator for fold $v$)}
    \State With $\hat{\rho}^v := \frac{1}{|I_v|} \sum_{i \in I_v} \ind(A_i=0)$, set
    \begin{equation}\label{eq: label shift foldwise estimator}
        \hat{r}_{\labshift}^v := \frac{1}{\hat{\rho}^v |I_v|}\sum_{i \in I_v} \{ (\hat{g}_{Y}^{-v}(Y_i)) [\ell(X_i,Y_i) - \hat{\mathcal{L}}_{Y}^{-v}(Y_i)] + (1-A_i) \hat{\mathcal{L}}_{Y}^{-v}(Y_i) \}.
    \end{equation}
\EndFor
\State Obtain the cross-fit estimator:
    \begin{equation}\label{eq: label shift estimator}
    \hat{r}_{\labshift} := \frac{1}{n} \sum_{v=1}^V |I_v| \hat{r}_{\labshift}^v.
    \end{equation}
\end{algorithmic}
\end{algorithm2}

\begin{STcondition2}{STcondition: cov shift remainder}[Sufficient rate of convergence for nuisance estimators] \label{STcondition: label shift remainder}
It holds that
\begin{align*}
    & \max_{v \in [V]} \left| \int (\hat{g}_{Y}^{-v}(x) - g_{*,Y}(x)) (\hat{\mathcal{L}}_{Y}^{-v}(x) - \mathcal{L}_{*,Y}(x)) P_{*,X} (\intd x) \right| = \smallo_p(n^{-1/2}), \\
    & \max_{v \in [V]} \left\| \hat{g}_{Y}^{-v} - g_{*,Y} \right\|_{L^2(P_*)} = \smallo_p(1), \qquad \max_{v \in [V]} \left\| \hat{\mathcal{L}}_{Y}^{-v} - \mathcal{L}_{*,Y} \right\|_{L^2(P_*)} = \smallo_p(1).
\end{align*}
\end{STcondition2}

\begin{corollary2}{corollary: cov shift efficiency}[Efficiency of $\hat{r}_{\labshift}$] \label{thm: label shift efficiency}
    Under Conditions~\ref{DScondition: label shift} and \ref{STcondition: label shift remainder}, the estimator $\hat{r}_{\labshift}$ in \eqref{eq: label shift estimator} is asymptotically efficient:
    $$\hat{r}_{\labshift} = r_* + \frac{1}{n} \sum_{i=1}^n D_{\labshift}(\rho_*,g_{*,Y},\mathcal{L}_{*,Y},r_*)(O_i) + \smallo_p(n^{-1/2}).$$
    Therefore, $\sqrt{n} (\hat{r}_{\labshift} - r_*) \overset{d}{\to} \mathrm{N}(0, \expect_{P_*}[D_{\labshift}(\rho_*,g_{*,Y},\mathcal{L}_{*,Y},r_*)(O)^2])$.
\end{corollary2}

\begin{STcondition2}{STcondition: cov shift DR consistent}[Consistent estimation of one nuisance function] \label{STcondition: label shift DR consistent}
It holds that
\begin{align*}
    & \max_{v \in [V]} \left| \int (\hat{g}_Y^{-v} - g_{*,Y}) (\hat{\mathcal{L}}_Y^{-v} - \mathcal{L}_{*,Y}) \intd P_* \right| = \smallo_p(1), \\
    & \max_{v \in [V]} \left\| \hat{g}_Y^{-v} - g_{*,Y} \right\|_{L^2(P_*)} = \bigO_p(1), \qquad \max_{v \in [V]} \left\| \hat{\mathcal{L}}_Y^{-v} - \mathcal{L}_{*,Y} \right\|_{L^2(P_*)} = \bigO_p(1).
\end{align*}
\end{STcondition2}

\begin{corollary2}{corollary: cov shift DR consistent}[Doubly robust consistency of $\hat{r}_{\labshift}$] \label{thm: label shift DR consistent}
    Under Conditions~\ref{DScondition: label shift} and \ref{STcondition: label shift remainder}, $\hat{r}_{\labshift} \overset{p}{\to} r_*$.
\end{corollary2}

\begin{lemma2}{lemma: cov shift no eff and robust} \label{lemma: label shift no eff and robust}
    Suppose that 
    the label shift condition \ref{DScondition: label shift} holds, but no further assumption on $P_*$ is made.
    \begin{compactenum}
        \item Under the parametrization $(P_Y,P_{A \mid Y},P_{X \mid Y})$ of a distribution $P$, suppose that $\IF (P_{*,Y},$ $ P_{*,A \mid Y},P_{*,X \mid Y},r_*)$ is an influence function for estimating $r_*$ at $P_*$, and so is $\IF(P_{*,Y},$ $P_{A \mid Y},P_{X \mid Y},r_*)$, for arbitrary $(P_{A \mid Y},P_{X \mid Y})$. Then, we have that
        $$\IF(P_{*,Y},P_{A \mid Y},P_{X \mid Y},r_*) = D_\nonparametric(\rho_*,r_*).$$
        \item Under the parametrization $(P_A,P_{Y \mid A},P_{X \mid Y})$ of a distribution $P$, suppose that $\IF(P_{*,A},$ $P_{*,Y \mid A},P_{*,X \mid Y},r_*)$ is an influence function for estimating $r_*$ at $P_*$, and so is $\IF$ $(P_{*,A},$ $P_{Y \mid A},P_{X \mid Y},r_*)$ for arbitrary $(P_{Y \mid A}$, $P_{X \mid Y})$. 
        Then, we have that
        $$\IF(P_{*,A},P_{Y \mid A},P_{X \mid Y},r_*) = D_\nonparametric(\rho_*,r_*).$$
    \end{compactenum}
\end{lemma2}

\section{Discussion of Part~1 of Condition~\ref{STcondition: independence eff}} \label{sec: sufficient STcondition: independence eff}

Let $\nu_{k-1}$ denote the distribution of $\bar{Z}_{k-1} \mid A \in \mathcal{S}'_k$ under $P_*$.
For every $k \in [2:K]$, every fold $v \in [V]$, and any function $g: z_k \mapsto g(z_k) \in \real$, 
we denote the oracle regression by $\mathcal{R}^{k-1}_v(g): z_{k-1} \mapsto \expect_{P_*}[g(Z_k) \mid Z_{k-1}=z_{k-1}, A \in \mathcal{S}'_k]$, and use $\hat{\mathcal{R}}^{k-1}_v(g)$ to denote the estimator of $\mathcal{R}^{k-1}_v(g)$ obtained using data out of fold $v$ and the same regression technique $\mathcal{K}$ used to obtain $\hat{\ell}^{k-1}_v$ as in Alg.~\ref{alg: independence estimator}.

A sufficient condition for part~1 of Condition~\ref{STcondition: independence eff} is the following: for every $k \in [2:K]$, there exists a function class $\mathcal{G}_k$ that contains $\hat{\ell}^{k}_v$ with probability tending to one as $n \to \infty$, and such that
$$\sup_{g \in \mathcal{G}_k} \left\| \hat{\mathcal{R}}^{k-1}_{v}(g) - \mathcal{R}^{k-1}_{v}(g) \right\|_{L^2(\nu_{k-1})} \left\| \frac{1}{1+\hat{\theta}^{k-1}_{v}} - \frac{1}{1+\theta^{k-1}_*} \right\|_{L^2(\nu_{k-1})} = \smallo_p(n^{-1/2})$$
and
$$\sup_{g \in \mathcal{G}_k} \left\| \hat{\mathcal{R}}^{k-1}_{v}(g) - \mathcal{R}^{k-1}_{v}(g) \right\|_{L^2(\nu_{k-1})} = \smallo_p(1).$$
This sufficient condition is satisfied, for example, if both $\sup_{g \in \mathcal{G}_k} \| \hat{\mathcal{R}}^{k-1}_{v}(g) - \mathcal{R}^{k-1}_{v}(g) \|_{L^2(\nu_{k-1})}$ and $\| 1/(1+\hat{\theta}^{k-1}_{v}) - 1/(1+\theta^{k-1}_*) \|_{L^2(\nu_{k-1})}$ are $\smallo_p(n^{-1/4})$.
Our conditions \ref{STcondition: cov shift remainder} and \ref{STcondition: label shift remainder} for the estimators for Conditions~\ref{DScondition: cov shift} and \ref{DScondition: label shift} respectively, are similar to the above sufficient condition.
Thus, this section also applies to these conditions.

Such $\smallo_p(n^{-1/4})$ convergence rates are achievable for many nonparametric regression or machine learning methods under ordinary conditions. We list a few examples below when the range $\mathcal{Z}$ of $Z$ is a Euclidean space. 
For other methods, once their $L^2$-convergence rates are established, the above sufficient condition can often be verified immediately.
In the following examples, it is useful to note that $1/(1+\theta^{k-1}_*)$ is the conditional probability function $\bar{z}_{k-1} \mapsto P_*(A=0 \mid \bar{Z}_{k-1}=\bar{z}_{k-1},A \in \mathcal{S}'_k)$.
\begin{compactitem}
    \item Series estimators \citepsupp{Chen2007}: Suppose that $\ell^{k-1}_*$ and $1/(1+\theta^{k-1}_*)$ are estimated using tensor product polynomial series, trigonometric series, or polynomial splines of sufficiently high degrees with approximately equally spaced knots. 
    Suppose that the support of $Z_{k-1} \mid A \in \mathcal{S}'_k$ is a Cartesian product of compact intervals. Let $\Lambda^{k-1}_p$ denote the $p$-smooth H\"older class on this support. Let $\kappa$ denote the degree of the polynomial series or the trigonometric series, or the number of knots in the spline series. 
    
    If 
    (i) the truth $1/(1+\theta^{k-1}_*)$ lies in $\Lambda^{k-1}_p$, 
    (ii) the class of oracle regression functions $\{\mathcal{R}^{k-1}_v(g): g \in \mathcal{G}_k\}$ is contained in $\Lambda^{k-1}_p$, 
    (iii) $g(\bar{Z}_k) - \mathcal{R}^{k-1}_v(g)(\bar{Z}_{k-1})$ has bounded variance uniformly over $g \in \mathcal{G}_k$, 
    (iv) the density of $\bar{Z}_{k-1} \mid A \in \mathcal{S}'_k$ is bounded away from zero and infinity, and 
    (v) $p > \dim(Z_{k-1})$ for polynomial series, $p > \dim(Z_{k-1})/2$ for trigonometric series, or $p > \dim(Z_{k-1})/2$ and the order of the spline at least $\lceil p \rceil$ (the smallest integer greater than or equal to $p$) for polynomial splines, then the $\smallo_p(n^{-1/4})$ rate is achieved with $\kappa$ growing at rate $n^{1/(2p+d)}$ as the sample size $n \to \infty$ \citep[see e.g., Proposition 3.6 in][]{Chen2007}.

    \item Highly adaptive lasso \citepsupp{Benkeser2016,VanderLaan2017}: Suppose that $\ell^{k-1}_*$ and $1/(1+\theta^{k-1}_*)$ are estimated with highly adaptive lasso using the squared-error loss or the logistic loss. If $\mathcal{G}_k$ is chosen such that the class of oracle regression functions $\{\mathcal{R}^{k-1}_v(g): g \in \mathcal{G}_k\}$ have bounded sectional variation norm and $\log(\theta^{k-1}_*)$ also has bounded sectional variation norm, then, under some additional technical conditions, the $\smallo_p(n^{-1/4})$ rate is achieved. These conditions allow for discontinuity and are often mild even for large $\dim(\bar{Z}_{k-1})$.

    \item Neural networks: Regression using a single hidden layer feedforward neural network with possibly non-sigmoid activation functions also achieves the $\smallo_p(n^{-1/4})$ rate under mild conditions \citepsupp{Chen1999}. Similarly to the highly adaptive lasso, the smoothness condition required does not depend on $\dim(\bar{Z}_{k-1})$.
    Further, \cite{bauer2019deep} showed that certain deep networks can achieve such rates under smooth generalized hierarchical interaction models, where the rate depends only on the 
    maximal number of linear combinations of variables 
    in the first layer and 
    maximal number of interacting variables 
    in each layer
    in the function class; see also \cite{kohler2022estimation, kohler2021rate,kohler2022analysis}.

    \item Boosting: \citetsupp{Luo2016} showed that $L_2$-boosting can achieve the $\smallo_p(n^{-1/4})$ rate for sparse linear models. Boosting has also been shown to perform well in some settings \citep[e.g.,][]{Blanchard2004,Buhlmann2003,Buhlmann2007}.

    \item Regression trees and random forests: A class of regression trees and random forest studied in \citetsupp{Wager2015} can achieve the $\smallo_p(n^{-1/4})$ rate under certain conditions.

    \item Ensemble learning with Super Learner \citepsupp{VanderLaan2007}: Suppose that $\ell^{k-1}_*$ and $1/(1+\theta^{k-1}_*)$ are estimated using Super Learner with a fixed or slowly growing number of candidate learners in the library. If one candidate learner achieves the $\smallo_p(n^{-1/4})$ rate, then the ensemble learner also achieves the $\smallo_p(n^{-1/4})$ rate under mild conditions.
\end{compactitem}

In our empirical simulation studies, we have found that Super Learner containing gradient boosting, in particular XGBoost \citepsupp{Chen2016}, with various tuning parameters as candidate learners appear to yield sufficient convergence rates within a reasonable computational time.
We have thus used this approach in our simulations and data analyses.

\section{Additional simulations} \label{sec: more sim}

\subsection{Simulations for risk estimation under full-data covariate shift} \label{sec: cov shift sim}

We run a simulation study similar to that from Section~\ref{sec: X con shift sim} to examine Corollaries~\ref{corollary: cov shift eff gain}--\ref{corollary: cov shift DR consistent}. 
Due to the similarity between simulations, 
we mainly emphasize the differences. 
Focusing on estimating the MSE, we consider five scenarios that are similar to those in Section~\ref{sec: X con shift sim}.
We first generate $X$ from the same trivariate normal distribution as in Section~\ref{sec: X con shift sim} and then generate $A$ using the same procedure as in Scenario~\ref{cov shift scenario: invalid} in Section~\ref{sec: X con shift sim}. 
Therefore, the proportion $\rho_*$ of the target population data is roughly 10\%.
For Scenarios~\ref{cov shift scenario: large gain}--\ref{cov shift scenario: no gain}, where the dataset shift condition \ref{DScondition: cov shift} holds, we generate the outcome $Y$ independently from $A$ given $X$ using the same procedure as the procedure for the target population in Section~\ref{sec: X con shift sim}. For Scenario~\ref{cov shift scenario: invalid}, where Condition~\ref{DScondition: cov shift} does not hold, we include dependence of $Y$ on $A$ given $X$ and generate $Y$ as $Y \mid X=x,A=a \sim \mathrm{Bernoulli}(\tilde{\mu}_*(x,a))$ where
$$\tilde{\mu}_*(x,a) = \begin{cases}
    \mu_*(x) & \text{if } a=0 \\
    1 + 0.5 x_1 - 1.5 x_2 + x_1 x_3 - 1.3 x_2 x_3^2 & \text{if } a=1
\end{cases}$$
for $\mu_*$ in \eqref{eq: sim mu0}.
The given fixed predictors $f$ are identical to those in Section~\ref{sec: X con shift sim}.

In terms of efficiency gains from using our proposed estimator $\hat{r}_{\covshift}$, 
we expect efficiency gains to decrease from scenarios \ref{cov shift scenario: large gain} to \ref{cov shift scenario: no gain}, 
due to the difference between Corollaries~\ref{corollary: cov shift eff gain} and \ref{corollary: X con shift eff gain}.

We investigate the performance of the following four sequences of estimators:
\begin{compactitem}
    \item \texttt{np}: the nonparametric estimator $\hat{r}_{\nonparametric}$ in \eqref{eq: np estimator};
    \item \texttt{covshift}: proposed estimator $\hat{r}_{\covshift}$ in Line~5 of Alg.~\ref{alg: cov shift estimator} with flexible and consistent estimators of nuisance functions $g_*$ and $\mathcal{L}_*$;
    \item \texttt{covshift.mis.L}: $\hat{r}_{\covshift}$ with an inconsistent estimator of $\mathcal{L}_*$;
    \item \texttt{covshift.mis.g}: $\hat{r}_{\covshift}$ with an inconsistent estimator of $g_*$.
\end{compactitem}
Nuisance estimators are constructed as in Section~\ref{sec: X con shift sim}.

\begin{figure}
    \centering
    \includegraphics[scale=0.8]{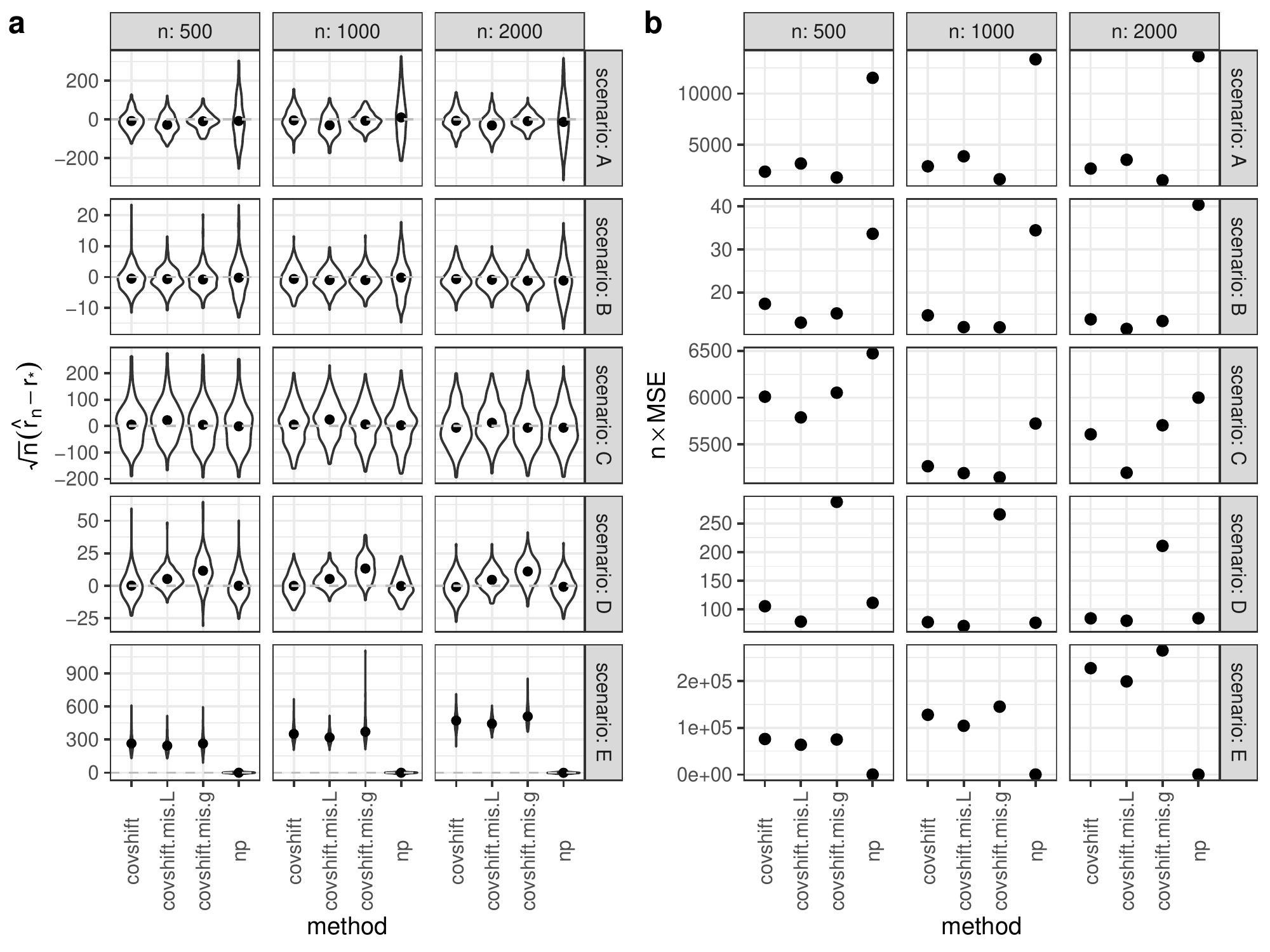}
    \caption{(a) Sampling distribution of the scaled difference between the estimators of MSE and the true MSE in the four scenarios under covariate shift. The point stands for the empirical average in Monte Carlo simulations. (b) Monte Carlo estimate of the scaled mean squared error of the estimators.}
    \label{fig: covshift}
\end{figure}

Figure~\ref{fig: covshift} presents the sampling distribution of scaled difference between the four estimators of MSE and the true MSE in the five scenarios. 
When the dataset shift condition \ref{DScondition: cov shift} holds (Scenarios~\ref{cov shift scenario: large gain}--\ref{cov shift scenario: no gain}), \texttt{np} and \texttt{covshift} both appear close to normal and centered around the truth; in contrast, \texttt{covshift.mis.L} and \texttt{covshift.mis.g} appear more biased in Scenario~\ref{cov shift scenario: no gain} where the nuisance function estimators are substantially far from the truth, but they appear to be consistent in all scenarios. These two estimators, though generally not root-$n$ consistent and asymptotically normal, appear similar to \texttt{covshift} in Scenarios~\ref{cov shift scenario: large gain}--\ref{cov shift scenario: small gain}. Thus, the asymptotic linearity of our proposed estimator might be robust against mild inconsistent estimation of one nuisance function. The variance of \texttt{covshift} is much smaller than that of \texttt{np} in both Scenarios~\ref{cov shift scenario: large gain} and \ref{cov shift scenario: medium gain}, indicating a large efficiency gain; in Scenarios~\ref{cov shift scenario: small gain} and \ref{cov shift scenario: no gain}, the variance of these two estimators is comparable. These results are consistent with Corollary~\ref{corollary: cov shift eff gain}. When Condition~\ref{DScondition: cov shift} does not hold (Scenario~\ref{cov shift scenario: invalid}), our proposed estimator is substantially biased, indicating that $\hat{r}_{\covshift}$ is not robust against failure of covariate shift condition \ref{DScondition: cov shift}.

\subsection{Simulations for model comparison under dataset shift} \label{sec: compare model sim}

Using notations for Corollary~\ref{corollary: model comparison}, in the simulations in this subsection, the estimand is the risk difference $r_*^{(1)} - r_*^{(2)}$ between two predictors $f^{(1)}$ and $f^{(2)}$.

\subsubsection{Concept shift in the features} \label{sec: X con shift compare}

The data is generated as in Scenarios~\ref{cov shift scenario: medium gain} and \ref{cov shift scenario: small gain} in Section~\ref{sec: X con shift sim}. The two predictors are from Scenarios~\ref{cov shift scenario: medium gain} ($f^{(1)}$) and \ref{cov shift scenario: small gain} ($f^{(2)}$), respectively. 
All other simulation settings are identical to Section~\ref{sec: X con shift sim}.
Simulation results are shown in Figure~\ref{fig: compare X con shift} and are similar to those in Section~\ref{sec: X con shift sim}. Our proposed estimator appears to be roo-$n$-consistent and asymptotically normal, and to have a smaller mean squared error than the nonparametric estimator.

\begin{figure}
    \centering
    \includegraphics[scale=0.8]{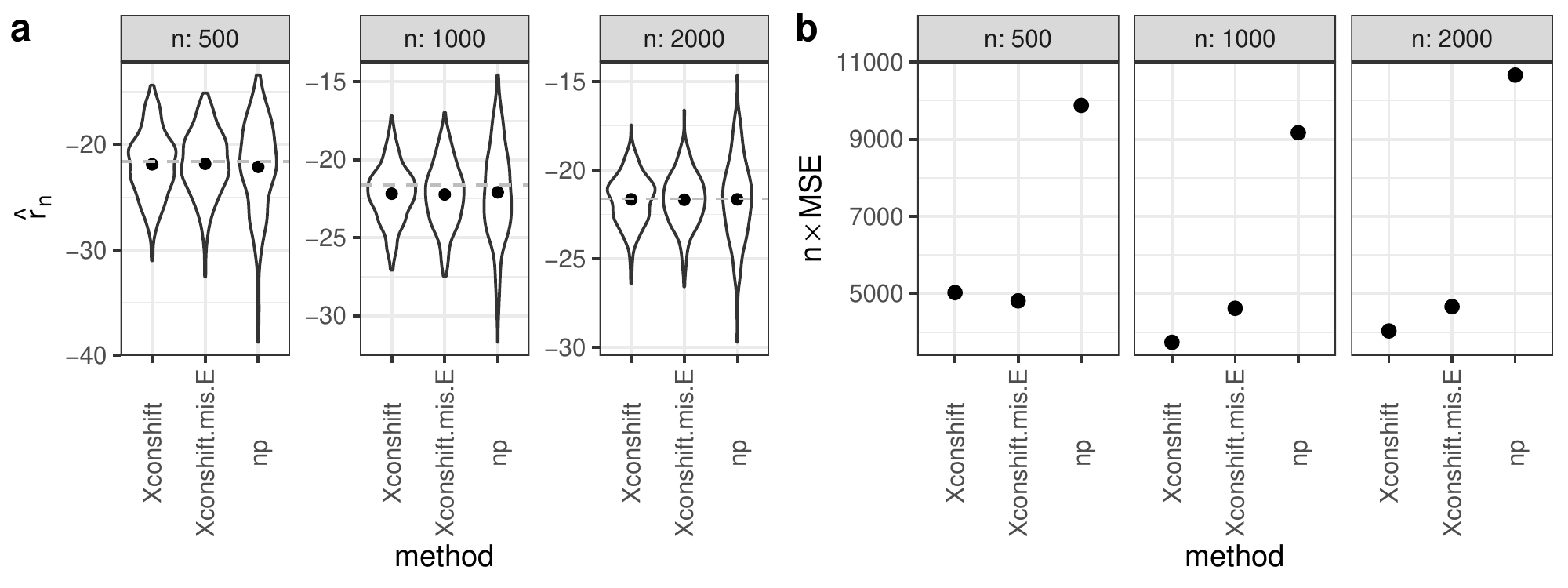}
    \caption{Plots similar to Figure~\ref{fig: X con shift} for estimating the risk difference between two predictors under concept shift.}
    \label{fig: compare X con shift}
\end{figure}

\subsubsection{Full-data covariate shift} \label{sec: cov shift compare model}

The data is generated as in Scenarios~\ref{cov shift scenario: medium gain} and \ref{cov shift scenario: small gain} in Section~\ref{sec: cov shift sim}. The two predictors are from Scenarios~\ref{cov shift scenario: medium gain} ($f^{(1)}$) and \ref{cov shift scenario: small gain} ($f^{(2)}$), respectively.
All other simulation settings are identical to Section~\ref{sec: cov shift sim}.
Simulation results are shown in Figure~\ref{fig: compare cov shift} and are similar to those in Section~\ref{sec: cov shift}.
When both nuisance function estimators are consistent, our proposed estimator appears asymptotically normal with a smaller mean squared error than the nonparametric estimator.
When one nuisance function estimator is inconsistent, our proposed estimator appears consistent and may still have a smaller mean squared error than the nonparametric estimator.

\begin{figure}
    \centering
    \includegraphics[scale=0.8]{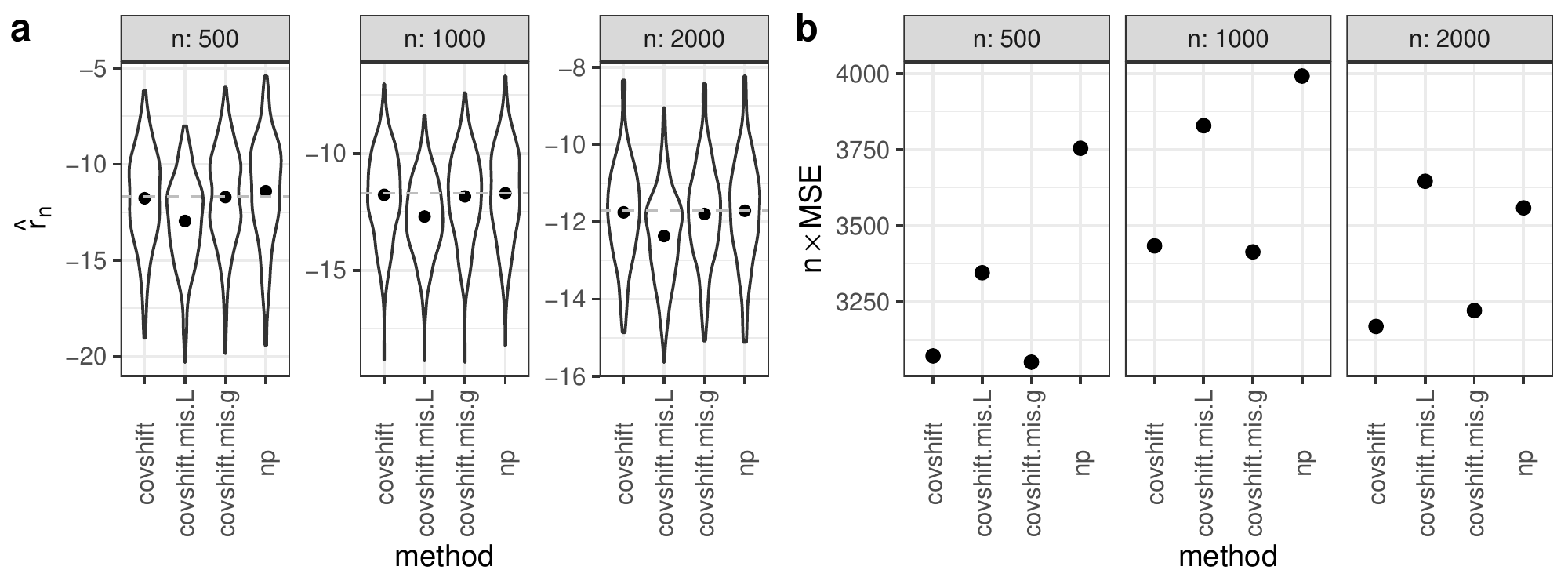}
    \caption{Plots similar to Figure~\ref{fig: covshift} for estimating the risk difference between two predictors under full-data covariate shift.}
    \label{fig: compare cov shift}
\end{figure}

\subsection{Simulations for logistic regression under dataset shift} \label{sec: GLM sim}

In the simulations in this subsection, the estimand is
$$\beta_* = \argmin_{\beta} \expect_{P_*}[- Y \log \{ \expit([1,X]^\top \beta) \} - (1-Y) \log \{ 1- \expit([1,X]^\top \beta) \} \mid A=0],$$
which is well-defined even if the logistic-linear model is misspecified.
Equivalently, $\beta_*$ is the solution in $\beta$ to
$$\expect_{P_*}[(Y-\expit([1,X]^\top \beta)) [1,X]^\top \mid A=0] = 0.$$

\subsubsection{Concept shift in the features} \label{sec: X con shift GLM sim}

\noindent\textbf{Data generating process:} Covariate $X$ and population index $A$ are generated as in Scenarios~\ref{cov shift scenario: large gain}--\ref{cov shift scenario: no gain} in Section~\ref{sec: X con shift sim}. We generate a binary label $Y$ from
$$Y \mid X=(X_1,X_2,X_3) \sim \text{Bern}(\expit\{-3+X_1^2+X_2^2+X_3^2+0.4 X_1 X_3-0.5 X_2 X_3+\sin(X_1+X_3)\}).$$

\noindent\textbf{Methods compared}: All estimators solve an empirical estimating equation as described in Remark~\ref{rmk: risk interpretation} and are computed using the Newton-Raphson method.
\begin{compactitem}
    \item \texttt{np}: risk minimizer based on the nonparametric risk estimator $\hat{r}_\nonparametric$, that is, logistic regression using only target population data;
    \item \texttt{Xconshift\textunderscore SL}: risk minimizer based on $\hat{r}_\xconshift$ with $\hat{\mathcal{E}}^{-v}$ obtained from Super Learner as in Section~\ref{sec: X con shift sim};
    \item \texttt{Xconshift\textunderscore np}: risk minimizer based on $\hat{r}_\xconshift$ with $\hat{\mathcal{E}}^{-v}$ obtained from logistic regression using only target population data, similarly to \texttt{np};
    \item \texttt{Xconshift\textunderscore updating}: risk minimizer based on $\hat{r}_\xconshift$ with $\hat{\mathcal{E}}^{-v}$ based on the logistic regression estimator in the previous Newton-Raphson iteration.
\end{compactitem}
\texttt{Xconshift\textunderscore np} and \texttt{Xconshift\textunderscore updating} are simple sensible alternatives to \texttt{Xconshift\textunderscore SL} with less computational cost.
Although $\hat{r}_\xconshift$, as a function of $\beta$, converges to a convex function, this estimated function itself might not be convex. So there can be convergence issues with the latter three methods.
We run 200 experiments for each sample size $n=500, 1000, 2000$.

\noindent\textbf{Results}: Fig.~\ref{fig: GLM x con shift} presents the sampling distribution. Experiments with convergence issues (75 runs for \texttt{Xconshift\textunderscore np} and one run for \texttt{Xconshift\textunderscore updating}) or a magnitude greater than 5 (9 runs for \texttt{Xconshift\textunderscore np}) are excluded from Fig.~\ref{fig: GLM x con shift}. Similar to the findings in Section~\ref{sec: X con shift sim}, \texttt{Xconshift\textunderscore SL} appears consistent and approximately normally distributed with a smaller asymptotic variance than \texttt{np}. The alternatives \texttt{Xconshift\textunderscore np} and \texttt{Xconshift\textunderscore updating} of \texttt{Xconshift\textunderscore SL} do not appear superior compared to \texttt{np}. Estimating $\mathcal{E}_*$ well appears important to achieve efficiency gain for estimating equation estimators for concept shift in the features.

\begin{figure}
    \centering
    \includegraphics[scale=0.8]{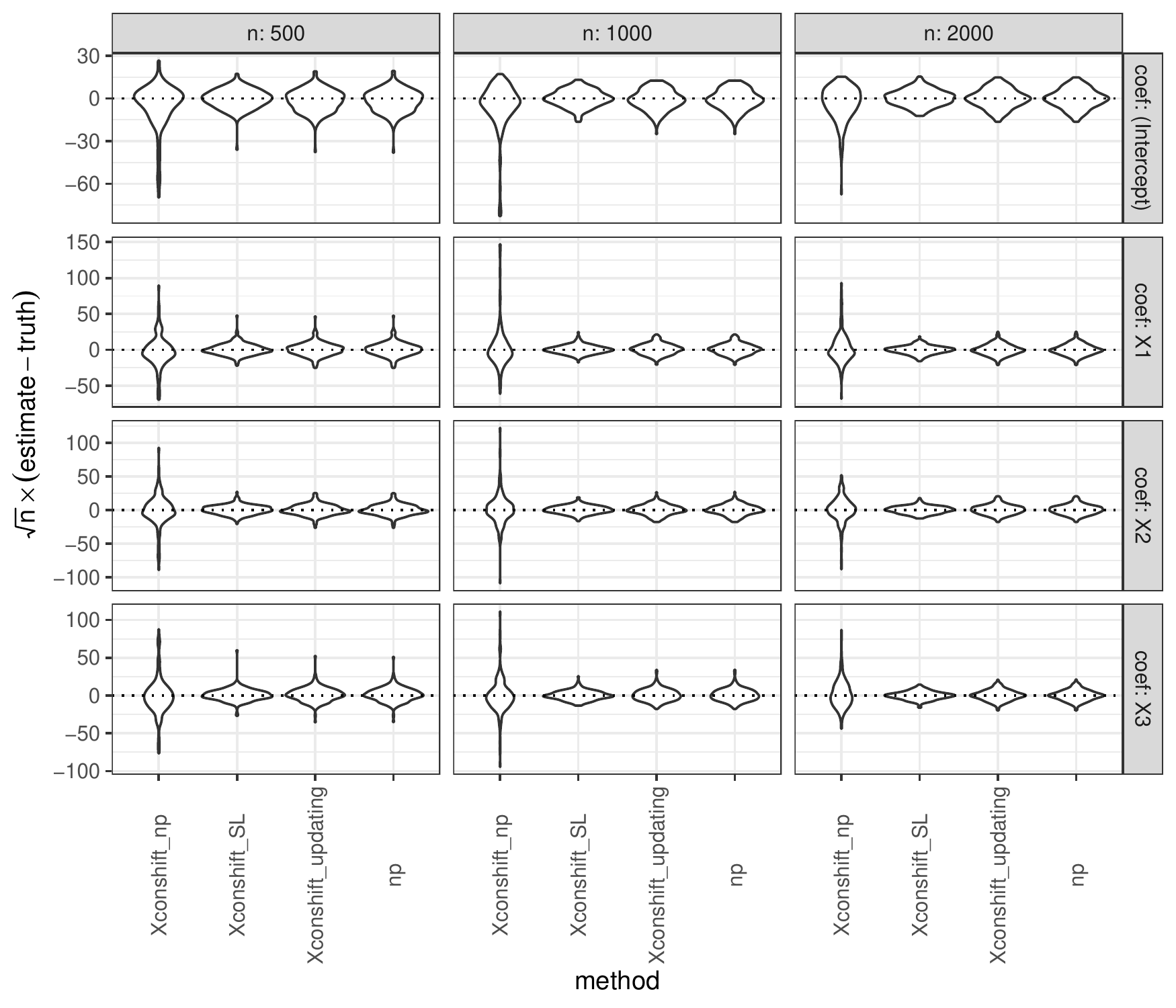}
    \caption{Sampling distribution of the scaled difference between the estimated coefficient and the truth for each of the four coefficients in $\beta_*$, under concept shift in the features.}
    \label{fig: GLM x con shift}
\end{figure}

\subsubsection{Full-data covariate shift} \label{sec: cov shift GLM sim}

\noindent\textbf{Data generating process:} Covariate $X$ and population index $A$ are generated as in Section~\ref{sec: cov shift sim}. We generate a binary label $Y$ as in Section~\ref{sec: X con shift GLM sim}.

\noindent\textbf{Methods compared}: Similar to those in Section~\ref{sec: X con shift GLM sim}:
\begin{compactitem}
    \item \texttt{np}: same as in Section~\ref{sec: X con shift GLM sim};
    \item \textbf{naive}: logistic regression using both source and target data (without weighting);
    \item \texttt{covshift\textunderscore SL}: risk minimizer based on $\hat{r}_\covshift$ with nuisance functions estimated via Super Learner as in Section~\ref{sec: X con shift sim};
    \item \texttt{covshift\textunderscore naive}: same as \texttt{covshift\textunderscore SL} except that $\hat{\mathcal{L}}^{-v}$ is obtained via logistic regression similarly to \texttt{naive};
    \item \texttt{covshift\textunderscore updating}: same as \texttt{covshift\textunderscore SL} except that $\hat{\mathcal{L}}^{-v}$ is based on the logistic regression estimator in the previous Newton-Raphson iteration.
\end{compactitem}
\texttt{naive} is a na\"ive approach to incorporate source data under covariate shift.
\texttt{covshift\textunderscore naive} and \texttt{covshift\textunderscore updating} are simple sensible alternatives to \texttt{covshift\textunderscore SL} with less computational cost.
Convergence issues might also occur with the latter three methods, similarly to Section~\ref{sec: X con shift GLM sim}.
We still run 200 experiments for each sample size $n=500, 1000, 2000$.

\noindent\textbf{Results}: Fig.~\ref{fig: GLM cov shift} presents the sampling distribution. Experiments with convergence issues (7 runs for \texttt{Xconshift\textunderscore updating}) are excluded from Fig.~\ref{fig: GLM x con shift}. Similar to the findings in Section~\ref{sec: cov shift sim}, \texttt{covshift\textunderscore SL} is approximately normal with a smaller variance than \texttt{np}. Due to the misspecification of the logistic model, \texttt{naive} is substantially biased. Although the alternatives \texttt{covshift\textunderscore naive} and \texttt{covshift\textunderscore updating} of \texttt{covshift\textunderscore SL} do not appear asymptotically unbiased, they still appear consistent.

\begin{figure}
    \centering
    \includegraphics[scale=0.8]{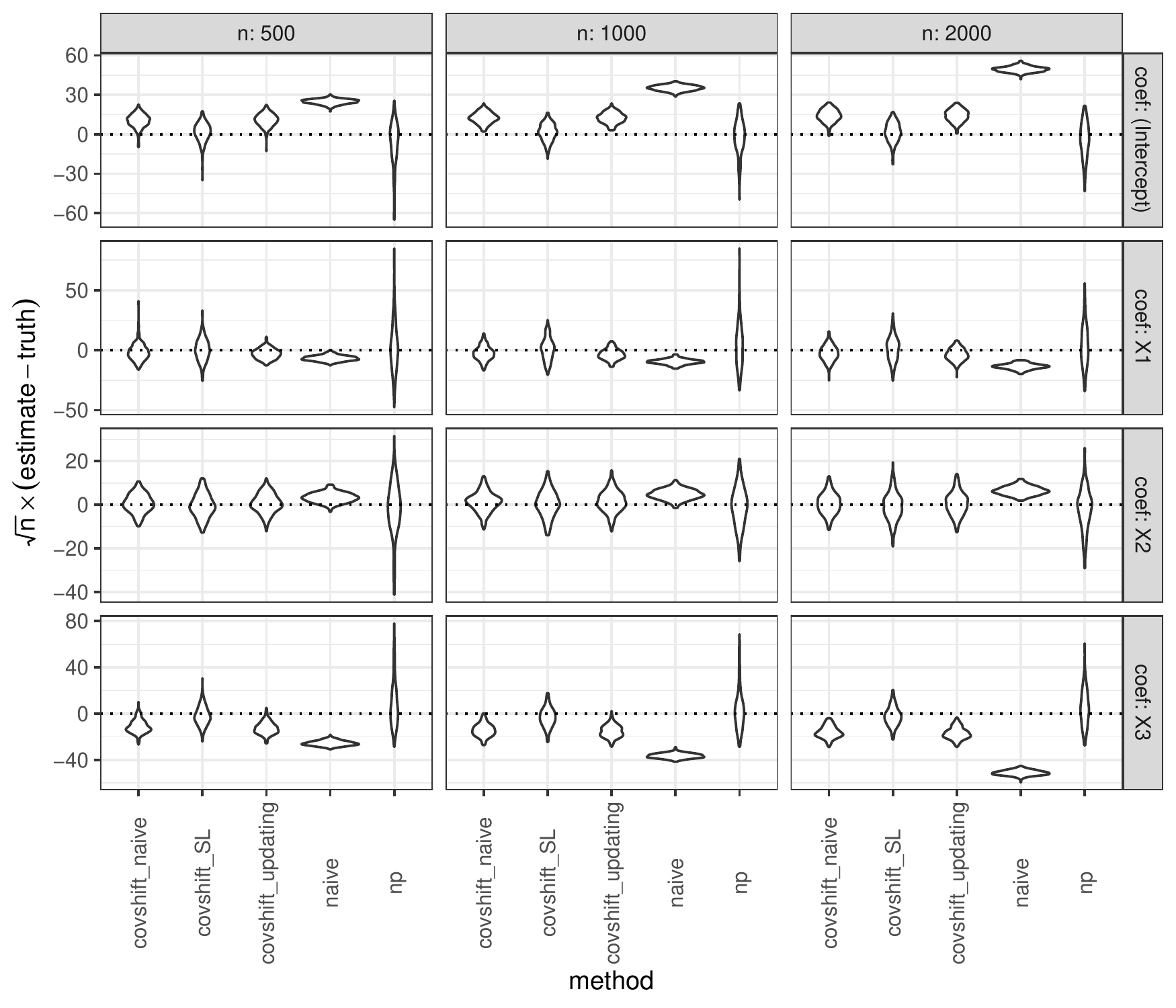}
    \caption{Sampling distribution of the scaled difference between the estimated coefficient and the truth for each of the four coefficients in $\beta_*$, under covariate shift.}
    \label{fig: GLM cov shift}
\end{figure}

\section{Analysis of HIV risk prediction data in South Africa} \label{sec: data analysis}

\subsection{Analysis under each of Conditions~\ref{DScondition: X con shift}--\ref{DScondition: label shift}} \label{sec: data analysis all}

We illustrate our methods by evaluating the performance of HIV prediction models using a dataset from a South African cohort study \citepsupp{Tanser2013}. The study was a large population-based prospective cohort study in KwaZulu-Natal, South Africa. In this study, 16,667 individuals that were HIV-uninfected at baseline were followed up from 2004 to 2011.

In this analysis, we used HIV seroconversion as the outcome $Y$ and used the following covariates: number partners in the past 12 months, current marital status, wealth quintile, age and sex, community antiretroviral therapy (ART) coverage, and community HIV relevance. All covariates are binned and discrete. We set individuals from peri-urban communities with community ART coverage below 15\% as the target population and those from urban and rural communities as the source population. There are 1,418 individuals from the target population and 12,385 from the source population. We used the last observed value for time-varying covariates, and missing data was treated as a separate category \citepsupp{Groenwold2012}.

We randomly selected half of the source population data as the held-out dataset to train an HIV risk prediction model $f$. We then used the other half of the source population data and randomly selected 50 individuals from the target population to estimate the inaccuracy of $f$ in the target population (recall Example~\ref{ex: regression}).
In other words, we had a large source population data and a small target population data.
We applied the nonparametric estimator $\hat{r}_{\nonparametric}$ and our proposed efficient estimator under each of the four most common dataset shift conditions \ref{DScondition: X con shift}--\ref{DScondition: label shift} connecting these two populations.
We used the rest of the target population data as a validation dataset to obtain a good estimate of the risk and evaluate the performance of risk estimators.

In all regressions, to obtain potentially improved nuisance regression function estimators and make sufficient conditions for our estimators' efficiency more plausible, we further added single hidden layer feedforward neural networks with various numbers of hidden nodes into the Super Learner library that we used for the simulations.
With an appropriate number of hidden nodes, such neural networks have been shown to achieve an $\smallo_p(n^{-1/4})$ convergence rate under mild smoothness conditions \citepsupp{Chen1999}.

Since we do not expect the individuals from the two populations to be random draws from a common population, the concept shift assumption (either \ref{DScondition: X con shift} or \ref{DScondition: Y con shift}) might fail and their corresponding efficient estimators might be biased. 
Since covariates are likely to affect the outcome causally and the other causal direction is implausible due to time ordering, Assumption~\ref{DScondition: label shift} might also fail. The covariates may have reasonable predictive power, and thus Assumption~\ref{DScondition: cov shift} may approximately hold.
In the analysis, we also tested whether the efficient estimator under each dataset shift condition is consistent for the risk of interest using the specification test implied by Corollary~\ref{corollary: test}.

\begin{table}
    \centering
    \caption{Risk estimates under at most one dataset shift condition on the source population. The risk estimate from the validation dataset is 0.24 (95\% confidence interval: 0.22--0.26). All confidence intervals (CIs) are Wald intervals.}
    \label{tab: data analysis result}
    \begin{tabular}{l|r|r|r|r}
        Dataset Shift Condition & Estimate & S.E. & 95\% CI & P-value \\
        \hline\hline
        None & 0.24 & 0.060 & (0.12, 0.36) & --- \\
        Concept shift in the features \ref{DScondition: X con shift} & 0.26 & 0.057 & (0.15, 0.38) & 0.29 \\
        Concept shift in the labels \ref{DScondition: Y con shift} & 0.10 & 0.010 & (0.08, 0.12) & 0.02 \\
        Full-data covariate shift \ref{DScondition: cov shift} & 0.19 & 0.026 & (0.14, 0.25) & 0.41 \\
        Full-data label shift \ref{DScondition: label shift} & 0.23 & 0.059 & (0.11, 0.34) & 0.42
    \end{tabular}
\end{table}

The risk estimates and the p-values for testing the four dataset shift conditions are presented in Table~\ref{tab: data analysis result}. As expected, our proposed estimator under concept shift in the labels (\ref{DScondition: Y con shift}) is biased and this dataset shift condition can be rejected at level 0.05, indicating evidence that concept shift in the labels does not hold. 
Concept shift in the features (\ref{DScondition: X con shift}) is not rejected, and our proposed estimator under this condition leads to a similar conclusion to the nonparametric estimator with a confidence interval that is around 5\% shorter.
The label shift assumption (\ref{DScondition: label shift}) is not rejected either, and our proposed estimator under this condition leads to an almost identical conclusion compared to the nonparametric estimator. 
Covariate shift (\ref{DScondition: cov shift}) is also not rejected, and our proposed efficient estimator has a much smaller standard error, leading to a confidence interval that is more than 50\% shorter.
We emphasize that our test may be underpowered, so failure to reject these dataset shift conditions should not be interpreted as strong evidence that these conditions may be plausible.
Since covariate shift is the most plausible assumption among the four \textit{a priori}---although, in this particular case, all point estimates appear consistent except for concept shift in the labels, for which we find evidence of bias---we conclude that using our proposed efficient estimator under a conceivable dataset shift condition may lead to substantial efficiency gains and thus more informative estimates of risk, namely with substantially smaller variance.

\subsection{Semi-synthetic data analysis} \label{sec: data analysis2}

For each of Conditions~\ref{DScondition: X con shift}--\ref{DScondition: label shift}, we also conducted a semi-synthetic data analysis. Using the original data $D$ analyzed in Section~\ref{sec: data analysis}, we generated a semi-synthetic dataset satisfying the dataset shift condition under consideration. We then ran a similar analysis as in Section~\ref{sec: data analysis}. We next describe the semi-synthetic data-generating procedures in more detail and report the analysis results. From our analyses, we conclude that, under a plausible dataset shift condition, our proposed efficient risk estimator that fully exploits the dataset shift condition may lead to substantial efficiency gains and shorter CIs.

\subsubsection{Concept shift in the features (\ref{DScondition: X con shift})} \label{sec: data analysis2 X conshift}

To generate a semi-synthetic data satisfying Condition~\ref{DScondition: X con shift}, we first randomly generated the synthetic population indicator $A$ following the empirical distribution observed in the original data $D$ and then set the outcome $Y$ to be missing for individuals from the synthetic source population ($A=1$). We left the covariate $X$ in both synthetic populations and the outcome $Y$ in the synthetic target population unchanged. In the analysis, we used 50 individuals from the synthetic target population to train a prediction model $f$. We then used all individuals from the synthetic source population and 50 individuals from the synthetic target population to estimate the risk of $f$. We used the rest of the individuals from the synthetic target population as the validation dataset to evaluate the performance of risk estimators.

The risk estimate from the validation dataset is 0.10 with a 95\% confidence interval (CI) of 0.09--0.12. The estimate under no dataset shift condition is 0.06 (S.E.: 0.034, 95\% CI: 0--0.13); the estimate from our proposed estimator is 0.06 (S.E.: 0.033, 95\% CI: 0--0.12). Here, the CIs are truncated to fall in the unit interval $[0,1]$.
The test of Condition~\ref{DScondition: X con shift} yields a p-value $>0.99$, correctly indicating that Condition~\ref{DScondition: X con shift} may be plausible in this semi-synthetic dataset.
The CI based on our proposed estimator is almost identical to that based on the nonparametric estimator.
The small efficiency gain might be attributed to the limited sample size to learn the conditional mean risk function $\mathcal{E}_*$, which is a function of the moderately high-dimensional covariate $X$ and thus difficult to estimate accurately; additional information about the distribution of $X$ cannot improve estimation of $\mathcal{E}_*$.

\subsubsection{Concept shift in the labels (\ref{DScondition: Y con shift})} \label{sec: data analysis2 Y conshift}

Our data-generating procedure and analysis are similar to those in Section~\ref{sec: data analysis2 X conshift}. The only difference is that we set covariate $X$, rather than outcome $Y$, to be missing in the synthetic source population.

The risk estimate from the validation dataset is 0.10 (95\% CI: 0.09--0.12). The nonparametric estimate is 0.12 (S.E.: 0.046, 95\% CI: 0.03--0.21), while the estimate from our proposed estimator is 0.14 (S.E.: 0.035, 95\% CI: 0.07--0.21). The p-value of from the test of Condition~\ref{DScondition: Y con shift} is 0.53. The CI based on our proposed estimator is more than 20\% shorter than that based on the nonparametric estimator.

\subsubsection{Full-data covariate shift (\ref{DScondition: cov shift})} \label{sec: data analysis2 cov shift}

To generate a semi-synthetic dataset satisfying Condition~\ref{DScondition: cov shift}, we first randomly split the original dataset $D$ into two halves $D_1$ and $D_2$. We used $D_1$ to obtain estimates $\hat{g}$ and $\hat{\mathcal{L}}$ of nuisance functions $g_*: x \mapsto P_*(A=0 \mid X=x)$ and $\mathcal{L}_*: x \mapsto \expect_{P_*}[Y \mid X=x] = P_*(Y=1 \mid X=x)$, respectively. With the observed covariates $X$ in $D_2$, we generated synthetic population indicator $A$ and synthetic outcome $Y$ independently given $X$ based on $\hat{g}$ and $\hat{\mathcal{L}}$, respectively. We used this semi-synthetic dataset in the subsequent analysis. We used half of the individuals from the synthetic source population to train a prediction model $f$. We then used the other half of the individuals from the synthetic source population and 50 individuals from the synthetic target population to estimate the risk of $f$. We used the rest of the individuals from the synthetic target population as the validation dataset.

In the validation dataset, the risk is estimated to be 0.22 (95\% CI: 0.19--0.26). The nonparametric risk estimator yields a point estimate of 0.20 with S.E. 0.057 and 95\% CI 0.09--0.31; our proposed estimator yields a point estimate of 0.23 with S.E. 0.030 and 95\% CI 0.17--0.29. Condition~\ref{DScondition: cov shift} is not rejected (p-value=0.51). The CI based on our proposed estimator is about 50\% shorter than that based on the nonparametric estimator.

\subsubsection{Full-data label shift (\ref{DScondition: label shift})} \label{sec: data analysis2 label shift}

The data-generating procedure is similar to that in Section~\ref{sec: data analysis2 cov shift}. A major difference is that, when generating synthetic covariate $X$ and synthetic population indicator $A$ in $D_2$, we sample $X$ and $A$ independently given $Y$ from the empirical distribution of $X \mid Y$ and $A \mid Y$ observed in $D_1$. We could use the empirical distribution because $Y$ is binary.
Super Learner uses cross-validation to select the weight to ensemble learners in the library, and cross-validation can be problematic with ties in the data \citepsupp{Bickel1997}.
To alleviate this issue, for each covariate of each individual in the semi-synthetic data, we independently replaced the observed value with a uniform draw from the range of this covariate with 10\% probability.
Such random corruption for discrete variables has been extensively studied \citep[e.g.,][etc.]{Ma2018,Angluin1988,Aslam1996}.

In the validation dataset, the risk is estimated to be 0.27 (95\% CI: 0.24--0.31). The estimates from the nonparametric estimator and our proposed estimator are 0.20 (S.E.: 0.057, 95\% CI: 0.09--0.31) and 0.21 (S.E.: 0.056, 95\% CI: 0.11--0.32), respectively. Condition~\ref{DScondition: label shift} is not rejected (p-value=0.07). The length of the CI based on our proposed estimator is almost identical to than that based on the nonparametric estimator. The similar performance of these two estimators in this analysis may be explained by the fact that, under label shift (\ref{DScondition: label shift}), a main contribution from the source population is additional data to estimate the conditional mean risk function $\mathcal{L}_{*,Y}$. Since $Y$ is binary, $\mathcal{L}_{*,Y}$ can be estimated fairly accurately with the target population data alone, and the improvement in the risk estimator due to the improved estimation of $\mathcal{L}_{*,Y}$ is limited.

\subsubsection{Model comparison under covariate shift} \label{sec: data analysis compare}

We demonstrate model comparison described in Corollary~\ref{corollary: model comparison} using our proposed methods in the HIV risk prediction data.
We focus on covariate shift because our analysis and prior knowledge suggest that this dataset shift condition may be plausible while all other three may not.
The first prediction model $f^{(1)}$ is the same as in Section~\ref{sec: data analysis all}. The second prediction model $f^{(2)}$ is logistic regression fitted using the same training data, all from the source population, as the first prediction model.
All other settings are identical to Section~\ref{sec: data analysis all}.
We compare our proposed risk difference estimator for covariate shift and the nonparametric estimator.

In the held-out validation dataset, the risk difference $r_*^{(1)}-r_*^{(2)}$ is estimated to be 0.0022 (95\% CI: -0.0016, 0.0060).
The predictions from the two models coincide on all 50 target population data points in the data used to estimate the risk difference.
Therefore, the nonparametric estimate is zero with a standard error of zero. Inference based on this nonparametric estimate is likely to be anti-conservative, because the two models are similar but not identical, and the size of the target population data is too small to distinguish them.
In contrast, the estimate based on our proposed method is -0.0049 (95\% CI: -0.0123, 0.0025). Our proposed method leads to a consistent conclusion with the validation dataset and outputs an appropriate non-degenerate CI, indicating that the inaccuracy of the two prediction models is very similar but not necessarily exactly identical.

\section{Other dataset shift conditions} \label{sec: other DS conditions}

In addition to the four dataset shift conditions \ref{DScondition: X con shift}--\ref{DScondition: label shift}, 
we study a few other classes of conditions from prior work 
that are model-agnostic and lead to semiparametric models. 
For all these conditions, data $Z$ consists of covariate/feature $X$ and outcome/label $Y$; only one source population is present, and so the population indicator $A$ is binary.
In each following subsection, we first present a condition and then characterize semiparametric efficiency bounds under this condition.
Theorem~7.2 in Part~III of \citetsupp{Bolthausen2002} ensures that one-step estimators based on these efficiency bounds achieve efficiency under conditions on the nuisance estimators. We leave more thorough study of these conditions for future work.

\subsection{Posterior drift conditions}

\cite{Scott2019} introduced two similar conditions for a binary label $Y \in \{0,1\}$. The more general one is the following \emph{covariate shift posterior drift}.
\begin{DScondition}[Covariate Shift Posterior Drift] \label{DScondition: cov shift posterior drift}
    There exists an unknown differentiable function $\tilde{\phi}_*$ such that for all $x$,
    \begin{align}
        & P_*(Y=1 \mid X=x,A=1) = \tilde{\phi}_*\left( P_*(Y=1 \mid X=x,A=0) \right), \label{eq: cov shift posterior drift} \\
        & \left. \frac{\intd \tilde{\phi}_*(\theta)}{\intd \theta} \right|_{\theta = P_*(Y=1 \mid X=x,A=0)} \neq 0 \nonumber
    \end{align}
\end{DScondition}

As usual, in the above notations, the subscript $*$ stands for the true data-generating distribution $P_*$.
We have slightly modified the original condition in \cite{Scott2019}; see Supplement~\ref{sec: other DS conditions2} for more details. 
Condition~\ref{DScondition: cov shift posterior drift} models labeling errors in the training data ($A=1$). 
In general, in the source population data where objects are manually labeled, labeling errors may occur, 
and these errors are often related to the true label in such a way that, given any object with a high probability of having the true label being one, humans are also likely to label it as one. 
As shown in Section~3 of \cite{Scott2019}, Condition~\ref{DScondition: cov shift posterior drift} is implied by label shift (\ref{DScondition: label shift}).

We have the following results on the efficient influence function for estimating $r_*$ under Condition~\ref{DScondition: cov shift posterior drift}, the proof of which can be found in Supplement~\ref{sec: posterior drift proof}.
\begin{theorem} \label{thm: posterior drift EIF}
    Recall $\rho_* := P_*(A=0)$ and $\mathcal{E}_*: x \mapsto \expect_{P_*}[\ell(X,Y) \mid X=x,A=0]$. For any functions $f_1,f_2 \in L^2(P_{*,X})$, define
    \begin{align*}
        \theta_*(x) &:= \logit P_*(A=1 \mid X=x), \\
        F(f_1,f_2) &: o \mapsto a \left\{ y-\expit \circ \phi_* \circ \theta_*(x)  \right\} f_1(x) + (1-a) \left\{ y - \expit \circ \theta_*(x)  \right\} f_2(x), \\
        v(x;a) &:= \Var_{P_*}(Y \mid A=a,X=x) P_*(A=a \mid X = x), \\
        \kappa(x) &:= \frac{v(x;0) \expect[v(X;1) \mid \theta_*(X) = \theta_*(x)]}{v(x;1) \{ \phi' \circ \theta_*(x) \}^2} \frac{\ell(x,1)-\ell(x,0)}{\rho_*}, \\
        \mu(x) &:= \expect[v(X;1) \mid \theta_*(X) = \theta_*(x)] \left\{ 1 + \frac{v(x;0)}{v(x;1) \{ \phi' \circ \theta_*(x) \}^2} \right\}.
    \end{align*}
    Let $\mathscr{A}$ be the linear integral operator such that, for any function $g \in L^2(P_{*,X})$,
    $$\mathscr{A} g (x) = \expect_{P_*} \left[ g(X) \frac{v(X;1)}{\mu(X)} \mid \theta_*(X) = \theta_*(x) \right].$$
    Suppose that the following Fredholm integral equation of the second kind \citep[see, e.g.,][]{kress2014linear} for unknown $\zeta: \mathcal{X} \to \real$ has a unique solution $\zeta_*$:
    \begin{equation}
        \mathscr{A} \zeta + \kappa = \zeta. \label{eq: F.I.E 2nd}
    \end{equation}
    Let $g_1: x \mapsto - \mathscr{A} \zeta_*(x) \cdot \phi_*' \circ \theta_*(x)/\mathscr{A}\mu(x) + \phi_*' \circ \theta_*(x) \cdot \zeta_*/\mu$ and $g_2 := \zeta_*/\mu$.
    Under Condition~\ref{DScondition: cov shift posterior drift}, the efficient influence function for estimating $r_*$ is $D_\covshiftpostdrift: o \mapsto (1-a) \left\{ \mathcal{E}_*(x) - r_* \right\} / \rho_* + F(g_1,g_2)(o)$.
\end{theorem}

The existence and uniqueness of the solution to \eqref{eq: F.I.E 2nd} can be shown using, for example, Theorem~2.14 or Corollary~3.5 of \cite{kress2014linear}.

The other dataset shift condition in \cite{Scott2019}, \emph{posterior drift}, is similar to Condition~\ref{DScondition: cov shift posterior drift}. We present the corresponding results in Supplement~\ref{sec: other DS conditions2}.

\subsection{Location-scale shift conditions}

\cite{Zhang2013} introduced two dataset shift conditions for vector-valued covariates $X \in \real^d$. The more general condition is \emph{location-scale generalized target shift}.

\begin{DScondition}[Location-Scale Generalized Target Shift] \label{DScondition: ls generalized target shift}
    There exists an unknown function $W_*: \mathcal{Y} \to \real^{d \times d}$ whose values are invertible matrices, 
    and an unknown function $b_*: \mathcal{Y} \to \real^d$ such that, 
    for all $y\in \mY$,
    $$W_*(Y) X + b_*(Y)\mid Y=y, A=1 \quad \text{and} \quad X \mid Y=y, A=0$$
    are identically distributed.
\end{DScondition}

Condition~\ref{DScondition: ls generalized target shift} imposes a specific form on the relationship between the two distributions $X \mid Y,A=1$ and $X \mid Y,A=0$, a location-scale shift, and is thus more general than label shift (\ref{DScondition: label shift}).
We have also slightly modified the original condition in \cite{Zhang2013}; see Supplement~\ref{sec: other DS conditions2}.

We next present the results on the efficient influence function for estimating $r_*$ under the above conditions, the proof of which can be found in Supplement~\ref{sec: location-scale shift proof}.
For any differentiable function $f: \real^d \to \real$, let $\partial_i$ denote  partial differentiation with respect to the $i$-th variable, and let $f'$ or $\partial f(x)/\partial x$ denote the gradient of $f$. 
Let $v_i$ denote the $i$-the entry of a vector $v$, $I_{d\times d}$ be the $\real^{d \times d}$-identity matrix, 
and $1_d$ be the $\real^d$-vector of ones.

\begin{theorem} \label{thm: ls shift EIF}
    Suppose that the distribution of $X \mid A=a,Y=y$ has a differentiable Lebesgue density supported on $\real^d$ for all $a \in \{0,1\}$ and $y \in \mathcal{Y}$. Recall $\rho_* := P_*(A=0)$, $g_{*,Y}: y \mapsto P_*(A=0 \mid Y=y)$ and $\mathcal{E}_{*,Y}: y \mapsto \expect_{P_*}[\ell(X,Y) \mid Y=y,A=0]$.  
    Let $\tau_*(x \mid y)$ be the Lebesgue density of $X \mid A=0,Y=y$ under $P_*$, and define\footnote{When taking derivatives of $\tau_*$ or $\log \tau_*$, we treat $\tau_*$ as a function of $x$ with $y$ fixed; for example, $(\log \tau_*)'$ denotes the function $(x \mid y) \mapsto \partial_\sharp \log \tau_*(\sharp \mid y) |_{\sharp=x} \in \real^d$.}
    \begin{align*}
        \psi_y(x) &:= W_*(y) x + b_*(y), \qquad
        \Phi(x,y) := \left\{ \ell(x,y) - \mathcal{E}_{*,Y}(y) \right\}/\rho_*, \\
        \Gamma(x,y) &:= (x-b_{*}(y)) \cdot \partial_1 \log \tau_*(x \mid y),\quad \Lambda(x,y) :=  I_{d\times d} + (x - b_*(y)) \{ (\log \tau_*)'(x \mid y)\}^\top,\\
        a(y) &:= \begin{pmatrix}
            \expect_{P_*} \left[ \Phi(X,Y) \left\{ 1 +\Gamma(X,Y)\right\} \mid A=0,Y=y \right] \\
            \expect \left[ \Phi(X,Y) \cdot \partial_1 \log \tau_*(X \mid Y) \mid A=0,Y=y \right]
        \end{pmatrix}, \\
        M_{11}(y) &:= \expect_{P_*} \Big[ \tr \left( W_*(Y)^{-1} \Lambda(X,Y) \right) \cdot
        \left\{ 1 +\Gamma(X,Y)\right\} \mid A=0,Y=y \Big], \\
        M_{12}(y) &:= \expect_{P_*} \Big[ \{ (\log \tau_*)'(X \mid Y) \}^\top 1_d \cdot  \left\{ 1 + 
        \Gamma(X,Y) (X_1-b_{*,1}(Y)) \right\} \mid A=0,Y=y \Big], \\
        M_{21}(y) &:= \expect_{P_*} \Big[ \tr \left( W_*(Y)^{-1} \Lambda(X,Y) \right) \cdot
        \partial_1 \log \tau_*(X \mid Y) \mid A=0,Y=y \Big], \\
        M_{22}(y) &:= \expect_{P_*} \left[ \{ (\log \tau_*)'(X \mid Y) \}^\top 1_d \cdot \partial_1 \log \tau_*(X \mid Y) \mid A=0,Y=y \right], \\
        M(y) &:= \begin{pmatrix}
            M_{11}(y) & M_{12}(y) \\
            M_{21}(y) & M_{22}(y)
        \end{pmatrix}, \qquad
        c(x,y) := \begin{pmatrix}
        \tr \left( W_*(y)^{-1} \Lambda(x,y) \right) \\
        \{ (\log \tau_*)'(x \mid y) \}^\top 1_d
    \end{pmatrix}, \\
    N(o) &:= 
    (1-a) \left\{ (g_{*,Y}(y)) \Phi(x,y) - (1-g_{*,Y}(y)) c(x,y)^\top M(y)^{-1} a(y) \right\} \\
            &\quad+ a g_{*,Y}(y) \left\{ \Phi(\psi_y(x),y) + c(\psi_y(x), y) M(y)^{-1} a(y) \right\}.
    \end{align*}
 Under Condition~\ref{DScondition: ls generalized target shift}, the efficient influence function for estimating $r_*$ is 
 $D_\LSgentargetshift: o \mapsto (1-a) \left\{ \mathcal{E}_{*,Y}(y) - r_* \right\}/\rho_* +N(o)$.
\end{theorem}

\citetsupp{Zhang2013} introduced another similar dataset shift condition, \emph{location-scale conditional shift}. We present the corresponding results in Supplement~\ref{sec: other DS conditions2}.

\subsection{Invariant-density-ratio-type conditions}

\cite{Tasche2017} introduced a dataset shift condition termed \textit{invariant density ratio} for a binary label $Y \in \{0,1\}$. Motivated by this condition, we consider the following weaker and more general version.
\begin{DScondition}[Invariant density ratio shape] \label{DScondition: invariant density ratio shape}
    The following equality between two Radon-Nikodym derivatives holds:
    $$\frac{\intd P_{*,X \mid Y=0,A=0}}{\intd P_{*,X \mid Y=1,A=0}} = \const_* \frac{\intd P_{*,X \mid Y=0,A=1}}{\intd P_{*,X \mid Y=1,A=1}}$$
    for an unknown constant
    $$\const_* = \left\{ \int \frac{\intd P_{*,X \mid Y=0,A=1}}{\intd P_{*,X \mid Y=1,A=1}}(x) P_{*,X \mid Y=1,A=0}(\intd x) \right\}^{-1} > 0$$
\end{DScondition}
The constant $\const_*$ is a normalizing constant to ensure that $P_{*,X \mid Y=0,A=a}$ is a probability measure for each $a \in \{0,1\}$.
The above condition specializes to the original invariant density ratio condition in \cite{Tasche2017} when $\const_*=1$.
See Condition~\ref{DScondition: invariant density ratio} in the Supplemental material for a formal statement.
Both these conditions are implied by, and thus weaker than, label shift (\ref{DScondition: label shift}). Invariant density ratio also coincides with covariate shift (\ref{DScondition: cov shift}) under equality between the outcome prevalence in the two populations, namely $\Prob_{P_*}(Y=1 \mid A=1) = \Prob_{P_*}(Y=1 \mid A=0)$ \citepsupp{Tasche2017}.

We have the following result on the efficient influence function for estimating $r_*$ under Condition~\ref{DScondition: invariant density ratio shape}.

\begin{theorem} \label{thm: invariant density ratio shape}
    Recall $\rho_* := P_*(A=0)$ and $\mathcal{E}_{*,Y}: y \mapsto \expect_{P_*}[\ell(X,Y) \mid Y=y,A=0]$. Let $\eta_*(x \mid y,a)$ denote the density of $P_{*,X \mid Y=y,A=a}$ (with respect to a dominating measure $\mu$) evaluated at $x$, $\iota_*: x \mapsto \eta_*(x \mid 1,0) \eta_*(x \mid 0,1) / \eta_*(x \mid 1,1)$ denote the unnormalized density of $P_{*,X \mid Y=0,A=0}$, and $w_*(y,a) := \Prob_{P_*}(Y=y,A=a)$.
    Define
    \begin{align*}
        \mathcal{I}(x) &:= \frac{\ell(x,0) - \mathcal{E}_{*,Y}(0) + \expect_{P_*}[\ell(X,1) \mid Y=0,A=0] - \ell(x,1)}{\rho_*}, \\
        \mathcal{J}(x) &:= \frac{1}{w_*(0,0) \const_* \iota(x)} + \frac{1}{w_*(1,0) \eta_*(x \mid 1,0)} + \frac{1}{w_*(0,1) \eta_*(x \mid 0,1)} + \frac{1}{w_*(1,1) \eta_*(x \mid 1,1)}, \\
        \mathcal{D} &:= \frac{ \expect_{P_*} \left[ \left( 1- \frac{1}{\mathcal{J}(X) w_*(0,0) \const_* \iota_*(X)} \right) \mathcal{I}(X) \mid Y=0,A=0 \right] }{ \expect_{P_*} \left[ \frac{1}{\mathcal{J}(X) w_*(0,0) \const_* \iota_*(X)} \mid Y=0,A=0 \right] }, \\
        \mathscr{F} &: x \mapsto \frac{\mathcal{I}(x) + \mathcal{D}}{\mathcal{J}(x)}, \\
        H_{1,1} &: x \mapsto - \frac{\mathscr{F}(x)}{w_*(1,1) \eta_*(x \mid 1,1)}, \\
        H_{0,1} &: x \mapsto \frac{\mathscr{F}(x)}{w_*(0,1) \eta_*(x \mid 0,1)}, \\
        H_{1,0} &: x \mapsto \frac{\ell(x,1)-\mathcal{E}_{*,Y}(1)}{\rho*} + \frac{\mathscr{F}(x)}{w_*(1,0) \eta_*(x \mid 1,0)}, \\
        H_{0,0} &: x \mapsto H_{1,0}(x) + H_{0,1}(x) - H_{1,1}(x) - \expect_{P_*}[ H_{1,0}(X) + H_{0,1}(X) - H_{1,1}(X) \mid A=0, Y=0].
    \end{align*}
    Under Condition~\ref{DScondition: invariant density ratio shape}, the efficient influence function for estimating $r_*$ is $D_\IDRS: o \mapsto (1-a) \left\{ \mathcal{E}_{*,Y}(y) - r_* \right\}/\rho_* + \sum_{(y',a') \in \{0,1\}^2} \ind(y=y',a=a') H_{y,a}(x)$.
\end{theorem}

\section{Additional results on other dataset shift conditions} \label{sec: other DS conditions2}

In this section, we present additional results for dataset shift conditions other than sequential conditionals and its special cases (Conditions~\ref{DScondition: independence}--\ref{DScondition: label shift}).
As in Supplement~\ref{sec: label shift and Y con shift}, we use a superscript $\dagger$ in the labels to represent results parallel to those in Section~\ref{sec: other DS conditions} in the main text.

\subsection{Posterior drift conditions}

\begin{DScondition2}{DScondition: cov shift posterior drift}[Posterior Drift] \label{DScondition: posterior drift}
    Condition~\ref{DScondition: cov shift posterior drift} holds and $X \independent A$.
\end{DScondition2}

\citepsupp{Scott2019} further assumed $\tilde{\phi}_*$ (and thus $\phi_*$) to be strictly increasing in the original conditions.
We have dropped the monotonicity requirement on $\tilde{\phi}_*$ because this condition does not change the geometry of the induced semiparametric model. Nevertheless, the monotonicity condition may be useful when constructing an efficient estimator of the risk $r_*$ if the nuisance function $\tilde{\phi}_*$ needs to be estimated. We have additionally assumed differentiability of $\tilde{\phi}_*$ (and $\phi_*$) to ensure sufficient smoothness of the density and the risk functional $P \mapsto R(P) = \expect_P[\ell(X,Y) \mid A=0]$.

\cite{cai2021transfer} 
used a different condition for posterior drift,
namely that the supports of $P_{*,X \mid A=0}$ and $P_{*,X \mid A=1}$ are the same. This latter condition in \cite{cai2021transfer} does not change the geometry of the model space and therefore we do not study it here.

\begin{theorem2}{thm: posterior drift EIF} \label{thm: posterior drift EIF2}
    With the notations in Theorem~\ref{thm: posterior drift EIF}, under Condition~\ref{DScondition: posterior drift}, the efficient influence function for estimating $r_*$ is $D_\postdrift: o \mapsto \mathcal{E}_*(x) - r_* + F(g_1,g_2)(o)$.
\end{theorem2}

\subsection{Location-scale shift conditions}

\begin{DScondition2}{DScondition: ls generalized target shift}[Location-Scale Conditional Shift] \label{DScondition: ls conditional shift}
    Condition~\ref{DScondition: ls generalized target shift} holds and $Y \independent A$.
\end{DScondition2}

\citetsupp{Zhang2013} further assumed that $W_*(y)$ is diagonal, but not necessarily invertible. 
We have relaxed the requirement that $W_*(y)$ is diagonal, but restricted to invertible matrices for analytic tractability.

\begin{theorem2}{thm: ls shift EIF} \label{thm: ls shift EIF2}
    Suppose that the distribution of $X \mid A=a,Y=y$ has a differentiable Lebesgue density supported on $\real^d$ for all $a \in \{0,1\}$ and $y \in \mathcal{Y}$.
    With the notations in Theorem~\ref{thm: ls shift EIF}, under Condition~\ref{DScondition: ls conditional shift},
    the efficient influence function for estimating $r_*$ is $D_\LScondshift: o \mapsto \mathcal{E}_{*,Y}(y) - r_* +N(o)$.
\end{theorem2}

\subsection{Invariant density ratio condition}

\cite{Tasche2017} introduced the following dataset shift condition for a binary label $Y \in \{0,1\}$.

As mentioned in Section~\ref{sec: other DS conditions}, the following invariant density ratio condition, introduced by \cite{Tasche2017}, is a special case of Condition~\ref{DScondition: invariant density ratio shape} with $\const_*=1$.
\begin{DScondition2}{DScondition: invariant density ratio shape}[Invariant density ratio] \label{DScondition: invariant density ratio}
    The following equality between two Radon-Nikodym derivatives holds:
    \begin{equation}
        \frac{\intd P_{*,X \mid Y=0,A=1}}{\intd P_{*,X \mid Y=1,A=1}} = \frac{\intd P_{*,X \mid Y=0,A=0}}{\intd P_{*,X \mid Y=1,A=0}}. \label{eq: invariant density ratio}
    \end{equation}
\end{DScondition2}
The efficiency bound under this condition is for future research. A major challenge is constructing regular parametric submodels satisfying this condition.

We next connect this condition to the collapsibility of odds ratios.
By Bayes' Theorem, Condition~\ref{DScondition: invariant density ratio} is equivalent to
\begin{align*}
    &\frac{P_*(Y=0 \mid X=x,A=1)}{P_*(Y=1 \mid X=x,A=1)} \frac{P_*(Y=1 \mid A=1)}{P_*(Y=0 \mid A=1)} \\
    &= \frac{P_*(Y=0 \mid X=x,A=0)}{P_*(Y=1 \mid X=x,A=0)} \frac{P_*(Y=1 \mid A=0)}{P_*(Y=0 \mid A=0)}.
\end{align*}
With $\mathrm{Odds}(Y \mid X,A) := P_*(Y=1 \mid X,A)/P_*(Y=0 \mid X,A)$ and $\mathrm{Odds}(Y \mid A) := P_*(Y=1 \mid A)/P_*(Y=0 \mid A)$, this condition is further equivalent to
$$\frac{\mathrm{Odds}(Y \mid X=x,A=1)}{\mathrm{Odds}(Y \mid X=x,A=0)} = \frac{\mathrm{Odds}(Y \mid A=1)}{\mathrm{Odds}(Y \mid A=0)}.$$
In other words, a conditional odds ratio equals a marginal odds ratio; that is, the odds ratio is collapsible. 
More precisely, this collapsibility is called \textit{strict collapsibility} in \citetsupp{Ducharme1986} or \textit{simple collapsibility} in \citetsupp{Xie2008}.
Despite results suggesting that odds ratios are collapsible if and only if conditional independence holds (e.g., Theorem~1 in \cite{Ducharme1986} and Theorem~4 in \cite{Xie2008}), we remark that these results concern a stronger form of collapsibility.
The collapsibility equivalent to Condition~\ref{DScondition: invariant density ratio} is weaker and may hold without conditional independence.
\citetsupp{Whittemore1978} studied this problem in more depth for the special case of finitely discrete covariate $X$.

\section{Prediction set construction based on risk estimation} \label{sec: pred set}

Following Example~\ref{ex: prediction sets}, in this supplement, we briefly review a few coverage guarantees for prediction sets and describe how our risk estimation methods can be helpful in constructing prediction sets. 
We consider using a dataset with the structure described in Section~\ref{sec: setup} and sample size $n$ to construct a prediction set function $\hat{C}: \mathcal{X} \to \mathscr{B}$, where $\mathscr{B} \subseteq 2^\mathcal{Y}$ is a Borel sigma algebra on $\mathcal{Y}$. We focus on coverage guarantees in the target population. We slightly abuse terminology and also call $\hat{C}$ a prediction set.

Most of the popular coverage guarantees fall into two categories: (i) marginal validity, and (ii) training-set conditional validity, also termed \emph{Probably Approximately Correct} (PAC). A discussion of the differences between these two categories can be found in the Supplementary Material of \cite{Qiu2022}. We describe these two categories in the following subsections. More details on this topic can be found in, for example, \cite{Angelopoulos2021,cp,Barber2022,Bian2022,Chernozhukov2018,lei2015conformal,Lei2021,Park2020,Park2022,Qiu2022,Tibshirani2019,shafer2008tutorial,vovk1999machine,vovk2005algorithmic,Vovk2013,Yang2022}.

\subsection{Marginal validity} \label{sec: marginal pred set}

A prediction set $\hat{C}$ is said to achieve marginal validity with coverage error at most $\alpha \in (0,1)$ if $P_*(Y \notin \hat{C}(X) \mid A=0) \leq \alpha$, where the triple $(X,Y,A)$ is a new draw from $P_*$ independent of the training data. 
This finite-sample guarantee can be achieved by applying conformal prediction methods to the target population data while ignoring the source population data \citep[see, e.g.,][etc.]{shafer2008tutorial,vovk1999machine,vovk2005algorithmic,Vovk2013}.

However, if the source population data is relevant under a dataset shift condition (e.g., one of Conditions~\ref{DScondition: X con shift}--\ref{DScondition: label shift}), it may be possible to prediction sets of smaller sizes with an asymptotic marginal coverage guarantee, namely $P_*(Y \notin \hat{C}(X) \mid A=0) \leq \alpha+\smallo(1)$. 
Such guarantees have been adopted when dataset shift is present \citep[see, e.g.,][]{Lei2021,Tibshirani2019,Yang2022}.

Suppose that a predictor $f$ trained on a held-out dataset and a scoring function $s: \mathcal{X} \times \mathcal{Y} \to \real$ are given. Without loss of generality, we assume that a higher score corresponds to more conformity. For example, when the outcome $Y$ is continuous, a common scoring function is $s(x,y) = -|y-f(x)|$; when the label $Y$ is discrete and $f(x)$ is a vector of estimated probabilities $P_*(Y=y \mid X=x)$, a common scoring function is $s(x,y)=f_y(x)$, the $y$-th entry of $f(x)$. For simplicity in our illustration, we assume that the distribution of $s(X,Y) \mid A=0$ is continuous.

Consider prediction sets of the form $C_\tau: x \mapsto \{ y \in \mathcal{Y}: s(x,y) \geq \tau \}$, where $\tau \in \real \cup \{\pm \infty\}$ is a threshold. Following the approach from \cite{Yang2022}, given an estimator $\hat{r}_{\tau}$ of the coverage error $r_{*,\tau} := \expect_{P_*}[\ind(Y \notin C_\tau(X)) \mid A=0]$ of $C_\tau$ that is consistent uniformly over $\tau$, we may find an appropriate threshold $\hat{\tau}$ by solving the estimating equation $\hat{r}_{\tau} = \alpha$ for $\tau$. If $\hat{r}_{\tau}$ is efficient, we expect that, with a high probability, $\hat{r}_{\tau}$ is close to $r_{*,\tau}$ and thus the selected threshold $\hat{\tau}$ is close to the true optimal threshold $\tau_*$, the $\alpha$-th quantile of $s(X,Y) \mid A=0$. In other words, the sampling distribution of $\hat{\tau}$ is expected to be more concentrated around $\tau_*$ compared to the threshold selected based on an inefficient estimator of $r_{*,\tau}$.

\subsection{Training-set conditional validity} \label{sec: PAC pred set}

A prediction set $\hat{C}$ is said to achieve training-set conditional validity, or to be $(\alpha,\delta)$-PAC, if $P_*( P_*( Y \notin \hat{C}(X) \mid A=0, \text{training data}) \leq \alpha ) \geq 1-\delta$, where the triple $(X,Y,A)$ is independent of the training data. This guarantee implies marginal validity with coverage error at most $\alpha+\delta$ by a union bound \citep[see, e.g., Lemma~S2 in][]{Qiu2022}. Similarly to marginal validity, this finite-sample guarantee can be achieved by applying split (or inductive) conformal prediction to the target population data while ignoring the source population data \citep[see, e.g.,][]{Park2020,Vovk2013}.

Using the source population data and our proposed efficient risk estimators, 
it is possible to construct prediction sets of smaller sizes with asymptotic training-set conditional validity.
\cite{Qiu2022} identified two types of asymptotic training-set conditional validity: (i) \emph{Asymptotically Probably Approximately Correct} (APAC), where $P_*( P_*( Y \notin \hat{C}(X) \mid A=0, \text{training data}) \leq \alpha ) \geq 1-\delta-\smallo(1)$, 
and (ii) \emph{Probably Asymptotically Approximately Correct} (PAAC), where $P_*( P_*( Y \notin \hat{C}(X) \mid A=0, \text{training data}) \leq \alpha+\smallo_p(1) ) \geq 1-\delta$. 
See \cite{Qiu2022} for a thorough discussion of the similarities and differences between these guarantees. 
In many cases, PAAC is similar to asymptotic marginal validity \citep[see, e.g.,][]{Bian2022,Yang2022}, and thus the approach we described in Section~\ref{sec: marginal pred set} can be adopted.

In contrast, the APAC guarantee requires techniques more similar to those used for obtaining PAC guarantees. 
We consider a setup similar to Section~\ref{sec: marginal pred set} and follow the approaches of \cite{Angelopoulos2021,Qiu2022} based on multiple testing. 
Let $\mathcal{T} \subseteq \real \cup \{\pm \infty\}$ be a given finite set of candidate thresholds, and $U_{\tau}$ be an upper confidence bound for $r_{*,\tau} = \expect_{P_*}[\ind(Y \notin C_\tau(X)) \mid A=0]$, whose coverage probability converges to $1-\delta$ uniformly over $\tau\in \mathcal{T}$. 
Then a threshold $\hat{\tau}$ can be selected as $\max\{\tau \in \mathcal{T}: U_{\tau'} < \alpha \text{ for all } \tau' \in \mathcal{T} \text{ such that } \tau' \leq \tau\}$, and $C_{\hat{\tau}}$ would be APAC. 
Such an upper confidence bound can be constructed based on asymptotic linear estimators $\hat{r}_{\tau}$ of $r_*=\expect_{P_*}[\ind(Y \notin C_\tau(X)) \mid A=0]$ via, for example, normal approximation or bootstrap. 
Efficient estimators would lead to small confidence upper bounds and thus a selected threshold $\hat{\tau}$ close to the true optimal threshold $\tau_*$. In other words, the sampling distribution of $\hat{\tau}$ is expected to be less conservative and lead to smaller prediction sets, compared to the threshold selected based on an inefficient estimator of $r_{*,\tau}$.

\section{Proofs} \label{sec: proof}

\subsection{Review of basic results in semiparametric efficiency theory}

The following lemma contains basic results in semiparametric efficiency theory and can be easily proved by the law of total expectation.

\begin{lemma} \label{lemma: orthogonality}
    Each collection of subspaces of $L^2(P_*)$ in the following list is mutually orthogonal:
    \begin{compactitem}
        \item $L^2(P_{*,X})$, $\tangent_{A \mid X}$, $\tangent_{Y \mid X,A}$;
        \item $L^2(P_{*,A})$, $\tangent_{X \mid A}$, $\tangent_{Y \mid X,A}$;
        \item $L^2(P_{*,X,A})$, $\tangent_{Y \mid X,A}$;
        \item $L^2(P_{*,A})$, $\tangent_{Z_k \mid \bar{Z}_{k-1},A}$ ($k \in [K]$).
    \end{compactitem}
    Consequently, $L^2_0(P_*) = \tangent_A \oplus \{ \oplus_{k=1}^K \tangent_{\bar{Z}_k \mid \bar{Z}_{k-1},A} \}$;
    for any function $f \in L^2_0(P_*)$, the projection of $f$ onto $\tangent_\sharp$ is $s \mapsto \expect_{P_*}[f(O) \mid \sharp=s]$, and the projection of $f \in L^2_0(P_*)$ onto $\tangent_{\sharp \mid \natural}$ is $(s,t) \mapsto \expect_{P_*}[f(O) \mid \sharp=s,\natural=t] - \expect_{P_*}[f(O) \mid \natural=t]$.
\end{lemma}

When we derive the efficient influence function for each dataset shift condition, we first characterize the tangent space at $P_*$, which is a subspace of the nonparametric tangent space $L^2_0(P_*)$. We also characterize the orthogonal complement of the tangent space. The efficient influence function is then the projection of the nonparametric influence function $D_\nonparametric(P_*,r_*)$ onto the tangent space. We also ignore the difference between ``almost everywhere'' and ``everywhere'' for convenience; the exceptional probability-zero set can always be excluded from our analysis.

\subsection{No dataset shift condition: nonparametric model} \label{sec: proof nonparametric}

The following lemma states the asymptotic efficiency of the nonparametric estimator $\hat{r}_{\nonparametric}$ under a nonparametric model.

\begin{lemma} \label{lemma: nonparametric efficiency}
Under a nonparametric model at $P_*$, the efficient influence function for estimating $r_*$ is $D_\nonparametric(\rho_*,r_*)$ from \eqref{dnp}.
The sequence of estimators $\hat{r}_{\nonparametric}$ is RAL with the influence function $D_\nonparametric(\rho_*,r_*)$, and thus is nonparametrically efficient. Consequently, $\sqrt{n} (\hat{r}_{\nonparametric} - r_*) \overset{d}{\to} \mathrm{N}(0,\sigma_{*,\nonparametric}^2)$ with $\sigma_{*,\nonparametric}^2 := \expect_{P_*}[D_\nonparametric(\rho_*,r_*)$$(O)^2]$.
\end{lemma}

\begin{proof}
    We first prove that $D_\nonparametric(P_*,r_*)$ is the efficient influence function. Let $H$ be an arbitrary bounded function in $L^2_0(P_*)$. 
    The $L^2(P_*)$-closure of the collection of such functions is $L^2_0(P_*)$. 
    Consider the projections of $H$ onto $\tangent_A$ and $\tangent_{Z \mid A}$, respectively:
    $$H_A: a \mapsto \expect_{P_*}[H(O) \mid A=a], \quad H_{Z \mid A}: (z \mid a) \mapsto H(O) - H_A(a).$$
    Define a regular parametric submodel indexed by $\epsilon$ via
    $$\left\{P^\epsilon: \frac{\intd P^\epsilon}{\intd P_*}(o) = (1+ \epsilon H_A(a)) (1+ \epsilon H_{Z \mid A}(z \mid a)) \right\}$$
    for $\epsilon$ in a small neighborhood of zero such that $P^\epsilon$ has non-negative density with respect to $P_*$. 
    Clearly, $P^\epsilon=P_*$ at $\epsilon=0$.
    Further, by 
    \citet[Lemma~1.8 in Part III][]{Bolthausen2002}
    and direct calculation, 
    one can verify that and this submodel is a differentiable path at $\epsilon=0$ with score function being $H$ \citep[see, e.g., Definition~1.6 in Part III of][]{Bolthausen2002}.

    Recall that $R(P)$ denotes the risk in the target population associated with data-generating distribution $P$, namely $\expect_{P}[\ell(Z) \mid A=0]$. 
    To show that  $R(P)$ is differentiable at $P_*$ relative to the tangent set $L^2_0(P_*)$,
    with efficient influence function  $D_\nonparametric(\rho_*,r_*)$, it suffices to show that
    \begin{equation}
        \left. \frac{\intd R(P^\epsilon)}{\intd \epsilon} \right|_{\epsilon=0} = P_* H D_\nonparametric(\rho_*, r_*) \label{eq: np EIF def},
    \end{equation}
    since, for nonparametric models, the tangent space is $L^2_0(P_*)$ and the influence function is unique \citep[see, e.g., Example~1.12 in Part~III of][]{Bolthausen2002}. 
    
    We next derive \eqref{eq: np EIF def}.
    By definition,
    $$R(P) = \int \ell(z) P_{Z \mid A}(\intd z \mid 0).$$
    Thus, since integration and differentiation can be interchanged due to boundedness,
    $$\left. \frac{\intd R(P^\epsilon)}{\intd \epsilon} \right|_{\epsilon=0} = \int \ell(z) H_{Z \mid A}(z \mid 0) P_{*,Z \mid A}(\intd z \mid 0).$$
    Since $H_{Z \mid A} \in \tangent_{Z \mid A}$ has mean zero conditional on $A$, we may center $\ell(z)$ and show that the quantity on the right-hand side above equals
    $$\int (\ell(z) - r_*) H_{Z \mid A}(z \mid 0) P_{*,Z \mid A}(\intd z \mid 0).$$
    We can further use inverse-probability weighting to take into account the randomness in $A$ to show that the above quantity equals
    $$\iint \frac{\ind(a=0)}{\rho_*} (\ell(z) - r_*) H_{Z \mid A}(z \mid a) P_{*,Z \mid A}(\intd z \mid a) P_{*,A}(\intd a).$$
    Since $o \mapsto \ell(z) - r_*$ is orthogonal to $o \mapsto \frac{\ind(a=0)}{\rho_*} H_A(a)$, the above quantity further equals
    \begin{align*}
        & \iint \frac{\ind(a=0)}{\rho_*} (\ell(z) - r_*) \{ H_{Z \mid A}(z \mid a) + H_A(a) \} P_{*,Z \mid A}(\intd z \mid a) P_{*,A}(\intd a) \\
        &= \int \frac{\ind(a=0)}{\rho_*} (\ell(z) - r_*) H(o) P_*(\intd o) = P_* H D_\nonparametric(\rho_*,r_*).
    \end{align*}
    We have shown that $D_\nonparametric(\rho_*,r_*)$ is the efficient influence function.

    We next show that the influence function of $\hat{r}_{\nonparametric}$ is $D_\nonparametric(\rho_*,r_*)$. By the central limit theorem, 
    \begin{align*}
    \frac{1}{n} \sum_{i=1}^n \ind(A_i=0) \ell(Z_i) 
    &= \rho_* r_* + \frac{1}{n} \sum_{i=1}^n \{ \ind(A_i=0) \ell(Z_i) - \rho_* r_* \}
    = \rho_* r_* + \bigO_p\left(\frac{1}{n^{1/2}}\right)
    \end{align*} 
    and $\frac{1}{n} \sum_{i=1}^n \ind(A_i=0) = \rho_* + \frac{1}{n} \sum_{i=1}^n \{\ind(A_i=0) - \rho_*\} = \rho_* + \bigO_p(n^{-1/2})$.
    Hence, by a Taylor expansion of $x \mapsto 1/x$ around $\rho_*>0$, we have that
    \begin{align*}
        &\hat{r}_{\nonparametric} = \frac{\frac{1}{n} \sum_{i=1}^n \ind(A_i=0) \ell(Z_i)}{\frac{1}{n} \sum_{i=1}^n \ind(A_i=0)} 
        = \frac{\rho_* r_* + \frac{1}{n} \sum_{i=1}^n \{ \ind(A_i=0) \ell(Z_i) - \rho_* r_* \}}{ \rho_* + \frac{1}{n} \sum_{i=1}^n \{ \ind(A_i=0) - \rho_* \} } \\
        &= \left[ \rho_* r_* + \frac{1}{n} \sum_{i=1}^n \{ \ind(A_i=0) \ell(Z_i) - \rho_* r_* \} \right] \times \left[ \frac{1}{\rho_*} - \frac{1}{\rho_*^2} \frac{1}{n} \sum_{i=1}^n \{ \ind(A_i=0) - \rho_* \} + \smallo_p(n^{-1/2}) \right] \\
        &= r_* + \frac{1}{n} \sum_{i=1}^n D_\nonparametric(\rho_*,r_*)(O_i) + \smallo_p(n^{-1/2}).
    \end{align*}
    We have thus shown that $\hat{r}_{\nonparametric}$ is asymptotically linear with influence function $D_\nonparametric(\rho_*,r_*)$. Its regularity and asymptotic efficiency follow from Lemma~2.9 in Part~III of \cite{Bolthausen2002}.
\end{proof}

\subsection{General sequential conditionals} \label{sec: independence proof}

\begin{proof}[Proof of the efficiency bound under Condition~\ref{DScondition: general independence} or \ref{DScondition: independence}]
    This result
    is implied by Theorem~2 and Corollary~1 from \citetsupp{Li2023}.
    Specifically, our Condition~\ref{DScondition: independence} reduces to Conditions~1 and 2 from \citetsupp{Li2023} with the target population distribution---$\underline{Q}^0$ in their notation---being the distribution $Q$ under Condition~\ref{DScondition: general independence} and the distribution of $Z \mid A=0$ under $P_*$ under Condition~\ref{DScondition: independence}, respectively.
    Moreover, our $A$ is their $S$, and our $\mathcal{S}_k'$ is their $\mathcal{S}_k$.
    
    Condition~3 from \citetsupp{Li2023} is assumed for Condition~\ref{DScondition: general independence}. For Condition~\ref{DScondition: independence}, by Bayes' Theorem, the Radon-Nikodym derivative of the distribution of $\bar{Z}_{k-1} \mid A=0$ relative to $\bar{Z}_{k-1} \mid A \in \mathcal{S}'_{k}$---$\lambda_{k-1}$ in their notation---is shown in \eqref{eq: independence density ratio}, 
    and is uniformly bounded away from zero and infinity in the support $\mathcal{Z}_{k-1}$ of $\bar{Z}_{k-1} \mid A=0$---which is $\mathcal{Z}^\dagger_{j-1}$ in their notation.
    Hence, Condition~3 from \citetsupp{Li2023} is also satisfied.
    
    Moreover, when the observed data is only from the target population, since we impose no restriction on the distribution of $Z$ in the target population, the corresponding tangent space---$\mathcal{T}(Q_0, \mathcal{Q})$ in their notation---is $L^2(P_{*,Z \mid A=0})$; the influence function of $r_*$ is unique and equals $z \mapsto \ell(z) - r_*$. Thus, $r_j(\bar{z}_j)$ in Corollary~1 from \citetsupp{Li2023} equals the product of the Radon-Nikodym derivative in \eqref{eq: independence density ratio} and $\ell^j_*(\bar{z}_j) - \ell^{j-1}_*(\bar{z}_{k-1})$, which is zero for $\bar{z}_j$ outside the support $\mathcal{Z}_{k-1}$ of $\bar{Z}_{k-1} \mid A=0$.
    We plug the above results into Corollary~1 from \citetsupp{Li2023} to obtain the claimed efficiency bound under Condition~\ref{DScondition: independence}.
\end{proof}

We next present a lemma that will be useful in the proof of Theorem~\ref{thm: independence eff and DR}.
\begin{lemma} \label{lemma: general independence k-th bias}
    For every $k \in [2:K]$ and $v \in [V]$, 
    with $B_k$ from \eqref{bk},
    \begin{align}
        & \left( \sum_{a \in \mathcal{S}'_1} \hat{\pi}^a_v \right) \expect_{P_*} \left[ \frac{\ind(A \in \mathcal{S}'_k)} {\sum_{a \in \mathcal{S}'_k} \hat{\pi}^{a}_{v}} \hat{\lambda}^{k-1}_{v}(\bar{Z}_{k-1}) [\hat{\ell}^k_{v}(\bar{Z}_k) - \hat{\ell}^{k-1}_{v}(\bar{Z}_{k-1})] \right] \nonumber \\
        &= \left( \sum_{a \in \mathcal{S}'_1} \hat{\pi}^a_v \right) B_{k,v} + \left( \sum_{a \in \mathcal{S}'_1} \pi^a_* \right) \expect_Q \left[ \expect_{P_*} \left[ \hat{\ell}^k_{v}(\bar{Z}_k) \mid \bar{Z}_{k-1}, A \in \mathcal{S}'_k \right] - \hat{\ell}^{k-1}_{v}(\bar{Z}_{k-1}) \right]. \label{eq: general independence k-th bias}
    \end{align}
    Additionally, under Condition~\ref{DScondition: general independence}, we have the following identity for the second term on the right-hand side of \eqref{eq: general independence k-th bias}:
    $$\expect_Q \left[ \expect_{P_*} \left[ \hat{\ell}^k_{v}(\bar{Z}_k) \mid \bar{Z}_{k-1}, A \in \mathcal{S}'_k \right] - \hat{\ell}^{k-1}_{v}(\bar{Z}_{k-1}) \right] = \expect_Q[\hat{\ell}^k_{v}(\bar{Z}_k) - \hat{\ell}^{k-1}_{v}(\bar{Z}_{k-1})].$$
\end{lemma}
\begin{proof}
    Equation~\ref{eq: general independence k-th bias} follows by adding and subtracting $\lambda^{k-1}_*$ from $\hat{\lambda}^{k-1}_{v}$,
    by the definition of $\lambda^{k-1}_*$ as the Radon-Nikodym derivative of the distribution of $\bar{Z}_{k-1}$ under $Q$ relative to that of $\bar{Z}_{k-1} \mid A \in \mathcal{S}'_k$ under $P_*$, 
    by $P_*(A \in \mathcal{S}'_k) = \sum_{a \in \mathcal{S}'_k} \pi^{a}_*$,
    and by using the law of iterated expectation. 
    Under Condition~\ref{DScondition: general independence}, for $Q$-almost every $\bar{z}_{k-1} \in \widetilde{\mathcal{Z}}_{k-1}$,
    $$\expect_{P_*} \left[ \hat{\ell}^k_{v}(\bar{Z}_k) \mid \bar{Z}_{k-1}=\bar{z}_{k-1}, A \in \mathcal{S}'_k \right] = \expect_Q \left[ \hat{\ell}^k_{v}(\bar{Z}_k) \mid \bar{Z}_{k-1}=\bar{z}_{k-1} \right]$$
    and the second claim follows.
\end{proof}

\begin{proof}[Proof of Theorem~\ref{thm: independence eff and DR}]
    We first focus on the more general case, namely Condition~\ref{DScondition: general independence} and Alg.~\ref{alg: general independence estimator}.
    We study the asymptotic behavior of the estimator $\hat{r}_{v}$ for each fold $v$ via methods for analyzing Z-estimators \citep[see, e.g., Section~3.3 in][]{vandervaart1996}, 
    and the asymptotic behavior of the cross-fit estimator $\hat{r}$ follows immediately.
    By the definitions of the nuisance functions $\boldsymbol \ell_*$ from \eqref{genell}, the Radon-Nikodym derivatives $\boldsymbol \lambda_*=(\lambda^k_*)_{k=1}^{K-1}$,
    the marginal probabilities $\boldsymbol \pi_* := (\pi_*^{a})_{a \in \mathcal{A}}$ of the source populations,
    and $D_\generalindependence$ from \eqref{eq: general independence EIF}, we have that
    \begin{equation}
        P_* \left( \sum_{a \in \mathcal{S}'_1} \pi^{a}_* \right) D_\generalindependence(\boldsymbol \ell_*,\boldsymbol \lambda_*,\boldsymbol \pi_*,r_*) = 0. \label{eq: general independence EIF mean zero}
    \end{equation}
    By the definitions of
    $\widetilde{\mathcal{T}}$ from \eqref{ttilde},
    $\hat{r}_{v}$ from \eqref{eq: general independence foldwise estimator}, 
    and $D_\generalindependence$ from \eqref{eq: general independence EIF}, 
    noting that 
    \begin{equation*}
    \sum_{i\in I_v} \frac{\ind(A_i \in \mathcal{S}'_1)}{\sum_{a \in \mathcal{S}'_1} \hat\pi_v^{a}}
    =
    \sum_{b \in \mathcal{S}'_1} \frac{\sum_{i\in I_v}\ind(A_i = b)}{\sum_{a \in \mathcal{S}'_1} \hat\pi_v^{a}} = |I_v|,
\end{equation*}
    we find that
    $$P^{n,v} \left( \sum_{a \in \mathcal{S}'_1} \hat\pi^{a}_{v} \right) D_\generalindependence(\widehat{\boldsymbol \ell}_{v}, \widehat{\boldsymbol \lambda}_{v}, \widehat{\boldsymbol \pi}_{v}, \hat{r}_{v}) = 0.$$
    Then we have the following equation:
    \begin{align}
        0 &= (P^{n,v}-P_*) \left\{ \left( \sum_{a \in \mathcal{S}'_1} \hat\pi^{a}_{v} \right) D_\generalindependence(\widehat{\boldsymbol \ell}_{v}, \widehat{\boldsymbol \lambda}_{v}, \widehat{\boldsymbol \pi}_{v}, \hat{r}_{v}) - \left( \sum_{a \in \mathcal{S}'_1} \pi^{a}_* \right) D_\generalindependence(\boldsymbol \ell_*,\boldsymbol \lambda_*,\boldsymbol \pi_*,r_*) \right\} \nonumber\\
        &\quad+ P_* \left\{ \left( \sum_{a \in \mathcal{S}'_1} \hat\pi^{a}_{v} \right) D_\generalindependence(\widehat{\boldsymbol \ell}_{v}, \widehat{\boldsymbol \lambda}_{v}, \widehat{\boldsymbol \pi}_{v}, \hat{r}_{v}) - \left( \sum_{a \in \mathcal{S}'_1} \pi^{a}_* \right) D_\generalindependence(\boldsymbol \ell_*,\boldsymbol \lambda_*,\boldsymbol \pi_*,r_*) \right\} \nonumber\\
        &\quad+ (P^{n,v}-P_*) \left( \sum_{a \in \mathcal{S}'_1} \pi^{a}_* \right) D_\generalindependence(\boldsymbol \ell_*,\boldsymbol \lambda_*,\boldsymbol \pi_*,r_*).
    \label{eq: general independence Z expanssion}
    \end{align}
    We next further analyze the first two terms on the right-hand side of \eqref{eq: general independence Z expanssion}.

    \textbf{Term~1:}
    By the definition of
    $D_\generalindependence$ from \eqref{eq: general independence EIF}, this term equals
    $$(P^{n,v}-P_*) \left\{ \left( \sum_{a \in \mathcal{S}'_1} \hat\pi^{a}_{v} \right) \widetilde{\mathcal{T}}(\widehat{\boldsymbol \ell}_{v}, \widehat{\boldsymbol \lambda}_{v}, \widehat{\boldsymbol \pi}_{v}) - \left( \sum_{a \in \mathcal{S}'_1} \pi^{a}_* \right) \widetilde{\mathcal{T}}(\boldsymbol \ell_*,\boldsymbol \lambda_*,\boldsymbol \pi_*) - \ind_{\mathcal{S}'_1} \left( \hat{r}_{v} - r_* \right) \right\},$$
    which further equals
    $$(P^{n,v}-P_*) \left\{ \left( \sum_{a \in \mathcal{S}'_1} \hat\pi^{a}_{v} \right) \widetilde{\mathcal{T}}(\widehat{\boldsymbol \ell}_{v}, \widehat{\boldsymbol \lambda}_{v}, \widehat{\boldsymbol \pi}_{v}) - \left( \sum_{a \in \mathcal{S}'_1} \pi^{a}_* \right) \widetilde{\mathcal{T}}(\boldsymbol \ell_*,\boldsymbol \lambda_*,\boldsymbol \pi_*) \right\} - \left( \sum_{a \in \mathcal{S}'_1} \hat{\pi}^a_v - \sum_{a \in \mathcal{S}'_1} \pi^a_* \right) (\hat{r}_v - r_*).$$

    \textbf{Term~2:} By the definitions of $D_\generalindependence$ and $\widetilde{\mathcal{T}}$ from \eqref{eq: general independence EIF} and \eqref{ttilde} respectively, 
    \eqref{eq: general independence EIF mean zero}, and Lemma~\ref{lemma: general independence k-th bias},
    this term equals
    \begin{align*}
        & P_* \left\{ \left( \sum_{a \in \mathcal{S}'_1} \hat\pi^{a}_{v} \right) \widetilde{\mathcal{T}}(\widehat{\boldsymbol \ell}_{v}, \widehat{\boldsymbol \lambda}_{v}, \widehat{\boldsymbol \pi}_{v}) - \ind_{\mathcal{S}'_1} \hat{r}_{v} \right\} \\
        &= \left( \sum_{a \in \mathcal{S}'_1} \hat{\pi}^a_v \right) \sum_{k=2}^K B_{k,v} + \left( \sum_{a \in \mathcal{S}'_1} \pi^a_* \right) \sum_{k=2}^K \expect_Q \left[ \expect_{P_*} \left[ \hat{\ell}^k_{v}(\bar{Z}_k) \mid \bar{Z}_{k-1}, A \in \mathcal{S}'_k \right] - \hat{\ell}^{k-1}_{v}(\bar{Z}_{k-1}) \right] \\
        &\quad+ \sum_{a \in \mathcal{S}'_1} \pi^{a}_* \expect_{P_*}[\hat{\ell}^1_v(Z_1) \mid A \in \mathcal{S}'_1] - \sum_{a \in \mathcal{S}'_1} \pi^{a}_* \hat{r}_v.
    \end{align*}
    By the definition of $\Delta_v$ from \eqref{eq: Delta v}, 
    and as $\hat{\ell}^K_{v}=\ell$, 
    this further equals
    \begin{align*}        
        &\left( \sum_{a \in \mathcal{S}'_1} \hat{\pi}^a_v \right) \sum_{k=2}^K B_{k,v} + \left( \sum_{a \in \mathcal{S}'_1} \pi^a_* \right) \Bigg\{ \sum_{k=2}^K \expect_Q \left[ \expect_{P_*} \left[ \hat{\ell}^k_{v}(\bar{Z}_k) \mid \bar{Z}_{k-1}, A \in \mathcal{S}'_k \right] - \hat{\ell}^{k-1}_{v}(\bar{Z}_{k-1}) \right] \\
        &\qquad\qquad+ \expect_{P_*}[\hat{\ell}^1_v(Z_1) \mid A \in \mathcal{S}'_1] -  \hat{r}_v \Bigg\} \\
        &= \left( \sum_{a \in \mathcal{S}'_1} \hat{\pi}^a_v \right) \sum_{k=2}^K B_{k,v} + \left( \sum_{a \in \mathcal{S}'_1} \pi^a_* \right) \left( r_* - \hat{r}_v + \frac{\sum_{a \in \mathcal{S}'_1} \hat{\pi}^a_v}{\sum_{a \in \mathcal{S}'_1} \pi^a_*} \Delta_v \right).
    \end{align*}

    With the above analyses, rearranging the terms in \eqref{eq: general independence Z expanssion}, we have that
    \begin{align}
        \hat{r}_v - r_* - \Delta_v &= \frac{1}{\sum_{a \in \mathcal{S}'_1} \hat{\pi}^a_v} (P^{n,v}-P_*) \left\{ \left( \sum_{a \in \mathcal{S}'_1} \hat\pi^{a}_{v} \right) \widetilde{\mathcal{T}}(\widehat{\boldsymbol \ell}_{v}, \widehat{\boldsymbol \lambda}_{v}, \widehat{\boldsymbol \pi}_{v}) - \left( \sum_{a \in \mathcal{S}'_1} \pi^{a}_* \right) \widetilde{\mathcal{T}}(\boldsymbol \ell_*,\boldsymbol \lambda_*,\boldsymbol \pi_*) \right\} \nonumber \\
        &\quad+ \sum_{k=2}^K B_{k,v} + \frac{\sum_{a \in \mathcal{S}'_1} \pi^{a}_*}{\sum_{a \in \mathcal{S}'_1} \pi^{a}_v} (P^{n,v}-P_*) D_\generalindependence(\boldsymbol \ell_*,\boldsymbol \lambda_*,\boldsymbol \pi_*,r_*). \label{eq: r hat v expansion}
    \end{align}
    By the definitions of $\hat{r}$ in Alg.~\ref{alg: general independence estimator} and of $\Delta$ in Theorem~\ref{thm: independence eff and DR}, we obtain \eqref{eq: r hat expansion}.

    We next show efficiency and multiply robust consistency. We first focus on $\hat{r}_v$ and condition on data out of fold $v$.
    We first focus on the most challenging term
    \begin{equation}
        (P^{n,v}-P_*) \left\{ \left( \sum_{a \in \mathcal{S}'_1} \hat\pi^{a}_{v} \right) \widetilde{\mathcal{T}}(\widehat{\boldsymbol \ell}_{v}, \widehat{\boldsymbol \lambda}_{v}, \widehat{\boldsymbol \pi}_{v}) - \left( \sum_{a \in \mathcal{S}'_1} \pi^{a}_* \right) \widetilde{\mathcal{T}}(\boldsymbol \ell_*,\boldsymbol \lambda_*,\boldsymbol \pi_*) \right\} \label{eq: empirical process}
    \end{equation}
    By the 
    definition of $\widetilde{\mathcal{T}}$ from \eqref{ttilde}, as well as the
    facts that
    (i) $\hat{\pi}^{a}_{v}$ is root-$n$ consistent for $\pi^{a}_*$, $a \in \mathcal{A}$ in $P^{n,v}$-probability, 
    and (ii) we condition on $\widehat{\boldsymbol \ell}_{v}, \widehat{\boldsymbol \theta}_{v}$, 
    by Lemmas~9.6 and 9.9 in \cite{Kosorok2008}, we have that, with $P^{n,v}$-probability tending to one as $n\to\infty$,
    the random function $o\mapsto \left( \sum_{a \in \mathcal{S}'_1} \hat\pi^{a}_{v} \right) \widetilde{\mathcal{T}}(\widehat{\boldsymbol \ell}_{v}, \widehat{\boldsymbol \lambda}_{v}, \widehat{\boldsymbol \pi}_{v})(o)- \left( \sum_{a \in \mathcal{S}'_1} \pi^{a}_* \right) \widetilde{\mathcal{T}}(\boldsymbol \ell_*,\boldsymbol \lambda_*,\boldsymbol \pi_*)(o)$,  
    whose randomness comes from $\widehat{\boldsymbol \pi}_{v}$,
    falls into a fixed VC-subgraph, and thus $P_*$-Donsker, class.
    Thus, under Condition~\ref{STcondition: independence DR}, the term in \eqref{eq: empirical process} is $\bigO_p(n^{-1/2})$
    conditional on data out of fold $v$.
    Under either Condition~\ref{STcondition: independence eff} or \ref{STcondition: independence DR}, since the $L^2(P_*)$ norm of $\left( \sum_{a \in \mathcal{S}'_1} \hat\pi^{a}_{v} \right) \widetilde{\mathcal{T}}(\widehat{\boldsymbol \ell}_{v}, \widehat{\boldsymbol \lambda}_{v}, \widehat{\boldsymbol \pi}_{v}) - \left( \sum_{a \in \mathcal{S}'_1} \pi^{a}_* \right) \widetilde{\mathcal{T}}(\boldsymbol \ell_*,\boldsymbol \lambda_*,\boldsymbol \pi_*)$ is $\bigO_p(1)$, we have that the term in \eqref{eq: empirical process} is $\bigO_p(n^{-1/2})$ that does not depend on nuisance estimators $(\widehat{\boldsymbol \ell}_{v}, \widehat{\boldsymbol \lambda}_{v})$, by the proof of Theorem~2.5.2 in \citetsupp{vandervaart1996} (specifically, Line~17 on page~128, the fourth displayed equation on that page).
    Therefore, the term in \eqref{eq: empirical process} is also $\bigO_p(n^{-1/2})$ unconditionally by Lemma~6.1 in \citetsupp{Chernozhukov2018debiasedML}.
    
    Similarly, under Condition~\ref{STcondition: independence eff}, 
    since the $L^2(P_*)$-norm of $o\mapsto \left( \sum_{a \in \mathcal{S}'_1} \hat\pi^{a}_{v} \right) \widetilde{\mathcal{T}}(\widehat{\boldsymbol \ell}_{v}, \widehat{\boldsymbol \lambda}_{v}, \widehat{\boldsymbol \pi}_{v})(o)- \left( \sum_{a \in \mathcal{S}'_1} \pi^{a}_* \right) \widetilde{\mathcal{T}}(\boldsymbol \ell_*,\boldsymbol \lambda_*,\boldsymbol \pi_*)(o)$
    converges to zero in probability,
    by the asymptotic uniform equicontinuity of $P_*$-Donsker classes \citep[Theorem~1.5.7 in][]{vandervaart1996}, 
    taking the distance on the space of functions to be the $L^2(P_*)$ distance, and evaluating it at zero and the above function,
    the associated empirical process converges to zero in probability.
    We conclude that the term in \eqref{eq: empirical process} is $\smallo_p(n^{1/2})$ under Condition~\ref{STcondition: independence eff}.
    \begin{compactitem}
        \item Under Condition~\ref{STcondition: independence DR}, by \eqref{eq: r hat v expansion}, $\hat{r}_{v} - \Delta_{v} - r_* = \smallo_p(1)$ conditional on data out of fold $v$. Additionally under Condition~\ref{DScondition: general independence}, clearly $\Delta_{v}=\Delta=0$. By the definitions of $\hat{r}$ in \eqref{eq: general independence estimator} and of $\Delta$, part~2 of Theorem~\ref{thm: independence eff and DR} follows from Lemma~6.1 in \citetsupp{Chernozhukov2018debiasedML}.

        \item Under Condition~\ref{STcondition: independence eff}, 
        by \eqref{eq: r hat v expansion}, $\hat{r}_{v} - \Delta_{v} - r_* = (P^{n,v} - P_*) D_\generalindependence(\boldsymbol \ell_*,\boldsymbol \lambda_*,\boldsymbol \pi_*,r_*) + \smallo_p(n^{-1/2})$.
        We then use the definitions of $\hat{r}$ from Line~9 of Alg.~\ref{alg: general independence estimator} and of $\Delta$ to obtain part~1 of Theorem~\ref{thm: independence eff and DR}.
        The regularity and efficiency of $\hat{r}$ follows from Lemma~2.9 in Part~III of \cite{Bolthausen2002}.
    \end{compactitem}

    The results for Condition~\ref{DScondition: independence} follows from the fact that Alg.~\ref{alg: general independence estimator} reduces to Alg.~\ref{alg: independence estimator} with $\hat{\lambda}^{k-1}_{v} = \sum_{a \in \mathcal{S}'_k} \hat{\pi}^a_v/\{\hat{\pi}^0_{v} (1+\hat{\theta}^{k-1}_{v})\}$.

    We finally prove \eqref{eq: r hat finite sample confidence}. We use $\const$ and $\const'$ to denote generic positive constants that may depend on $\epsilon$, $P_*$ and the bound on $\hat{\lambda}^{k-1}_v$ only and may vary from line to line. \eqref{eq: r hat expansion} is equivalent to
    \begin{align}
        & \hat{r} - \Delta -r_* = \frac{1}{n} \sum_{i=1}^n D_\generalindependence(\boldsymbol \ell_*,\boldsymbol \lambda_*,\boldsymbol \pi_*,r_*)(O_i) \nonumber \\
        &\quad+ \sum_{v \in [V]} \frac{|I_v|}{n \sum_{a \in \mathcal{S}'_1} \hat{\pi}^a_v} (P^{n,v}-P_*) \left\{ \left( \sum_{a \in \mathcal{S}'_1} \hat\pi^{a}_{v} \right) \widetilde{\mathcal{T}}(\widehat{\boldsymbol \ell}_{v}, \widehat{\boldsymbol \lambda}_{v}, \widehat{\boldsymbol \pi}_{v}) - \left( \sum_{a \in \mathcal{S}'_1} \pi^{a}_* \right) \widetilde{\mathcal{T}}(\boldsymbol \ell_*,\boldsymbol \lambda_*,\boldsymbol \pi_*) \right\} \nonumber \\
        &\quad+ \sum_{v \in [V]}  \frac{|I_v| \sum_{a \in \mathcal{S}'_1} \hat{\pi}^a_v}{n} \sum_{k=2}^K B_{k,v} + \sum_{v \in [V]} \frac{|I_v|}{n} \left( \frac{ \sum_{a \in \mathcal{S}'_1} \pi^{a}_*}{\sum_{a \in \mathcal{S}'_1} \pi^{a}_v}-1 \right) (P^{n,v}-P_*) D_\generalindependence(\boldsymbol \ell_*,\boldsymbol \lambda_*,\boldsymbol \pi_*,r_*). \label{eq: eq r hat expansion2}
    \end{align}
    
    We first show a basic result. For any $\delta>0$, by Hoeffding’s inequality and the fact that $|I_v|/n$ converges to a constant as $n \to \infty$,
    $$\Prob \left( \left|\sum_{a \in \mathcal{S}'_1} \pi^{a}_v - \sum_{a \in \mathcal{S}'_1} \pi^{a}_* \right| > \delta \right) \leq \exp(- \const \delta^2 |I_v|) \leq \exp(- \const \delta^2 n).$$

    Let $U$ denote the following union event:
    \begin{align*}
        & \| \hat{\lambda}^{k-1}_v - \lambda^{k-1}_* \|_{L^2(P_*)} > a_{n,k,v} \text{ for some $k$ and $v$}, \\
        \text{or } & \| \hat{\ell}^{k-1}_v - h^{k-1}_v \|_{L^2(P_*)} > b_{n,k,v} \text{ for some $k$ and $v$}.
    \end{align*}
    We consider two complementary events, $U$ and $U^C$.

    When $U$ occurs, by the union bound, $\Prob(U) \leq \sum_{v \in [V]} \sum_{k=2}^K (c_{n,k,v} + d_{n,k,v})$. We next focus on the case where $U^C$ occurs. We first bound the last three terms on the right-hand side of \eqref{eq: eq r hat expansion2}.

    For the second term on the right-hand side of \eqref{eq: eq r hat expansion2}, by the boundedness of $\hat{\lambda}^{k-1}_v$ and the proof of Theorem~2.5.2 in \citetsupp{vandervaart1996} (specifically, Line~17 on page~128, the fourth displayed equation on that page), we have that, conditioning on the data out of fold $v$,
        \begin{align*}
            & \sqrt{n} (P^{n,v}-P_*) \left\{ \left( \sum_{a \in \mathcal{S}'_1} \hat\pi^{a}_{v} \right) \widetilde{\mathcal{T}}(\widehat{\boldsymbol \ell}_{v}, \widehat{\boldsymbol \lambda}_{v}, \widehat{\boldsymbol \pi}_{v}) - \left( \sum_{a \in \mathcal{S}'_1} \pi^{a}_* \right) \widetilde{\mathcal{T}}(\boldsymbol \ell_*,\boldsymbol \lambda_*,\boldsymbol \pi_*) \right\} \\
            &= \bigO_p \left( \sum_{k=2}^K \left\{ \| \hat{\lambda}^{k-1}_v - \lambda^{k-1}_* \|_{L^2(P_*)} + \| \hat{\ell}^{k-1}_v - h^{k-1}_v \|_{L^2(P_*)} \right\} \right).
        \end{align*}
        Since $\frac{|I_v|}{n \sum_{a \in \mathcal{S}'_1} \hat{\pi}^a_v} = \bigO_p(1)$, for some constant $\const > 0$,
        \begin{align*}
            &\Prob \Bigg( \left| \frac{|I_v|}{n \sum_{a \in \mathcal{S}'_1} \hat{\pi}^a_v} (P^{n,v}-P_*) \left\{ \left( \sum_{a \in \mathcal{S}'_1} \hat\pi^{a}_{v} \right) \widetilde{\mathcal{T}}(\widehat{\boldsymbol \ell}_{v}, \widehat{\boldsymbol \lambda}_{v}, \widehat{\boldsymbol \pi}_{v}) - \left( \sum_{a \in \mathcal{S}'_1} \pi^{a}_* \right) \widetilde{\mathcal{T}}(\boldsymbol \ell_*,\boldsymbol \lambda_*,\boldsymbol \pi_*) \right\} \right| \\
            &\qquad> \const n^{-1/2} \sum_{k=2}^K (a_{n,k,v} + b_{n,k,v}) \text{ and $U^C$ occurs} \Bigg) \leq \frac{\epsilon}{3 V}.
        \end{align*}
        Therefore, for the second term on the right-hand side of \eqref{eq: eq r hat expansion2},
        \begin{align*}
            &\Prob \Bigg( \left| \sum_{v \in [V]} \frac{|I_v|}{n \sum_{a \in \mathcal{S}'_1} \hat{\pi}^a_v} (P^{n,v}-P_*) \left\{ \left( \sum_{a \in \mathcal{S}'_1} \hat\pi^{a}_{v} \right) \widetilde{\mathcal{T}}(\widehat{\boldsymbol \ell}_{v}, \widehat{\boldsymbol \lambda}_{v}, \widehat{\boldsymbol \pi}_{v}) - \left( \sum_{a \in \mathcal{S}'_1} \pi^{a}_* \right) \widetilde{\mathcal{T}}(\boldsymbol \ell_*,\boldsymbol \lambda_*,\boldsymbol \pi_*) \right\} \right| \\
            &\qquad> \const n^{-1/2} \sum_{v \in [V]} \sum_{k=2}^K (a_{n,k,v} + b_{n,k,v}) \text{ and $U^C$ occurs} \Bigg) \\
            &\leq \sum_{v \in [V]} \Prob \Bigg( \left| \frac{|I_v|}{n \sum_{a \in \mathcal{S}'_1} \hat{\pi}^a_v} (P^{n,v}-P_*) \left\{ \left( \sum_{a \in \mathcal{S}'_1} \hat\pi^{a}_{v} \right) \widetilde{\mathcal{T}}(\widehat{\boldsymbol \ell}_{v}, \widehat{\boldsymbol \lambda}_{v}, \widehat{\boldsymbol \pi}_{v}) - \left( \sum_{a \in \mathcal{S}'_1} \pi^{a}_* \right) \widetilde{\mathcal{T}}(\boldsymbol \ell_*,\boldsymbol \lambda_*,\boldsymbol \pi_*) \right\} \right| \\
            &\qquad\qquad> \const n^{-1/2} \sum_{k=2}^K (a_{n,k,v} + b_{n,k,v}) \text{ and $U^C$ occurs} \Bigg) \leq \frac{\epsilon}{3}.
        \end{align*}

    For the third term on the right-hand side of \eqref{eq: eq r hat expansion2}, since $|\sum_{k=2}^K B_{k,v}| \leq \const \sum_{k=2}^K \| \hat{\lambda}^{k-1}_v - \lambda^{k-1}_* \|_{L^2(P_*)} \| \hat{\ell}^{k-1}_v - h^{k-1}_v \|_{L^2(P_*)}$ and $\frac{|I_v| \sum_{a \in \mathcal{S}'_1} \hat{\pi}^a_v}{n} = \bigO_p(1)$, we have that
        $$\Prob \left( \left| \frac{|I_v| \sum_{a \in \mathcal{S}'_1} \hat{\pi}^a_v}{n} \sum_{k=2}^K B_{k,v} \right| > \const \sum_{k=2}^K a_{n,k,v} b_{n,k,v} \text{ and $U^C$ occurs} \right) \leq \frac{\epsilon}{3 V}$$
        and thus
        $$\Prob \left( \left| \sum_{v \in [V]}  \frac{|I_v| \sum_{a \in \mathcal{S}'_1} \hat{\pi}^a_v}{n} \sum_{k=2}^K B_{k,v} \right| > \const \sum_{v \in [V]} \sum_{k=2}^K a_{n,k,v} b_{n,k,v} \text{ and $U^C$ occurs} \right) \leq \frac{\epsilon}{3}.$$

    For the last term on the right-hand side of \eqref{eq: eq r hat expansion2}, since
        \begin{align*}
            \frac{|I_v|}{n} \left( \frac{ \sum_{a \in \mathcal{S}'_1} \pi^{a}_*}{\sum_{a \in \mathcal{S}'_1} \pi^{a}_v}-1 \right) (P^{n,v}-P_*) D_\generalindependence(\boldsymbol \ell_*,\boldsymbol \lambda_*,\boldsymbol \pi_*,r_*) = \bigO_p(n^{-1/2}) \bigO_p(n^{-1/2}) = \bigO_p(n^{-1}),
        \end{align*}
        we have that
        $$\Prob \left( \left| \frac{|I_v|}{n} \left( \frac{ \sum_{a \in \mathcal{S}'_1} \pi^{a}_*}{\sum_{a \in \mathcal{S}'_1} \pi^{a}_v}-1 \right) (P^{n,v}-P_*) D_\generalindependence(\boldsymbol \ell_*,\boldsymbol \lambda_*,\boldsymbol \pi_*,r_*) \right| > \frac{\const}{n} \text{ and $U^C$ occurs} \right) \leq \frac{\epsilon}{3V}$$
        and thus
        $$\Prob \left( \left| \sum_{v \in [V]} \frac{|I_v|}{n} \left( \frac{ \sum_{a \in \mathcal{S}'_1} \pi^{a}_*}{\sum_{a \in \mathcal{S}'_1} \pi^{a}_v}-1 \right) (P^{n,v}-P_*) D_\generalindependence(\boldsymbol \ell_*,\boldsymbol \lambda_*,\boldsymbol \pi_*,r_*) \right| > \frac{\const}{n} \text{ and $U^C$ occurs} \right) \leq \frac{\epsilon}{3}.$$

        Combining the above three bounds, for any $t > \const \sum_{v \in [V]} \sum_{k=2}^K (a_{n,k,v} b_{n,k,v} + a_{n,k,v} + b_{n,k,v} + n^{-1})$, we have that
        \begin{align*}
            &\Prob(|\hat{r} - \Delta - r_*| > t \text{ and $U^C$ occurs}) \\
            &\leq \Prob \Bigg( \left|\frac{1}{n} \sum_{i=1}^n D_\generalindependence(\boldsymbol \ell_*,\boldsymbol \lambda_*,\boldsymbol \pi_*,r_*)(O_i)\right| > \\
            &\qquad\qquad t - \const \sum_{v \in [V]} \sum_{k=2}^K (a_{n,k,v} b_{n,k,v} + n^{-1/2} (a_{n,k,v} + b_{n,k,v}) + n^{-1}) \text{ and $U^C$ occurs} \Bigg)  \\
            &\quad+ \Prob \Bigg( \left| \sum_{v \in [V]} \frac{|I_v|}{n \sum_{a \in \mathcal{S}'_1} \hat{\pi}^a_v} (P^{n,v}-P_*) \left\{ \left( \sum_{a \in \mathcal{S}'_1} \hat\pi^{a}_{v} \right) \widetilde{\mathcal{T}}(\widehat{\boldsymbol \ell}_{v}, \widehat{\boldsymbol \lambda}_{v}, \widehat{\boldsymbol \pi}_{v}) - \left( \sum_{a \in \mathcal{S}'_1} \pi^{a}_* \right) \widetilde{\mathcal{T}}(\boldsymbol \ell_*,\boldsymbol \lambda_*,\boldsymbol \pi_*) \right\} \right| \\
            &\qquad\qquad> \const n^{-1/2} \sum_{v \in [V]} \sum_{k=2}^K (a_{n,k,v} + b_{n,k,v}) \text{ and $U^C$ occurs} \Bigg) \\
            &\quad+ \Prob \left( \left| \sum_{v \in [V]}  \frac{|I_v| \sum_{a \in \mathcal{S}'_1} \hat{\pi}^a_v}{n} \sum_{k=2}^K B_{k,v} \right| > \const \sum_{v \in [V]} \sum_{k=2}^K a_{n,k,v} b_{n,k,v} \text{ and $U^C$ occurs} \right) \\
            &\quad+ \Prob \left( \left| \sum_{v \in [V]} \frac{|I_v|}{n} \left( \frac{ \sum_{a \in \mathcal{S}'_1} \pi^{a}_*}{\sum_{a \in \mathcal{S}'_1} \pi^{a}_v}-1 \right) (P^{n,v}-P_*) D_\generalindependence(\boldsymbol \ell_*,\boldsymbol \lambda_*,\boldsymbol \pi_*,r_*) \right| > \frac{\const}{n} \text{ and $U^C$ occurs} \right) \\
            &\leq \Prob \left( \left|\frac{1}{n} \sum_{i=1}^n D_\generalindependence(\boldsymbol \ell_*,\boldsymbol \lambda_*,\boldsymbol \pi_*,r_*)(O_i)\right| > t - \const \sum_{v \in [V]} \sum_{k=2}^K (a_{n,k,v} b_{n,k,v} + a_{n,k,v} + b_{n,k,v} + n^{-1}) \right) + \epsilon.
        \end{align*}
        If $\sigma_{*,\generalindependence}=0$, then $D_\generalindependence(\boldsymbol \ell_*,\boldsymbol \lambda_*,\boldsymbol \pi_*,r_*)(O)=0$ $P_*$-almost surely and
        \begin{align*}
            &\Prob \left( \left|\frac{1}{n} \sum_{i=1}^n D_\generalindependence(\boldsymbol \ell_*,\boldsymbol \lambda_*,\boldsymbol \pi_*,r_*)(O_i)\right| > t - \const \sum_{v \in [V]} \sum_{k=2}^K (a_{n,k,v} b_{n,k,v} + n^{-1/2} (a_{n,k,v} + b_{n,k,v}) + n^{-1}) \right) \\
            &= 0.
        \end{align*}
        Otherwise, the above bound equals
        \begin{align*}
            &\Prob \Bigg( \left|\frac{1}{\sqrt{n} \sigma_{*,\generalindependence}} \sum_{i=1}^n D_\generalindependence(\boldsymbol \ell_*,\boldsymbol \lambda_*,\boldsymbol \pi_*,r_*)(O_i)\right| \\
            &\qquad> \sqrt{n} \frac{t - \const \sum_{v \in [V]} \sum_{k=2}^K (a_{n,k,v} b_{n,k,v} + n^{-1/2} (a_{n,k,v} + b_{n,k,v}) + n^{-1})}{\sigma_{*,\generalindependence}} \Bigg) + \epsilon.
        \end{align*}
        Since $\frac{1}{\sqrt{n} \sigma_{*,\generalindependence}} \sum_{i=1}^n D_\generalindependence(\boldsymbol \ell_*,\boldsymbol \lambda_*,\boldsymbol \pi_*,r_*)(O_i)$ converges in distribution to the standard normal distribution, by $\expect_{P_*} |D_\independence(\boldsymbol \ell_*,\boldsymbol \theta_*,\boldsymbol \pi_*,r_*)(O)|^3 < \infty$ and Berry-Esseen Theorem, the above bound is less than or equal to
        \begin{align*}
            & 2 \left\{ 1- \Phi \left( \sqrt{n} \frac{t - \const \sum_{v \in [V]} \sum_{k=2}^K (a_{n,k,v} b_{n,k,v} + n^{-1/2} (a_{n,k,v} + b_{n,k,v}) + n^{-1})}{\sigma_{*,\generalindependence}} \right) \right\} + \frac{\const'}{\sqrt{n}} + \epsilon \\
            &= 2 \Phi \left( -\sqrt{n} \frac{t - \const \sum_{v \in [V]} \sum_{k=2}^K (a_{n,k,v} b_{n,k,v} + n^{-1/2} (a_{n,k,v} + b_{n,k,v}) + n^{-1})}{\sigma_{*,\generalindependence}} \right) + \frac{\const'}{\sqrt{n}} + \epsilon,
        \end{align*}
        where the last equality follows from the symmetry of the standard normal distribution.
        Therefore, regardless of whether $\sigma_{*,\generalindependence}$ is zero or not,
        \begin{align*}
            &\Prob(|\hat{r} - \Delta - r_*| > t \text{ and $U^C$ occurs}) \\
            &\leq 2 \Phi \left( -\sqrt{n} \frac{t - \const \sum_{v \in [V]} \sum_{k=2}^K (a_{n,k,v} b_{n,k,v} + n^{-1/2} (a_{n,k,v} + b_{n,k,v}) + n^{-1})}{\sigma_{*,\generalindependence}} \right) + \frac{\const'}{\sqrt{n}} + \epsilon.
        \end{align*}

    With the above results for the two complementary events $U$ and $U^C$, we have that
    \begin{align*}
        & \Prob(|\hat{r}-\Delta-r_*| > t) \\
        &= \Prob(|\hat{r}-\Delta-r_*| > t \text{ and $U$ occurs}) + \Prob(|\hat{r}-\Delta-r_*| > t \text{ and $U^C$ occurs}) \\
        &\leq \Prob(U) + \Prob(|\hat{r}-\Delta-r_*| > t \text{ and $U^C$ occurs}) \\
        &\leq 2 \Phi \left( -\sqrt{n} \frac{t - \const \sum_{v \in [V]} \sum_{k=2}^K (a_{n,k,v} b_{n,k,v} + n^{-1/2} (a_{n,k,v} + b_{n,k,v}) + n^{-1})}{\sigma_{*,\generalindependence}} \right) \\
        &\quad+ \frac{\const'}{\sqrt{n}} + \epsilon + \sum_{v \in [V]} \sum_{k=2}^K (c_{n,k,v} + d_{n,k,v}),
    \end{align*}
    namely \eqref{eq: r hat finite sample confidence} holds.
\end{proof}

Corollary~\ref{corollary: model comparison} is follows immediately from Theorem~\ref{thm: independence eff and DR}.

\subsection{Concept shift in the features} \label{sec: X con shift proof}

\begin{lemma} \label{lemma: X con shift tangent and orthogonal complement}
    Under Condition~\ref{DScondition: X con shift}, the tangent space is $\tangent_\xconshift := \tangent_X \oplus \tangent_A \oplus \tangent_{Y \mid X,A}$ and its orthogonal complement is $\tangent_\xconshift^\perp = \{ f \in L^2_0(P_{*,X,A}): \expect_{P_*}[f(X,A) \mid X] = \expect_{P_*}[f(X,A) \mid A] = 0 \}$.
\end{lemma}
This lemma follows from Lemma~\ref{lemma: orthogonality} and the definition of $\tangent_\xconshift^\perp$.

\begin{proof}[Proof of Corollary~\ref{corollary: X con shift EIF}]
Since Condition~\ref{DScondition: X con shift} is a special case of Condition~\ref{DScondition: independence} with $K=2$, $\mathcal{A}=\{0,1\}$, $\mathcal{S}_1=\{1\}$, $\mathcal{S}_2 = \emptyset$, $Z_1=X$, and $Z_2=Y$, Corollary~\ref{corollary: X con shift EIF} is thus a consequence of the efficiency bound under Condition~\ref{DScondition: independence} from \eqref{eq: independence EIF}. Specifically, we have that
$$\ell^1_* = \mathcal{E}_*, \qquad \pi^{0}_* = P_*(A=0) = \rho_* = 1-\pi^{1}_*, \qquad \theta^0_*=\frac{1-\rho_*}{\rho_*}.$$
The desired result follows by plugging the above equalities into $D_\independence$.

\end{proof}

\begin{proof}[Proof of Corollary~\ref{corollary: X con shift eff gain}]
    Since $D_\xconshift(\rho_*,\mathcal{E}_*,r_*) - D_\nonparametric(\rho_*,r_*)$ lies in $\tangent_\xconshift^\perp$ and is orthogonal to $D_\xconshift(\rho_*,\mathcal{E}_*,r_*)$ , we have that the relative efficiency gain is
    $$1-\frac{P_* D_\xconshift(\rho_*,\mathcal{E}_*,r_*)^2}{P_* D_\nonparametric(\rho_*,r_*)^2} = \frac{P_* \{ D_\nonparametric(\rho_*,r_*) - D_\xconshift(\rho_*,\mathcal{E}_*,r_*) \}^2}{P_* D_\nonparametric(\rho_*)^2}.$$
    By direct calculation, we find that
    \begin{align*}
        & D_\nonparametric(\rho_*,r_*)(o) - D_\xconshift(\rho_*,\mathcal{E}_*,r_*)(o) = \frac{1}{\rho_*} (a-1+\rho_*) (\mathcal{E}_*(x) - r_*), \\
        & P_* \{ D_\nonparametric(\rho_*,r_*) - D_\xconshift(\rho_*,\mathcal{E}_*,r_*) \}^2 = \frac{1-\rho_*}{\rho_*} \expect_{P_*} \left[ (\mathcal{E}_*(X) - r_*)^2 \right], \\
        & P_* D_\nonparametric(\rho_*,r_*)^2 = \frac{1}{\rho_*} \left\{ \expect_{P_*} \left[ \expect_{P_*} \left[ \{ \ell(X,Y) - \mathcal{E}_*(X) \}^2 \mid A=0,X \right] \right] + \expect_{P_*} \left[ \{ \mathcal{E}_*(X) - r_* \}^2 \right] \right\}.
    \end{align*}
    Corollary~\ref{corollary: X con shift eff gain} follows.
\end{proof}

\begin{proof}[Proof of Theorem~\ref{thm: X con shift efficiency robust}]
    We focus on $\hat{r}_{\xconshift}^v$ for a fixed fold $v$ as the desired property of $\hat{r}_{\xconshift}$ follows immediately. We first condition on the out-of-fold data.
    It follows directly that
    $$P_* D_\xconshift(\rho_*, \mathcal{E}_\infty,r_*) = 0 \quad \text{and} \quad P^{n,v} D_\xconshift(\hat{\rho}^v,\hat{\mathcal{E}}^{-v},\hat{r}_{\xconshift}^v) = 0.$$
    Then we have the following equation:
    \begin{align}
    \begin{split}
        0 &= (P^{n,v}-P_*) \left\{ D_\xconshift(\hat{\rho}^{-v},\hat{\mathcal{E}}^{-v},\hat{r}_{\xconshift}^v) - D_\xconshift(\rho_*,\mathcal{E}_\infty,r_*) \right\} \\
        &\quad+ P_* \left\{ D_\xconshift(\hat{\rho}^{-v},\hat{\mathcal{E}}^{-v},\hat{r}_{\xconshift}^v) - D_\xconshift(\rho_*,\mathcal{E}_\infty,r_*) \right\} \\
        &\quad+ (P^{n,v}-P_*) D_\xconshift(\rho_*,\mathcal{E}_\infty,r_*).
    \end{split} \label{eq: X con shift Z expanssion}
    \end{align}
    We next study the first two terms on the right-hand side of \eqref{eq: X con shift Z expanssion}.
    
    \textbf{Term~2:} 
    Using the definition 
    of $D_\xconshift$ from 
    \eqref{eq: X con shift EIF},
    direct calculation shows that term~2 equals
    \begin{align*}
    & P_* \Bigg\{ \left[ \frac{1-A}{\hat{\rho}^v} - \frac{1-A}{\rho_*} \right]\ell(X,Y) - \left[ \frac{1-A}{\hat{\rho}^v} \hat{\mathcal{E}}^{-v}(X) - \frac{1-A}{\rho_*} \mathcal{E}_\infty(X) \right] \\
    &\quad\quad+ [\hat{\mathcal{E}}^{-v}(X) - \mathcal{E}_\infty(X)] - [\hat{r}_{\xconshift}^v - r_*] \Bigg\}.
     \end{align*}
Conditioning on $A=0$,
and using that
$\mathcal{E}_*: x \mapsto \expect_{P_*}[\ell(X,Y) \mid X=x,A=0]$, this further equals
\begin{align*}
    & P_* \left\{ \left[ \frac{\rho_*}{ \hat{\rho}^v} - 1 \right] \mathcal{E}_* - \left[ \frac{\rho_*}{\hat{\rho}^v} \hat{\mathcal{E}}^{-v} - \mathcal{E}_\infty \right] + [\hat{\mathcal{E}}^{-v} - \mathcal{E}_\infty] \right\} - (\hat{r}_{\xconshift} - r_*) \\
    &= \frac{\hat{\rho}^v-\rho_*}{\hat{\rho}^v} P_* (\hat{\mathcal{E}}^{-v} - \mathcal{E}_*) - (\hat{r}_{\xconshift}^v - r_*) \\
    &= \frac{\hat{\rho}^v-\rho_*}{\hat{\rho}^v} P_* (\hat{\mathcal{E}}^{-v} - \mathcal{E}_\infty) + \frac{\hat{\rho}^v-\rho_*}{\hat{\rho}^v} P_* (\mathcal{E}_\infty - \mathcal{E}_*) - (\hat{r}_{\xconshift}^v - r_*).
    \end{align*}
    Since $\hat{\rho}^v = \rho_* + (P^{n,v} - P_*) ( 1-A - \rho_* ) = \rho_* + \bigO_p(n^{-1/2})$ and $\| \hat{\mathcal{E}}^{-v} - \mathcal{E}_\infty \|_{L^2(P_*)} = \smallo_p(1)$, term~2 further equals
    \begin{align*}
        & \frac{P_* (\mathcal{E}_\infty - \mathcal{E}_*)}{\rho_*} (P^{n,v} - P_*) (1 - A - \rho_*) - (\hat{r}_{\xconshift}^v - r_*) + \frac{\hat{\rho}^v-\rho_*}{\hat{\rho}^v} P_* (\hat{\mathcal{E}}^{-v} - \mathcal{E}_\infty) \\
        &= (P^{n,v} - P_*) \left\{ \frac{P_* \mathcal{E}_\infty - r_*}{\rho_*} (1 - A - \rho_*) \right\} - (\hat{r}_{\xconshift}^v - r_*) + \frac{\hat{\rho}^v-\rho_*}{\hat{\rho}^v} P_* (\hat{\mathcal{E}}^{-v} - \mathcal{E}_\infty),
    \end{align*}
    where $\frac{\hat{\rho}^v-\rho_*}{\hat{\rho}^v} P_* (\hat{\mathcal{E}}^{-v} - \mathcal{E}_\infty) = \smallo_p(n^{-1/2})$.

    \textbf{Term~1:} 
        Using the definition 
    of $D_\xconshift$ from 
    \eqref{eq: X con shift EIF}, clearly,
    \begin{align*}
        & D_\xconshift(\hat{\rho}^v,\hat{\mathcal{E}}^{-v},\hat{r}_{\xconshift}^v)(o) - D_\xconshift(\rho_*,\mathcal{E}_\infty,r_*)(o) \\
        &= \left[ \frac{1-a}{\hat{\rho}^v} - \frac{1-a}{\rho_*} \right]\ell(x,y) - \left[ \frac{1-a}{\hat{\rho}^v} \hat{\mathcal{E}}^{-v}(x) - \frac{1-a}{\rho_*} \mathcal{E}_\infty(x) \right] \\
        &\quad+ [\hat{\mathcal{E}}^{-v}(x) - \mathcal{E}_\infty(x)] - [\hat{r}_{\xconshift} - r_*].
    \end{align*}
    Define
    $$f(\rho,\mathcal{E}): o \mapsto \left[ \frac{1-a}{\rho} - \frac{1-a}{\rho_*} \right]\ell(x,y) - \left[ \frac{1-a}{\rho} \mathcal{E}(x) - \frac{1-a}{\rho_*} \mathcal{E}_\infty(x) \right] + [\mathcal{E}(x) - \mathcal{E}_\infty(x)].$$
    Since $\hat{r}_{\xconshift} - r_*$ is a scalar, term~1 equals $(P^{n,v}-P_*) f(\hat{\rho}^v,\mathcal{E}^{-v})$.
    For any positive sequence $\{\delta\}_{n \geq 1}$ such that $\sqrt{n} \delta \to \infty$ and $\delta \to 0$, the function class
    $$\funclass_{\delta} := \left\{ f(\rho,\mathcal{E}^{-v}): \rho \in [\rho_* - \delta, \rho_* + \delta] \right\}$$
    contains $f(\hat{\rho}^v,\mathcal{E}^{-v})$
    with probability tending to one. By Lemmas~9.6 and 9.9 in \cite{Kosorok2008}, this function class is VC-subgraph and thus of bounded uniform entropy integral.
    An envelope of $\funclass_{\delta}$ is
    \begin{align*}
        o &\mapsto (1-a) \left[ \frac{1}{\rho_*-\delta} - \frac{1}{\rho_*} \right] \ell(x,y)
        + \frac{1-a}{\rho_*-\delta} \left| \hat{\mathcal{E}}^{-v}(x) - \mathcal{E}_\infty(x) \right| \\
        &+ (1-a) \left[ \frac{1}{\rho_*-\delta} - \frac{1}{\rho_*} \right] \left| \mathcal{E}_\infty(x) \right|
        + \left| \hat{\mathcal{E}}^{-v}(x) - \mathcal{E}_\infty(x) \right|,
    \end{align*}
    whose $L^2(P_*)$-norm is of order $\delta + \left\|  \hat{\mathcal{E}}^{-v} - \mathcal{E}_\infty \right\|_{L^2(P_*)}$. Therefore, by the proof of Theorem~2.5.2 in \citetsupp{vandervaart1996} (specifically Line~17 on page~128, the fourth displayed equation on that page), we have that term~1 is upper bounded by $n^{-1/2} \left( \delta + \left\| \hat{\mathcal{E}}^{-v} - \mathcal{E}_\infty \right\|_{L^2(P_*)} \right)$ up to a multiplicative constant, conditional on the event that $\funclass_{\delta}$ contains $f(\hat{\rho}^v,\mathcal{E}^{-v})$. Since this event occurs with probability tending to one, and, $\left\| \hat{\mathcal{E}}^{-v} - \mathcal{E}_\infty \right\|_{L^2(P_*)} = \smallo_p(1)$ with randomness in the out-of-fold data accounted for, we conclude that term~1 is $\smallo_p(n^{-1/2})$ unconditionally by Lemma~6.1 in \citetsupp{Chernozhukov2018debiasedML}.

    With the above results for the terms in \eqref{eq: X con shift Z expanssion}, we conclude that
    \begin{align*}
        &\hat{r}_{\xconshift}^v - r_* \\
        &= (P^{n,v} - P_*) \left\{ D_\xconshift(\rho_*,\mathcal{E}_\infty,r_*)(O) + \frac{P_* \mathcal{E}_\infty - r_*}{\rho_*} (1 - A - \rho_*) \right\} \\
        &\quad+ \frac{\hat{\rho}^v-\rho_*}{\hat{\rho}^v} P_* (\hat{\mathcal{E}}^{-v} - \mathcal{E}_\infty) + (P^{n,v}-P_*) \left\{ D_\xconshift(\hat{\rho}^{-v},\hat{\mathcal{E}}^{-v},\hat{r}_{\xconshift}^v) - D_\xconshift(\rho_*,\mathcal{E}_\infty,r_*) \right\} \\
        &= (P^{n,v} - P_*) \left\{ D_\xconshift(\rho_*,\mathcal{E}_\infty,r_*)(O) + \frac{P_* \mathcal{E}_\infty - r_*}{\rho_*} (1 - A - \rho_*) \right\} + \smallo_p(n^{-1/2}).
    \end{align*}
    Then \eqref{eq: X con shift robust AL} follows from the definition of $\hat{r}_{\xconshift}$ in Alg.~\ref{alg: X con shift estimator};
    \eqref{eq: X con shift eff} follows from \eqref{eq: X con shift robust AL} by replacing $\mathcal{E}_\infty$ with $\mathcal{E}_*$.
    We note that the influence function of $\hat{r}_{\xconshift}$ in \eqref{eq: X con shift robust AL} equals $D_\xconshift(\rho_*,\mathcal{E}_*,r_*)$ plus
    $$o \mapsto \left( 1 - \frac{1-a}{\rho_*} \right) \left\{ \mathcal{E}_\infty(x) - P_* \mathcal{E}_\infty - \mathcal{E}_*(x) + r_* \right\},$$
    which, since $r_* = P_* \mathcal{E}_*$, is an element in $\tangent_\xconshift^\perp$ from Lemma~\ref{lemma: X con shift tangent and orthogonal complement}, and so the influence function of $\hat{r}_{\xconshift}$ is a gradient. The assertion of the regularity and asymptotic efficiency follows from Lemma~2.9 in Part~III of \cite{Bolthausen2002}.
\end{proof}

\subsection{Full-data covariate shift} \label{sec: cov shift proof}

We first characterize the tangent space and its orthogonal complement in the following lemma.

\begin{lemma}[Tangent space and orthogonal complement under full-data covariate shift] \label{lemma: cov shift tangent and orthogonal complement}
    In the semiparametric model under Condition~\ref{DScondition: cov shift}, the tangent space at $P_*$ is $\tangent_\covshift := \tangent_X \oplus \tangent_{A \mid X} \oplus \tangent_{Y \mid X}$ and its orthogonal complement is
    $$\tangent_\covshift^\perp = \left\{ o \mapsto \begin{cases}
    0 & \text{if } g_*(x)=0 \\
    (1-a-g_*(x)) \frac{b(x,y)}{g_*(x)} & \text{if } g_*(x) \neq 0
    \end{cases}: \text{ $b$ such that } \expect_{P_*}[b(X,Y) \mid X] = 0
    \right\}.$$
\end{lemma}
\begin{proof}
    By Condition~\ref{DScondition: cov shift}, $\tangent_{Y \mid X,A}=\tangent_{Y \mid X}$ and therefore the tangent space is $\tangent_\covshift$.
    
    Next, we characterize  the orthogonal complement $\tangent_\covshift^\perp$. By definition, $\tangent_\covshift^\perp = \{ f \in L^2_0(P_*): P_* fg=0 \text{ for all } g \in \tangent_\covshift \} = \{ f \in \tangent_{A,Y \mid X}: \expect_{P_*}[f(O) \mid X,Y] = \expect_{P_*}[f(O) \mid X,A] = 0 \}$. Since $A$ is binary, any function $f$ of $o$ can be written as $f(o) = a f_1(x,y) + (1-a) f_*(x,y)$ for two functions $f_*$ and $f_1$. 
    If $f: o \mapsto a f_1(x,y) + (1-a) f_*(x,y)$ lies in $\tangent_\covshift^\perp$, then
    recalling that 
    $g_*(x)=P_*(A=1 \mid X=x)$ 
    from \Cref{sec: cov shift EIF},
     \begin{align*}
        0 &= \expect_{P_*}[f(O) \mid X=x,A=a] = \expect_{P_*}[f_a(X,Y) \mid X=x], \\
        0 &= \expect_{P_*}[f(O) \mid X=x,Y=y] = (1-g_*(x)) f_1(x,y) + g_*(x) f_*(x,y).
    \end{align*}
    If $g_*(x)=0$, 
    then
    $A=1$ with probability one conditional on $X=x$, and the above equations are equivalent to $f_1(x,y)= 0$. In this case, $f(o)=0$.
    
    If $g_*(x) \neq 0$, the above equations are equivalent to $\expect_{P_*}[f_1(X,Y) \mid X] = 0$
    and $f_*(x,y) = - \frac{1-g_*(x)}{g_*(x)} f_1(x,y)$. In this case,
    \begin{align*}
        f(o) &= a f_1(x,y) - (1-a) \frac{1-g_*(x)}{g_*(x)} f_1(x,y) 
        = (1-a-g_*(x)) \frac{-f_1(x,y)}{g_*(x)}
    \end{align*}
    for $f_1$ such that $\expect_{P_*}[f_1(X,Y) \mid X] = 0$.
    The claimed form of $\tangent_\covshift^\perp$ follows by replacing the notation $-f_1$ with $b$.
\end{proof}

\begin{proof}[Proof of Corollary~\ref{corollary: cov shift EIF}]
     Condition~\ref{DScondition: cov shift} is a special case of Condition~\ref{DScondition: independence} with $K=2$, $\mathcal{A}=\{0,1\}$, $\mathcal{S}_1 = \emptyset$, $\mathcal{S}_2=\{1\}$, $Z_1=X$, and $Z_2=Y$. Thus, Corollary~\ref{corollary: cov shift EIF} follows from the efficiency bound from \eqref{eq: independence EIF} under Condition~\ref{DScondition: independence}. Specifically,
     \begin{align*}
        & \ell^1_* = \mathcal{L}_*, \qquad \rho^{1,0} = g_* = 1 - \rho^{1,1}_*, \qquad \pi^{0}_* = \rho_* = 1 - \pi^{1}, \qquad \theta^1_* = \frac{1-g_*}{g_*}.
     \end{align*}
     Plugging the above expressions into $D_\independence$ in \eqref{eq: independence EIF}, we have shown that $D_\covshift(\rho_*,g_*,\mathcal{L}_*,r_*)$ is the efficient influence function.
     
\end{proof}

Corollaries~\ref{corollary: cov shift efficiency} and \ref{corollary: cov shift DR consistent} are implied by Theorem~\ref{thm: independence eff and DR}.

\begin{proof}[Proof of Lemma~\ref{lemma: cov shift no eff and robust}]
    We first consider the parametrization $(P_X,P_{A \mid X},P_{Y \mid X})$. The influence function $\IF(P_{*,X},P_{A \mid X},P_{Y \mid X},r_*)$ must be a sum of $D_\nonparametric(\rho_*,r_*)$ and an element in the orthogonal complement $\tangent_\covshift^\perp$ of the tangent space. Without further conditions on $P_*$, we may consider a data-generating distribution $P_*$ such that $g_*(x) \in (0,1)$. 
    By Lemma~\ref{lemma: cov shift tangent and orthogonal complement} and the condition on $\IF(P_{*,X},P_{A \mid X},P_{Y \mid X},r_*)$, 
    there is a function  $b$ of $(P_{*,X},P_{A \mid X},P_{Y \mid X},r_*)$,
    such that
    $$\expect_{P_*}[b(P_{*,X},P_{A \mid X},P_{Y \mid X},r_*)(X,Y) \mid X]=0$$
    and 
    such that for, any  $(P_{A \mid X},P_{Y \mid X})$,
    \begin{equation}
        \IF(P_{*,X},P_{A \mid X},P_{Y \mid X},r_*)(o) = D_\nonparametric(\rho_*,r_*)(o) + (1-a-g_*(x)) \frac{b(P_{*,X},P_{A \mid X},P_{Y \mid X},r_*)(x,y)}{g_*(x)}. \label{eq: cov shift no eff and robust}
    \end{equation}
    Since the left-hand side of \eqref{eq: cov shift no eff and robust} and $b(P_{*,X},P_{A \mid X},P_{Y \mid X},r_*)$ 
    are both variationally independent of $g_*$, \eqref{eq: cov shift no eff and robust} can only hold if $b$ is the constant zero function $(x,y) \mapsto 0$. Thus, $\IF(P_{*,X},P_{A \mid X},P_{Y \mid X},r_*) = D_\nonparametric(\rho_*,r_*)$.

    The argument for the other parametrization $(P_A,P_{X \mid A},P_{Y \mid X})$ is similar. 
    By Bayes' Theorem,
    $$g_*(x) = \frac{\rho_* P_{*,X \mid A}(x \mid 0)}{\rho_* P_{*,X \mid A}(x \mid 0) + (1-\rho_*) P_{*,X \mid A}(x \mid 1)} = \frac{\rho_* w_*(x)}{\rho_* w_*(x) + 1-\rho_*}$$
    where $P_{*,X \mid A}$ denotes the density of $X \mid A$ under $P_*$ with respect to some measure, and $w_*: x \mapsto P_{*,X \mid A}(x \mid 0)/P_{*,X \mid A}(x \mid 1)$ is a likelihood ratio.
    Using the same argument as above, we have that $\IF(P_A^0,P_{X \mid A},P_{Y \mid X},r_*)(o)$ equals the right-hand side of \eqref{eq: cov shift no eff and robust}.
    Since $\IF(P_A^0,P_{X \mid A},P_{Y \mid X},r_*)(o)$ and $b(P_{*,X},P_{A \mid X},P_{Y \mid X},r_*)$ are both variationally independent of $w_*$, we also conclude that $b$ must be the constant zero function and thus $\IF(P_A^0,P_{X \mid A},P_{Y \mid X},r_*) = D_\nonparametric(\rho_*,r_*)$.
\end{proof}

\subsection{Posterior drift conditions} \label{sec: posterior drift proof}

\begin{proof}[Proof of Theorems~\ref{thm: posterior drift EIF} and \ref{thm: posterior drift EIF2}]
    We focus on proving the result Theorem~\ref{thm: posterior drift EIF} for the weaker condition~\ref{DScondition: cov shift posterior drift}. The argument for the other condition is similar. Our strategy of proof is to first characterize the tangent space $\tangent_{\covshiftpostdrift}$ under Condition~\ref{DScondition: cov shift posterior drift}, then characterize its orthogonal complement $\tangent_{\covshiftpostdrift}^\perp$, and finally find the efficient influence function by projecting $D_\nonparametric(\rho_*, r_*)$ onto $\tangent_{\covshiftpostdrift}$. When calculating the tangent space, we focus on the density function itself and use Lemma~1.8 in Part~III of \cite{Bolthausen2002} to obtain differentiability in square-root of measures.

    \noindent\textbf{Characterization of tangent space:} Since both $A$ and $Y$ are binary, the density $p$ of a general distribution $P$ satisfying Condition~\ref{DScondition: cov shift posterior drift} can be written as
    \begin{align*}
        p(o) &= p_X(x) \times p_{A \mid X}(a \mid x) \times \expit(\phi_P\circ\theta_P(x))^{ay} \{1 - \expit(\phi_P\circ\theta_P(x))\}^{a(1-y)} \\
        &\times  \expit(\theta_P(x))^{(1-a)y} (1-\expit(\theta_P(x)))^{(1-a)(1-y)}
    \end{align*}
    where $p_X$ is the density of $X$ with respect to some dominating measure, $p_{A \mid X}(a \mid x) = \Prob_P(A=a \mid X=x)$, $\theta_P := \logit \Prob_{P}(A=1 \mid X=x)$, and $\phi_P$, corresponding to  $\phi_*$ for $P_*$, is the assumed transformation between the log odds of $Y \mid A,X$ for the distribution $P$. Consider a regular path $\{P^\epsilon: \epsilon \text{ in a neighborhood of } 0\}$ that passes through $P_*$ at $\epsilon=0$. 
    We use a superscript or subscript $\epsilon$ to denote components of $P^\epsilon$. The score function at $\epsilon=0$ is
    \begin{align*}
        o &\mapsto \left. \frac{\partial \log p^\epsilon(o)}{\partial \epsilon} \right|_{\epsilon=0} 
        = \left. \frac{\partial \log p^\epsilon_X(x)}{\partial \epsilon} \right|_{\epsilon=0} + \left. \frac{\partial \log p^\epsilon_{A \mid X}(a \mid x)}{\partial \epsilon} \right|_{\epsilon=0} \\
        &\quad+ ay \left\{ 1-\expit \circ \phi_* \circ \theta_*(x) \right\} \left\{ \left. \frac{\partial \phi_\epsilon(\theta_*(x))}{\partial \epsilon} \right|_{\epsilon=0} + \phi_*'\circ\theta_*(x) \left. \frac{\partial \theta_\epsilon(x)}{\partial \epsilon} \right|_{\epsilon=0}\right\} \\
        &\quad- a(1-y) \expit \circ \phi_* \circ \theta_*(x) \left\{ \left. \frac{\partial \phi_\epsilon(\theta_*(x))}{\partial \epsilon} \right|_{\epsilon=0} + \phi_*'\circ\theta_*(x) \left. \frac{\partial \theta_\epsilon(x)}{\partial \epsilon} \right|_{\epsilon=0}\right\} \\
        &\quad+ (1-a)y \left\{ 1 - \expit \circ \theta_*(x) \right\} \left. \frac{\partial \theta_\epsilon(x)}{\partial \epsilon}\right|_{\epsilon=0} 
        - (1-a)(1-y) \expit \circ \theta_*(x) \left. \frac{\partial \theta_\epsilon(x)}{\partial \epsilon}\right|_{\epsilon=0}.
    \end{align*}
    This further equals
    \begin{align*}
        o &\mapsto \left. \frac{\partial \log p^\epsilon_X(x)}{\partial \epsilon} \right|_{\epsilon=0} + \left. \frac{\partial \log p^\epsilon_{A \mid X}(a \mid x)}{\partial \epsilon} \right|_{\epsilon=0} \\
        &\quad+ a \left\{ y-\expit \circ \phi_* \circ \theta_*(x)  \right\} 
        \times \left\{ \left. \frac{\partial \phi_\epsilon(\theta_*(x))}{\partial \epsilon} \right|_{\epsilon=0} + \phi_*'\circ\theta_*(x) \left. \frac{\partial \theta_\epsilon(x)}{\partial \epsilon} \right|_{\epsilon=0}\right\} \\
        &\quad+ (1-a) \left\{ y - \expit \circ \theta_*(x)  \right\} \left. \frac{\partial \theta_\epsilon(x)}{\partial \epsilon}\right|_{\epsilon=0}.
    \end{align*}
    In the above display, the first two terms are arbitrary elements of $\tangent_{X}$ and $\tangent_{A \mid X}$, respectively.
    Therefore, with the notations defined in Theorem~\ref{thm: posterior drift EIF}, the tangent space $\tangent_{\covshiftpostdrift} \subseteq \tangent_{X} \oplus \tangent_{A \mid X} \oplus \tangent_{\covshiftpostdrift, Y \mid X,A}$
    where $\tangent_{\covshiftpostdrift, Y \mid X,A}$ is the $L^2_0(P_*)$-closure of
    \begin{align}
    \begin{split}
        \tilde{\mathcal{F}} &:= \Bigg\{ G=F(\tilde{\lambda}_1\circ\theta_* + \phi_*'\circ\theta_*\cdot \tilde{\lambda}_2,\,\tilde{\lambda}_2): 
        \tilde{\lambda}_1,\,\tilde{\lambda}_2 \text{ such that } 
        G \in L^2(P_*) \Bigg\}.
    \end{split} \label{eq: covshiftpostdrift Y|X,A tangent set}
    \end{align}

    We next argue that for any given bounded function $s \in \tangent_{X} \oplus \tangent_{A \mid X} \oplus \tilde{\mathcal{F}}$, there exists a regular path $\{P^\epsilon\}$ for $\epsilon$ sufficiently small such that the score function equals $s$. Let $s = s_X + s_{A \mid X} + s_{Y \mid X,A}$ such that $s_X \in \tangent_{X}$, $s_{A \mid X} \in \tangent_{A \mid X}$ and $s_{Y \mid X,A} = F(\tilde{\lambda}_1\circ\theta_* + \phi_*'\circ\theta_*\cdot \tilde{\lambda}_2,\,\tilde{\lambda}_2)$. For all $\epsilon$ sufficiently close to zero, let the density of $P^\epsilon$ be
    \begin{align*}
        &o \mapsto P_{*,X}(x) \left\{ 1 + \epsilon s_X(x) \right\} \times P_{*,A \mid X}(a \mid x) \left\{ 1 + \epsilon s_{A \mid X}(a \mid x) \right\} \\
        &\times \expit \left\{ (\phi_* + \epsilon \tilde{\lambda}_1) \circ (\theta_* + \epsilon \tilde{\lambda}_2)(x) \right\}^{ay} \times \left[ 1 - \expit \left\{ (\phi_* + \epsilon \tilde{\lambda}_1) \circ (\theta_* + \epsilon \tilde{\lambda}_2)(x) \right\} \right]^{a (1-y)} \\
        &\times \expit \left\{ (\theta_* + \epsilon \tilde{\lambda}_2)(x) \right\}^{(1-a) y} \times \left[ 1 - \expit \left\{ (\theta_* + \epsilon \tilde{\lambda}_2)(x) \right\} \right]^{(1-a) (1-y)}.
    \end{align*}
    It is straightforward to verify that the score function of $P^\epsilon$ at $\epsilon=0$ equals $s$ and $\{P^\epsilon\}$ is a regular path passing through $P_*$ at $\epsilon=0$. Therefore, we have shown that $\tangent_{\covshiftpostdrift} \supseteq \tangent_{X} \oplus \tangent_{A \mid X} \oplus \tangent_{\covshiftpostdrift, Y \mid X,A}$, which yields that $\tangent_{\covshiftpostdrift} = \tangent_{X} \oplus \tangent_{A \mid X} \oplus \tangent_{\covshiftpostdrift, Y \mid X,A}$.

    \noindent\textbf{Characterization of orthogonal complement of tangent space:} Since the tangent space under a nonparametric model is the entire space $L^2_0(P_*) = \tangent_{X} \oplus \tangent_{A \mid X} \oplus \tangent_{Y \mid X,A}$ by Lemma~\ref{lemma: orthogonality}, $\tangent_{\covshiftpostdrift}$ only restricts the subspace $\tangent_{Y \mid X,A}$, and thus $\tangent_{\covshiftpostdrift}^\perp$ is also a subspace of $\tangent_{Y \mid X,A}$. Since $A$ and $Y$ are binary, any function in $\tangent_{Y \mid X,A}$ can be written as $F(f_1,f_2)$ for some functions $f_1, f_2 \in L^2(P_{*,X})$. Therefore, $\tangent_{\covshiftpostdrift}^\perp$ consists of all functions in $\tangent_{Y \mid X,A}$ with zero covariance inner product with functions in $\tangent_{\covshiftpostdrift,Y\mid X,A}$, namely, with $P_1(X) = \Prob(A=1 \mid X)$
    \begin{align*}
        & \tangent_{\covshiftpostdrift}^\perp
        = \left\{ f \in \tangent_{Y \mid X,A}: \expect[ f(O) g(O) ]=0, \text{ for all } g \in \tangent_{\covshiftpostdrift,Y \mid X,A} \right\}.
    \end{align*}
    Now, we can write
    \begin{align*}
        &\expect[ f(O) g(O) ]= \expect[F(f_1,f_2) F(\tilde{\lambda}_1\circ\theta_* + \phi_*'\circ\theta_*\cdot \tilde{\lambda}_2,\,\tilde{\lambda}_2)]\\
        &= \expect[
        ( A \left\{ Y-\E(Y|A=1,X)  \right\} f_1(X) + (1-A) \left\{ Y - \E(A|Y=0,X)  \right\} f_2(X))\\
        &\times ( A \left\{ Y-\E(Y|A=1,X)  \right\} [\tilde{\lambda}_1\circ\theta_* + \phi_*'\circ\theta_*\cdot \tilde{\lambda}_2](X) + (1-A) \left\{ Y - \E(A|Y=0,X)  \right\} \tilde{\lambda}_2(X))]\\
        &= \expect[
        ( A \left\{ Y-\E(Y|A=1,X)  \right\}^2 f_1(X)[\tilde{\lambda}_1\circ\theta_* + \phi_*'\circ\theta_*\cdot \tilde{\lambda}_2](X)\\
        &\qquad+ (1-A) \left\{ Y - \E(A|Y=0,X)  \right\}^2 f_2(X)\tilde{\lambda}_2(X)].
    \end{align*}
    Therefore,
        \begin{align*}
        &\Bigg\{ F(f_1,f_2): \expect \Bigg[ \Var(Y \mid A=1,X) f_1(X) \left\{ \tilde{\lambda}_1\circ\theta_*(X) + \phi_*'\circ\theta_*(X)\cdot \tilde{\lambda}_2(X) \right\} P_1(X) \\
        &\qquad+ \Var(Y \mid A=0,X) f_2(X) \tilde{\lambda}_2(X) (1-P_1(X)) \Bigg] =0,
        \quad\text{for all functions } \tilde{\lambda}_1,\,\tilde{\lambda}_2 \Bigg\},
    \end{align*}
    where the term ``all functions'' should be interpreted as in \eqref{eq: covshiftpostdrift Y|X,A tangent set}. With the notations defined in Theorem~\ref{thm: posterior drift EIF}, since $\tilde{\lambda}_1$ and $\tilde{\lambda}_2$ are arbitrary, we have that
    \begin{equation}
        \tangent_{\covshiftpostdrift}^\perp = \left\{ F(f_1,f_2):\, \expect[v(X;1)\cdot  f_1(X) \mid \theta_*(X)] = 0, f_2(x) = - \frac{v(x;1) \cdot\phi_*'\circ\theta_*(x) \cdot f_1(x) }{v(x;0)} \right\}. \label{eq: covshiftpostdrift tangent orthogonal}
    \end{equation}

    \noindent\textbf{Characterization of efficient inference function:} By Lemma~\ref{lemma: orthogonality}, we can decompose $D_\nonparametric(\rho_*,r_*)$ as
    $$D_\nonparametric(\rho_*,r_*)(o) = \underbrace{\frac{1-a}{\rho_*} \left\{ \ell(x,y) - \mathcal{E}_*(x) \right\} }_{\text{projection of $D_\nonparametric(\rho_*,r_*)$ onto $\tangent_{Y \mid X,A}$}} + \underbrace{ \frac{1-a}{\rho_*} \left\{ \mathcal{E}_*(x) - r_* \right\}}_{\text{projection of $D_\nonparametric(\rho_*,r_*)$ onto $\tangent_{A,X}$}}.$$
    Thus, the projection of $D_\nonparametric(\rho_*,r_*)$ onto $\tangent_\covshiftpostdrift$ is the sum of (i) the projection $ o\mapsto \frac{1-a}{\rho_*} \left\{ \mathcal{E}_*(x) - r_* \right\}$ of $D_\nonparametric(\rho_*,r_*)$ onto $\tangent_{A,X} = \tangent_{X} \oplus \tangent_{A \mid X}$, and (ii) the projection of $o \mapsto \frac{1-a}{\rho_*} \left\{ \ell(x,y) - \mathcal{E}_*(x) \right\}$ (i.e., of the projection of $D_\nonparametric(\rho_*,r_*)$ onto $\tangent_{Y \mid X,A}$) onto $D_{\covshiftpostdrift,Y\mid X, A}$.
    
    We now characterize the second projection. We first write $o \mapsto \frac{1-a}{\rho_*} \left\{ \ell(x,y) - \mathcal{E}_*(x) \right\}$ in the form of $F(f_1,f_2)$ by noting that $Y$ is binary:
    \begin{align*}
        & \frac{1-a}{\rho_*} \left\{ \ell(x,y) - \mathcal{E}_*(x) \right\} \\
        &= \frac{1-a}{\rho_*} \left\{ y \ell(x,1) + (1-y) \ell(x,0) - \expit \circ \theta_*(x) \ell(x,1) - \left\{ 1-\expit \circ \theta_*(x) \right\} \ell(x,0) \right\} \\
        &= (1-a) \left\{ y - \expit \circ \theta_*(x)  \right\} \frac{\ell(x,1)-\ell(x,0)}{\rho_*} 
        = F \left( 0, \frac{\ell(\cdot,1)-\ell(\cdot,0)}{\rho_*} \right)(o).
    \end{align*}
    
    Suppose that the projection of $F( 0, \{\ell(\cdot,1)-\ell(\cdot,0)\}/\rho_* )$ onto $\tangent_{\covshiftpostdrift,Y \mid X,A}$ equals $F(\tilde{\lambda}_1\circ\theta_* + \phi_*'\circ\theta_*\cdot \tilde{\lambda}_2,\,\tilde{\lambda}_2)$ for some functions $\tilde{\lambda}_1$ and $\tilde{\lambda}_2$. 
    Then
    it must hold that the difference between $F( 0, \{\ell(\cdot,1)-\ell(\cdot,0)\}/\rho_* )$ and $F(\tilde{\lambda}_1\circ\theta_* + \phi_*'\circ\theta_*\cdot \tilde{\lambda}_2,\,\tilde{\lambda}_2)$, 
    which, since  $F$ is bilinear, equal $F(\tilde{\lambda}_1\circ\theta_* + \phi_*'\circ\theta_*\cdot \tilde{\lambda}_2,\,\tilde{\lambda}_2-\{\ell(\cdot,1)-\ell(\cdot,0)\}/\rho_* )$, lies in $\tangent_{\covshiftpostdrift,Y \mid X,A}^\perp$. By \eqref{eq: covshiftpostdrift tangent orthogonal}, we have that, for $P_*$-a.e. $x$,
    \begin{align}
        & \expect \left[ v(X;1) \left\{ \tilde{\lambda}_1\circ\theta_*(X) + \phi_*'\circ\theta_*(X)\cdot \tilde{\lambda}_2(X) \right\} \mid \theta_*(X) = \theta_*(x) \right] = 0, \label{eq: covshiftpostdrift EIF1} \\
        & \tilde{\lambda}_2(x) - \frac{\ell(x,1)-\ell(x,0)}{\rho_*} 
        = - \frac{v(x;1) \cdot\phi_*'\circ\theta_*(x) \left\{ \tilde{\lambda}_1\circ\theta_*(x) + \phi_*'\circ\theta_*(x)\cdot \tilde{\lambda}_2(x) \right\} }{v(x;0)}. \label{eq: covshiftpostdrift EIF2}
    \end{align}
    Equation~\ref{eq: covshiftpostdrift EIF1} implies that
    $$\tilde{\lambda}_1\circ\theta_*(x) = - \frac{\expect[v(X;1) \tilde{\lambda}_2(X) \mid \theta_*(X) = \theta_*(x)] \phi_*'\circ\theta_*(x)}{\expect[v(X;1) \mid \theta_*(X) = \theta_*(x)]}.$$
    Plug this into \eqref{eq: covshiftpostdrift EIF2} to find that
    \begin{align*}
        & v(x;0) \left\{ \frac{\ell(x,1) - \ell(x,0)}{\rho_*} - \tilde{\lambda}_2(x) \right\} \\
        &= - v(x;1) \frac{\expect[v(X;1) \tilde{\lambda}_2(X) \mid \theta_*(X) = \theta_*(x)]}{\expect[v(X;1) \mid \theta_*(X) = \theta_*(x)]} \left\{ \phi_*'\circ\theta_*(x) \right\}^2 + v(x;1) \left\{ \phi_*'\circ\theta_*(x) \right\}^2 \tilde{\lambda}_2(x).
    \end{align*}
    With notations defined in Theorem~\ref{thm: posterior drift EIF} and $\zeta := \tilde{\lambda}_2 \mu$, the above equation is equivalent to the linear integral equation
    $$\mathscr{A} \zeta + \kappa = \zeta.$$
    With $\zeta_*$ being the solution, we have that $\tilde{\lambda}_2=\zeta_*/\mu$ and $\Tilde{\lambda}_1 \circ \theta_*(x) = -\mathscr{A} \zeta_* \cdot \phi' \circ \theta_*/\mathscr{A} \mu$, and it follows that the projection of $F( 0, \{\ell(\cdot,1)-\ell(\cdot,0)\}/\rho_* )$ onto $\tangent_{\covshiftpostdrift,Y \mid X,A}$ is $F(g_1,g_2)$,
    with $g_1$ and $g_2$ defined in Theorem~\ref{thm: posterior drift EIF}. 
    The desired result follows.

    The proof for Theorem~\ref{thm: posterior drift EIF2}, corresponding to the other condition~\ref{DScondition: posterior drift}, is similar. Because of the additional condition that $X \independent A$, we have that the tangent space is $\tangent_X \oplus \tangent_A \oplus \tangent_{\covshiftpostdrift,Y \mid X, A}$. Therefore, the efficient influence function is the sum of (i) the projection of $D_\nonparametric(\rho_*,r_*)$ onto $\tangent_X$, namely $o \mapsto \mathcal{E}_*(x) - r_*$, (ii) the projection of $D_\nonparametric(\rho_*,r_*)$ onto $\tangent_A$, namely $o \mapsto 0$, and (iii) the projection of $D_\nonparametric(\rho_*,r_*)$ onto $\tangent_{\covshiftpostdrift, Y \mid X, A}$ found above.
\end{proof}

\subsection{Location-scale shift conditions} \label{sec: location-scale shift proof}

\begin{proof}[Proof of Theorems~\ref{thm: ls shift EIF} and \ref{thm: ls shift EIF2}]
    Similarly to the proof of Theorem~\ref{thm: posterior drift EIF}, we focus on the result Theorem~\ref{thm: ls shift EIF} for the weaker condition~\ref{DScondition: ls generalized target shift} and use the same proof strategy. Hence, we abbreviate the presentation.

    \noindent\textbf{Characterization of tangent space:} Condition~\ref{DScondition: ls generalized target shift} is equivalent to
    $$\intd P_{X \mid A,Y}^0(x \mid a=1,y) = |\det(W(y))| \intd P_{X \mid A,Y} (W(y) x + b(y) \mid a=0, y).$$
    Thus, for a distribution $P$ satisfying Condition~\ref{DScondition: ls generalized target shift}, its density $p$ can be written as
    \begin{align*}
        p(o) &= p_Y(y) \times p_{A \mid Y}(a \mid y) \\
        &\quad\times \left\{ |\det(W_P(y))| \tau_P(W_P(y) x + b_P(y) \mid y) \right\}^a \{ \tau_P(x \mid y) \}^{1-a},
    \end{align*}
    where $p_Y$ is the density of $Y$ with respect to some dominating measure, $p_{A \mid Y}(a \mid y) = \Prob_P(A=a \mid Y=y)$; $W_P$, $b_P$ and $\tau_P$ are the counterparts (for $P$) of $W_*$, $b_*$ and $\tau_*$ (for $P_*$).  Consider a regular path $\{P^\epsilon: \epsilon \text{ in a neighborhood of } 0\}$ that passes through $P_*$ at $\epsilon=0$. We use a superscript or subscript $\epsilon$ to denote components of $P^\epsilon$. With the notations defined in Theorem~\ref{thm: ls shift EIF}, the score function at $\epsilon=0$ is
    \begin{align*}
        \left. \frac{\partial \log p^\epsilon(o)}{\partial \epsilon} \right|_{\epsilon=0} &= \left. \frac{\partial \log p_Y^\epsilon(y)}{\partial \epsilon} \right|_{\epsilon=0} + \left. \frac{\partial \log p_{A \mid Y}^\epsilon(a \mid y)}{\partial \epsilon} \right|_{\epsilon=0} \\
        &\quad+ (1-a) \left. \frac{\partial \log \tau_\epsilon(x \mid y)}{\partial \epsilon} \right|_{\epsilon=0} + a \tr \left( \left. \frac{\partial W_\epsilon(y)}{\partial \epsilon} \right|_{\epsilon=0} W_*(y)^{-1} \right) \\
        &\quad+ a \left. \frac{\partial \log \tau_\epsilon(\psi_y(x) \mid y)}{\partial \epsilon} \right|_{\epsilon=0} \\
        &\quad+ a \left\{ \left. \frac{\partial \log \tau_*(\sharp \mid y)}{\partial \sharp} \right|_{\sharp=\psi_y(x)} \right\}^\top \left\{ \left. \frac{\partial W_\epsilon(y)}{\partial \epsilon} \right|_{\epsilon=0} x + \left. \frac{\partial b_\epsilon(y)}{\partial \epsilon} \right|_{\epsilon=0} \right\},
    \end{align*}
    where we have used Jacobi's formula to calculate the derivative $\partial \log\det(W_\epsilon(y))/\partial \epsilon$.
    In the above display, the first two terms are arbitrary elements of $\tangent_Y$ and $\tangent_{A \mid Y}$, respectively; $\left. \frac{\partial \log \tau_\epsilon(x \mid y)}{\partial \epsilon} \right|_{\epsilon=0}$ is an arbitrary element of $\tangent_{X \mid Y,A=0}$. Therefore, with the notations defined in Theorem~\ref{thm: ls shift EIF}, and with 
    $$
    H_1(x \mid y; \lambda_1,\lambda_2) = 
    \tr(\lambda_1(y) W_*(y)^{-1}) + H_{*}(\psi_y(x) \mid y)+
    \left\{ (\log \tau_*)'(\psi_y(x) \mid y) \right\}^\top \left[ \lambda_1(y) x + \lambda_2(y) \right],
    $$
    the tangent space $\tangent_{\LSgentargetshift} \subseteq \tangent_{Y} \oplus \tangent_{A \mid Y} \oplus \tangent_{\LSgentargetshift, X \mid Y,A}$
    where $\tangent_{\LSgentargetshift, X \mid Y,A}$ is the $L^2_0(P_*)$-closure of
    \begin{align*}
        \funclass &:= \Bigg\{ o \mapsto (1-a) H_{*}(x \mid y) + a H_1(x \mid y; \lambda_1,\lambda_2):\,  \expect[H_{*}(X \mid Y) \mid A=0,Y]=0, \\
        &\qquad\qquad \lambda_1(y) \in \real^{d_X \times d_X},\lambda_2(y) \in \real^{d_X} \text{ such that the above function lies in } L^2(P_*) \Bigg\}.
    \end{align*}
    
    Let $s \in \tangent_{Y} \oplus \tangent_{A \mid Y} \oplus \funclass$ be bounded. Then we can decompose $s$ as $s=s_Y + s_{A \mid Y} + s_{X \mid Y,A}$ such that $s_Y \in \tangent_{Y}$, $s_{A \mid Y} \in \tangent_{A \mid Y}$ and
    \begin{align*}
        s_{X \mid Y,A}: o &\mapsto (1-a) H_{*}(x \mid y) + a  H_{1}(x \mid y; \lambda_1,\lambda_2)\in \funclass
    \end{align*}
    for some functions $H_*$, $\lambda_1$ and $\lambda_2$.
    For all $\epsilon$ sufficiently close to zero, let the density of $P^\epsilon$ be
    \begin{align*}
        o &\mapsto P_{*,Y(y)} \{ 1 + \epsilon s_Y(y) \} \times P_{*,A \mid Y}(a \mid y) \{ 1 + \epsilon s_{A \mid Y}(a \mid y)\} \\
        &\quad\times \big\{ |\det(W_*(y) + \epsilon \lambda_1(y))| \tau_*((W_*(y) + \epsilon \lambda_1(y)) x + (b_*(y) + \epsilon \lambda_2(y)) \mid y) \\
        &\qquad\qquad\times (1+\epsilon H_*((W_*(y) + \epsilon \lambda_1(y)) x + (b_*(y) + \epsilon \lambda_2(y)) \mid y)) \big\}^a \\
        &\quad\times \{ \tau_*(x \mid y) (1 + \epsilon H_*(x \mid y)) \}^{1-a}.
    \end{align*}
    It is straightforward to verify that the score function of $P^\epsilon$ at $\epsilon=0$ equals $s$ and $\{P^\epsilon\}$ is a regular path passing through $P_*$ at $\epsilon=0$. Therefore, we have shown that $\tangent_{\LSgentargetshift} = \tangent_{Y} \oplus \tangent_{A \mid Y} \oplus \tangent_{\LSgentargetshift, X \mid Y,A}$.

    \noindent\textbf{Characterization of orthogonal complement of tangent space:} The orthogonal complement $\tangent_{\LSgentargetshift}^\perp$ is a subspace of $\tangent_{X \mid A,Y}$. Since $A$ is binary, we further have that
    $$\tangent_{X \mid A,Y} = \left\{ o \mapsto (1-a) h_0(x \mid y) + a h_1(x \mid y): h_0 \in L^2_0(P_{*,X \mid Y,A=0}), h_1 \in L^2_0(P_{*,X \mid Y,A=1}) \right\}.$$
    Therefore, recalling the notations $g_{*,Y}$ and $\tau_*$ from the theorem statement,
    $\tangent_{\LSgentargetshift}^\perp$ equals
    \begin{align*}
        &\left\{ f: o \mapsto (1-a) h_0(x \mid y) + a h_1(x \mid y): f \in \tangent_{X \mid A,Y}, \expect[f(O) g(O)]=0 \text{ for all $g \in \funclass$} \right\} \\
        &= \Big\{ o \mapsto (1-a) h_0(x \mid y) + a h_1(x \mid y): \\
        &\qquad \expect[h_0(X \mid Y) \mid A=0,Y] = \expect[h_1(X \mid Y) \mid A=1,Y] = 0, \\
        &\qquad\iint g_{*,Y}(y) \tau_*(x \mid y) H_{*}(x \mid y) h_0(x \mid y) \intd x \  P_Y^0(\intd y) \\
        &\quad\qquad+ \iint (1-g_{*,Y}(y)) \cdot |\det(W_*(y))| \cdot \tau_*(\psi_y(x) \mid y) \cdot  h_1(x \mid y) \\
        &\qquad\qquad\times 
        H_1(x \mid y; \lambda_1,\lambda_2) \intd x \ P_Y^0(\intd y) =0 \text{ for any } H_{*}, \lambda_1,\lambda_2 \Big\}.
    \end{align*}
    By the change of variables $z = \psi_y(x)$, the second integral in the condition above can be written as 
    \begin{align*}
        &\iint (1-g_{*,Y}(y)) \tau_*(z \mid y) h_1 \left( \psi_y^{-1}(z) \mid y \right) \\
        &\times \left\{ \tr(\lambda_1(y) W_*(y)^{-1}) + H_{*}(z \mid y) + \{(\log \tau_*)'(z \mid y)\}^\top \left[ \lambda_1(y) \psi_y^{-1}(z) + \lambda_2(y) \right] \right\}  \intd z \ P_Y^0(\intd y).
    \end{align*}
    Changing the notation from $z$ to $x$ shows that this equals
    \begin{align*}
        &\iint (1-g_{*,Y}(y)) \tau_*(x \mid y) h_1 \left( \psi_y^{-1}(x) \mid y \right) \\
        &\times \left\{ \tr(\lambda_1(y) W_*(y)^{-1}) + H_{*}(x \mid y) + \{(\log \tau_*)'(x \mid y)\}^\top \left[ \lambda_1(y) \psi^{-1}_y(x) + \lambda_2(y) \right] \right\} \intd x \ P_Y^0(\intd y).
    \end{align*}
    Now note that since $\psi_y^{-1}(x) =  W_*(y)^{-1}(x-b_*(y)) $, we have, 
    with 
    $G := \tr(\lambda_1(y) W_*(y)^{-1}) + \{(\log \tau_*)'(x \mid y)\}^\top \lambda_1(y) \psi_y^{-1}(x),$
    that
    \begin{align*}
        & h_1 \left( \psi_y^{-1}(x) \mid y \right) G \\
        &= h_1 \left( \psi_y^{-1}(x) \mid y \right) \left\{ \tr(\lambda_1(y) W_*(y)^{-1}) + \{(\log \tau_*)'(x \mid y)\}^\top \lambda_1(y) W_*(y)^{-1}(x-b_*(y)) \right\}. 
       \end{align*}
    Further, since
    $$\{(\log \tau_*)'(x \mid y)\}^\top \lambda_1(y) W_*(y)^{-1}(x-b_*(y))  = \tr \left( \lambda_1(y) W_*(y)^{-1}(x-b_*(y)) \right. \left.\{(\log \tau_*)'(x \mid y)\}^\top \right),$$
    this also equals
    \begin{align*}
        & h_1 \left( \psi_y^{-1}(x) \mid y \right) \left\{ \tr(\lambda_1(y) W_*(y)^{-1}) + \tr \left( \lambda_1(y) W_*(y)^{-1}(x-b_*(y)) \{(\log \tau_*)'(x \mid y)\}^\top \right) \right\} \\
        &= \tr \left( \lambda_1(y) W_*(y)^{-1} \cdot h_1 \left( \psi_y^{-1}(x) \mid y \right)  \cdot \Lambda(x,y) \right).
    \end{align*}
    Next, since expectation and trace are commutative linear operators, we have
    that
    \begin{align*}
        & \iint (1-g_{*,Y}(y)) \tau_*(x \mid y) h_1 \left( \psi_y^{-1}(x) \mid y \right) G \intd x P_Y^0(\intd y) \\
        &= \expect \left[ \tr \left( \lambda_1(Y) W_*(Y)^{-1} 
        \expect \left[ h_1 \left( \psi_Y^{-1}(X) \mid Y \right)  \Lambda(X,Y) \mid A=0,Y \right] \right) \right].
    \end{align*}
    Since $H_*$, $\lambda_1$ and $\lambda_2$ are arbitrary and $W_*(y)$ is invertible, we can transform the constraints in $\tangent_{\LSgentargetshift}^\perp$ to equality constraints on $h_0$ and $h_1$ and obtain that
    \begin{align*}
        \tangent_{\LSgentargetshift}^\perp &= \Big\{ o \mapsto (1-a) h_0(x \mid y) + a h_1(x \mid y): \\
        &\qquad g_{*,Y}(y) h_0(x \mid y) + (1-g_{*,Y}(y)) h_1 \left( \psi_y^{-1}(x) \mid y \right) = 0, \\
        &\qquad \expect \left[ h_1 \left( \psi_Y^{-1}(X) \mid Y \right) \Lambda(X,Y) \mid A=0,Y=y \right] = 0_{d \times d}, \\
        &\qquad \expect \left[ h_1 \left( \psi_Y^{-1}(X) \mid Y \right) (\log \tau_*)'(X \mid Y) \mid A=0,Y=y \right] = 0_{d}, \\
        &\qquad \expect[h_0(X \mid Y) \mid A=0,Y=y] = \expect[h_1(X \mid Y) \mid A=1,Y=y] = 0 \Big\},
    \end{align*}
    where we display the dimension of each multivariate equation. 
    By the change of variable $z=\psi_y^{-1}(x)$, the first constraint is equivalent to
    $$h_1(x \mid y) = -\frac{g_{*,Y}(y)}{1-g_{*,Y}(y)} h_0(\psi_y(x) \mid y).$$
    Moreover, the constraint $\expect[h_0(X \mid Y) \mid A=0,Y]=0$ implies $\expect[h_1(X \mid Y) \mid A=1,Y] = 0$, and so it is enough to consider the first constraint. 
    With this, the orthogonal complement becomes
        \begin{align*}
        \tangent_{\LSgentargetshift}^\perp &= \Big\{ o \mapsto (1-a) h_0(x \mid y) - a \frac{g_{*,Y}(y)}{1-g_{*,Y}(y)} h_0(\psi_y(x) \mid y): \\
        &\qquad \expect \left[ h_0(X \mid Y) \mid A=0,Y \right] = 0 ,\qquad \expect \left[ h_0(X \mid Y) \Lambda(X,Y) \mid A=0,Y \right] = 0_{d \times d}, \\
        &\qquad \expect \left[ h_0(X \mid Y) (\log \tau_*)'(X \mid Y) \mid A=0,Y \right] = 0_{d}\Big\}.
    \end{align*}

    \noindent\textbf{Characterization of the efficient influence function:} By Lemma~\ref{lemma: orthogonality},
    $$D_\nonparametric(\rho_*,r_*)(o) = \underbrace{\frac{1-a}{\rho_*} \left\{ \ell(x,y) - \mathcal{E}_{*,Y}(y) \right\} }_{\text{projection of $D_\nonparametric(\rho_*,r_*)$ onto $\tangent_{X \mid Y,A}$}} + \underbrace{ \frac{1-a}{\rho_*} \left\{ \mathcal{E}_{*,Y}(y) - r_* \right\}}_{\text{projection of $D_\nonparametric(\rho_*,r_*)$ onto $\tangent_{A,Y}$}}.$$
    The projection of $D_\nonparametric(\rho_*,r_*)$ onto $\tangent_\LSgentargetshift$ is the sum of (i) $o \mapsto \frac{1-a}{\rho_*} \left\{ \mathcal{E}_{*,Y}(y) - r_* \right\}$ and (ii) the projection of $o \mapsto \frac{1-a}{\rho_*} \left\{ \ell(x,y) - \mathcal{E}_{*,Y}(y) \right\}$ onto $\tangent_{\LSgentargetshift,X \mid Y,A}$. We next find the latter.

    Suppose that,   for some functions $H_*$, $\lambda_1$ and $\lambda_2$, the latter projection is
    \begin{align*}
        o &\mapsto (1-a) H_{*}(x \mid y) + a H_1(x \mid y; \lambda_1,\lambda_2).
    \end{align*}
  With the notation $\Phi$ defined in Theorem~\ref{thm: ls shift EIF}, we have that the difference between $o \mapsto \frac{1-a}{\rho_*}$ $\times \left\{ \ell(x,y) - \mathcal{E}_{*,Y}(y) \right\}$ and the above function is
    \begin{align*}
        & (1-a) \left\{ \Phi(x,y) - H_{*}(x \mid y) \right\} 
        - a H_1(x \mid y; \lambda_1,\lambda_2),
    \end{align*}
     which lies in $\tangent_{\LSgentargetshift}^\perp$. Therefore,
     \begin{align}
        &\expect \left[ \left\{ \Phi(X,Y) - H_{*}(X \mid Y) \right\} \Lambda(X,Y) \mid A=0,Y \right] = 0, \label{eq: LS shift EIF1} \\
        &\expect \left[ \left\{ \Phi(X,Y) - H_{*}(X \mid Y) \right\} (\log \tau_*)'(X \mid Y) \mid A=0,Y \right] = 0, \label{eq: LS shift EIF2} \\
        \begin{split}
        & H_1(x \mid y; \lambda_1,\lambda_2) = \frac{g_{*,Y}(y)}{1-g_{*,Y}(y)} \left\{ \Phi(x,y) - H_{*}(\psi_y(x) \mid y) \right\}.
        \end{split} \label{eq: LS shift EIF3}
    \end{align}
    Rearranging \eqref{eq: LS shift EIF3}, by the change of variables $z=\psi_y^{-1}(x)$ and then changing the notation from $z$ to $x$, we have
    \begin{align}
    \begin{split}
        & H_{*}(x \mid y) 
        = g_{*,Y}(y) \Phi(x,y) - (1-g_{*,Y}(y)) \\
        &\times \left\{ \tr(\lambda_1(y) W_*(y)^{-1}) + \{(\log \tau_*)'(x \mid y)\}^\top \left[ \lambda_1(y) W_*(y)^{-1} (x-b_*(y)) + \lambda_2(y) \right] \right\} \\
        &= g_{*,Y}(y) \Phi(x,y) - (1-g_{*,Y}(y)) \cdot \big\{ \tr \left( \lambda_1(y) W_*(y)^{-1} \Lambda(x,y) \right) + \{ (\log \tau_*)'(x \mid y) \}^\top \lambda_2(y) \big\}.
    \end{split} \label{eq: LS shift EIF3 solved}
    \end{align}
    Plug $H_*$ in \eqref{eq: LS shift EIF3 solved} into \eqref{eq: LS shift EIF1} and \eqref{eq: LS shift EIF2} to obtain that, with
    $$
    N = \Phi(X,Y) + \tr \left( \lambda_1(Y) W_*(Y)^{-1} \Lambda(X,Y) \right) + \{ (\log \tau_*)'(X \mid Y) \}^\top \lambda_2(Y),
    $$
    we have
    \begin{align*}
        & \expect [ N \cdot \Lambda(X,Y) \mid A=0,Y ] = 0, \qquad
         \expect [  N \cdot (\log \tau_*)'(X \mid Y) \mid A=0,Y ] = 0.
    \end{align*}
    Although the dimensions of the above two equations are $d \times d$ and $d$, respectively, only one-dimensional summaries of the unknowns $\lambda_1$ and $\lambda_2$ are involved. 
    Thus, we search for solutions of the form $\lambda_1(y) = \xi_1(y) I_{d \times d}$ and $\lambda_2(y) = \xi_2(y) 1_d$ for real-valued functions $\xi_1$ and $\xi_2$. With the notations defined in Theorem~\ref{thm: ls shift EIF}, the above two equations imply the following linear system:
    $$a(y) + M(y) \begin{pmatrix}
    \xi_1(y) \\
    \xi_2(y)
    \end{pmatrix} = 0,$$
    whose solution is
     $$\begin{pmatrix}
    \xi_1(y) \\
    \xi_2(y)
    \end{pmatrix} = - M(y)^{-1} a(y).$$
    Using this solution in \eqref{eq: LS shift EIF3 solved}, we obtain that
    $$H_{*}(x \mid y) = g_{*,Y}(y) \Phi(x,y) - (1-g_{*,Y}(y)) c(x,y)^\top M(y)^{-1} a(y).$$
    Then the projection of $o \mapsto \frac{1-a}{\rho_*} \left\{ \ell(x,y) - \mathcal{E}_{*,Y}(y) \right\}$ onto $\tangent_{\LSgentargetshift,X \mid Y,A}$ is as stated in Theorem~\ref{thm: ls shift EIF}.
    The desired result follows.

    For the other condition~\ref{DScondition: ls conditional shift}, it is additionally assumed that $Y \independent A$, and so the tangent space is $\tangent_{Y} \oplus \tangent_{A} \oplus \tangent_{\LSgentargetshift, X \mid Y,A}$. Theorem~\ref{thm: ls shift EIF2} follows by summing up the projections of $D_\nonparametric(\rho_*,r_*)$ onto these three subspaces.
\end{proof}

\subsection{Invariant density ratio shape}

\begin{proof}[Proof of Theorem~\ref{thm: invariant density ratio shape}]
    \noindent\textbf{Characterization of tangent space:} Since $\const_* = (\int \iota_* \intd \mu)^{-1}$ is the normalizing constant such that $\const_* \iota_*$ is the density of $P_{*,X \mid Y=0,A=0}$, for any distribution $P$ satisfying Condition~\ref{DScondition: invariant density ratio shape}, its density $p$ can be written as
    \begin{align*}
        p(o) &= w_P(a,y) \times \eta_P(x \mid 1,1)^{ay} \eta_P(x \mid 0,1)^{a(1-y)} \\
        &\quad\times \eta_P(x \mid 1,0)^{(1-a)y} \left\{ \const_P \iota_P(x) \right\}^{(1-a)(1-y)},
    \end{align*}
    where $w_P$, $\eta_P$, $\iota_P$ and $\const_P$ are the counterparts (for $P$) of $w_*$, $\eta_*$, $\iota_*$ and $\const_*$ (for $P_*$). We still consider a regular path $\{P^\epsilon: \epsilon \text{ in a neighborhood of } 0\}$ that passes through $P_*$ at $\epsilon=0$, and use a superscript or subscript $\epsilon$ to denote components of $P^\epsilon$.
    With the notations in Theorem~\ref{thm: invariant density ratio shape}, the score function at $\epsilon=0$ equals
    \begin{align*}
        &\left. \frac{\partial \log p^\epsilon(o)}{\partial \epsilon} \right|_{\epsilon=0} \\
        &= \left. \frac{\partial \log w_\epsilon(y,a)}{\partial \epsilon} \right|_{\epsilon=0} + ay \left. \frac{\partial \log \eta_\epsilon(x \mid 1,1)}{\partial \epsilon} \right|_{\epsilon=0} + a(1-y) \left. \frac{\partial \log \eta_\epsilon(x \mid 0,1)}{\partial \epsilon} \right|_{\epsilon=0} \\
        &\quad+ (1-a) y \left. \frac{\partial \log \eta_\epsilon(x \mid 1,0)}{\partial \epsilon} \right|_{\epsilon=0} \\
        &\quad+ (1-a)(1-y) \Bigg\{ \left. \frac{\partial \{ \log \eta_\epsilon(x \mid 1,0) + \log \eta_\epsilon(x \mid 0,1) - \log \eta_\epsilon(x \mid 1,1) \}}{\partial \epsilon} \right|_{\epsilon=0} + \left. \frac{\intd \log \const_\epsilon}{\intd \epsilon} \right|_{\epsilon=0} \Bigg\}
    \end{align*}
    where
    \begin{align*}
        \left. \frac{\intd \log \const_\epsilon}{\intd \epsilon} \right|_{\epsilon=0} &= - \const_* \int \Bigg\{ \left. \frac{\partial \eta_\epsilon(\cdot \mid 1,0)}{\partial \epsilon} \right|_{\epsilon=0} \frac{\eta_*(\cdot \mid 0,1)}{\eta_*(\cdot \mid 1,1)} + \left. \frac{\partial \eta_\epsilon(\cdot \mid 0,1)}{\partial \epsilon} \right|_{\epsilon=0} \frac{\eta_*(\cdot \mid 1,0)}{\eta_*(\cdot \mid 1,1)} \\
        &\qquad\qquad\qquad- \left. \frac{\partial \eta_\epsilon(\cdot \mid 1,1)}{\partial \epsilon} \right|_{\epsilon=0} \frac{\eta_*(\cdot \mid 0,1) \eta_*(\cdot \mid 1,0)}{\eta_*(\cdot \mid 1,1)^2} \Bigg\} \intd \mu \\
        &= - \const_* \int \iota_*(\cdot) \left. \frac{\partial \{ \log \eta_\epsilon(\cdot \mid 1,0) + \log \eta_\epsilon(\cdot \mid 0,1) - \log \eta_\epsilon(\cdot \mid 1,1) \}}{\partial \epsilon} \right|_{\epsilon=0} \intd \mu
    \end{align*}
    Since the following score functions are arbitrary under Condition~\ref{DScondition: invariant density ratio shape}:
    \begin{align*}
        & (y,a) \mapsto \left. \frac{\partial \log w_\epsilon(y,a)}{\partial \epsilon} \right|_{\epsilon=0}, \qquad x \mapsto \left. \frac{\partial \log \eta_\epsilon(x \mid y,a)}{\partial \epsilon} \right|_{\epsilon=0} \quad \eqfor (y,a) \in \{(1,1),(0,1),(1,0)\},
    \end{align*}
    the tangent space $\tangent_{\IDRS} = \tangent_{A,Y} \oplus \tangent_{\IDRS, X \mid Y,A} = \tangent_{Y} \oplus \tangent_{A \mid Y} \oplus \tangent_{\IDRS, X \mid Y,A}$ where $\tangent_{\IDRS, X \mid Y,A}$ is the $L^2_0(P_*)$-closure of
    \begin{align*}
        \mathcal{H} &:= \Bigg\{ o \mapsto ay H_{1,1}(x) + a (1-y) H_{0,1}(x) + (1-a) y H_{1,0}(x) \\
        &\qquad\qquad+ (1-a)(1-y) \left[ H_{1,0}(x) + H_{0,1}(x) - H_{1,1}(x) - \const_* \int \iota_* (H_{1,0} + H_{0,1} - H_{1,1}) \intd \mu \right]:\\
        &\qquad\quad H_{y,a} \in L^2_0(P_{*,X \mid Y=y,A=a}) \eqfor (y,a) \in \{(1,1),(0,1),(1,0)\} \Bigg\}.
    \end{align*}

    \noindent\textbf{Characterization of orthogonal complement of tangent space:} Since
    \begin{align*}
        \tangent_\IDRS^\perp \subset \tangent_{X \mid Y,A} &= \Big\{f: o \mapsto ay f_{1,1}(x) + a(1-y) f_{0,1}(x) + (1-a)y f_{1,0}(x) + (1-a) (1-y) f_{0,0}(x): \\
        &\qquad\qquad f_{y,a} \in L^2_0(P_{*,X \mid Y=y,A=a}) \eqfor y,a \in \{0,1\} \Big\},
    \end{align*}
    we have that $\tangent_\IDRS^\perp$ equals
    \begin{align*}
        & \Bigg\{ f: o \mapsto ay f_{1,1}(x) + a(1-y) f_{0,1}(x) + (1-a)y f_{1,0}(x) + (1-a) (1-y) f_{0,0}(x): \\
        &\quad f_{y,a} \in L^2_0(P_{*,X \mid Y=y,A=a}) \eqfor y,a \in \{0,1\}, \expect_{P_*}[f(O) g(O)] = 0 \text{ for all } g \in \mathcal{H} \Bigg\} \\
        &= \Bigg\{ f: o \mapsto ay f_{1,1}(x) + a(1-y) f_{0,1}(x) + (1-a)y f_{1,0}(x) + (1-a) (1-y) f_{0,0}(x): \\
        &\qquad \int \left[ w_*(1,1) \eta_*(\cdot \mid 1,1) f_{1,1}(\cdot) - w_*(0,0) \const_* \iota_*(\cdot) f_{0,0}(\cdot) \right] H_{1,1}(\cdot) \intd \mu \\
        &\qquad\quad+ \int \left[ w_*(0,1) \eta_*(\cdot \mid 0,1) f_{0,1}(\cdot) + w_*(0,0) \const_* \iota_*(\cdot) f_{0,0}(\cdot) \right] H_{0,1}(\cdot) \intd \mu \\
        &\qquad\quad+ \int \left[ w_*(1,0) \eta_*(\cdot \mid 1,0) f_{1,0}(\cdot) + w_*(0,0) \const_* \iota_*(\cdot) f_{0,0}(\cdot) \right] H_{1,0}(\cdot) \intd \mu = 0 \\
        &\qquad\quad \text{for any } H_{y,a} \in L^2_0(P_{*,X \mid Y=y,A=a}) \where (y,a) \in \{(1,1),(0,1),(1,0)\} \\
        &\qquad f_{y,a} \in L^2_0(P_{*,X \mid Y=y,A=a}) \eqfor y,a \in \{0,1\} \Bigg\},
    \end{align*}
    where we have used the fact that $\const_* \int \iota_* (H_{1,0} + H_{0,1} - H_{1,1}) \intd \mu$ is a scalar and $\expect_{P_*}[f_{0,0}(X) \mid Y=0,A=0]=0$.
    Since each $H_{y,a}$ is arbitrary, we have that $\tangent_\IDRS^\perp$ equals
    \begin{align*}
        &\Big\{ f: o \mapsto ay f_{1,1}(x) + a(1-y) f_{0,1}(x) + (1-a)y f_{1,0}(x) + (1-a) (1-y) f_{0,0}(x): \\
        &\qquad  w_*(1,1) \eta_*(x \mid 1,1) f_{1,1}(x) - w_*(0,0) \const_* \iota_*(x) f_{0,0}(x) = 0, \\
        &\qquad w_*(0,1) \eta_*(x \mid 0,1) f_{0,1}(x) + w_*(0,0) \const_* \iota_*(x) f_{0,0}(x) = 0, \\
        &\qquad w_*(1,0) \eta_*(x \mid 1,0) f_{1,0}(x) + w_*(0,0) \const_* \iota_*(x) f_{0,0}(x) = 0, \\
        &\qquad f_{y,a} \in L^2_0(P_{*,X \mid Y=y,A=a}) \eqfor y,a \in \{0,1\} \Big\} \\
        &= \Bigg\{ f: o \mapsto f_{0,0}(x) w_*(0,0) \const_* \iota_*(x) \times \Bigg[ \frac{ay}{w_*(1,1) \eta_*(x \mid 1,1)} - \frac{a(1-y)}{w_*(0,1) \eta_*(x \mid 0,1)} \\
        &\qquad\qquad\quad- \frac{(1-a) y}{w_*(1,0) \eta_*(x \mid 1,0)} + \frac{(1-a)(1-y)}{w_*(0,0) \const_* \iota_*(x)} \Bigg]: f_{0,0} \in L^2_0(P_{*, X \mid Y=0,A=0}) \Bigg\}.
    \end{align*}

    \noindent\textbf{Characterization of the efficient influence function:} By Lemma~\ref{lemma: orthogonality}, using a similar argument as in the proof of Theorems~\ref{thm: ls shift EIF} and \ref{thm: ls shift EIF2}, the efficient ficient influence function must be the sum of (i) $o \mapsto \frac{1-a}{\rho_*} \left\{ \mathcal{E}_{*,Y}(y) - r_* \right\}$ and (ii) the projection of $o \mapsto \frac{1-a}{\rho_*} \left\{ \ell(x,y) - \mathcal{E}_{*,Y}(y) \right\}$ onto $\tangent_{\IDRS,X \mid Y,A}$. Suppose that this projection is
    \begin{align}
    \begin{split}
        o &\mapsto ay H_{1,1}(x) + a (1-y) H_{0,1}(x) + (1-a) y H_{1,0}(x) \\
        &\quad+ (1-a)(1-y) \left[ H_{1,0}(x) + H_{0,1}(x) - H_{1,1}(x) - \const_* \int \iota_* (H_{1,0} + H_{0,1} - H_{1,1}) \intd \mu \right]
    \end{split} \label{eq: invariant density ratio shape assumed projection}
    \end{align}
    for some functions $H_{1,1}$, $H_{0,1}$ and $H_{1,0}$. Then, the difference between $o \mapsto \frac{1-a}{\rho_*} \left\{ \ell(x,y) - \mathcal{E}_{*,Y}(y) \right\}$ and this function is
    \begin{align*}
        o &\mapsto -ay H_{1,1}(x) - a(1-y) H_{0,1}(x) + (1-a) y \left\{ \frac{\ell(x,1) - \mathcal{E}_{*,Y}(1)}{\rho_*} - H_{1,0}(x) \right\} \\
        &+ (1-a)(1-y) \Bigg\{ \frac{\ell(x,0) - \mathcal{E}_{*,Y}(0)}{\rho_*} - H_{1,0}(x) - H_{0,1}(x) + H_{1,1}(x) \\
        &\qquad\qquad\qquad\qquad\qquad+ \const_* \int \iota_* (H_{1,0} + H_{0,1} - H_{1,1}) \intd \mu \Bigg\},
    \end{align*}
    which must lie in $\tangent_{\IDRS,X \mid Y,A}^\perp$, namely equals
    \begin{align*}
        o &\mapsto f_{0,0}(x) w_*(0,0) \const_* \iota_*(x) \times \Bigg\{ \frac{ay}{w_*(1,1) \eta_*(x \mid 1,1)} - \frac{a(1-y)}{w_*(0,1) \eta_*(x \mid 0,1)} \\
        &\qquad- \frac{(1-a) y}{w_*(1,0) \eta_*(x \mid 1,0)} + \frac{(1-a)(1-y)}{w_*(0,0) \const_* \iota_*(x)} \Bigg\}
    \end{align*}
    for some function $f_{0,0} \in L^2_0(P_{*,X \mid Y=0,A=0})$. This implies that
    \begin{align*}
        & H_{1,1}(x) = -f_{0,0}(x) \frac{w_*(0,0) \const_* \iota_*(x)}{w_*(1,1) \eta_*(x \mid 1,1)}, \quad
         H_{0,1}(x) = f_{0,0}(x) \frac{w_*(0,0) \const_* \iota_*(x)}{w_*(0,1) \eta_*(x \mid 0,1)}, \\
        & \frac{\ell(x,1) - \mathcal{E}_{*,Y}(1)}{\rho_*} - H_{1,0}(x) = -f_{0,0}(x) \frac{w_*(0,0) \const_* \iota_*(x)}{w_*(1,0) \eta_*(x \mid 1,0)}, \\
        & \frac{\ell(x,0) - \mathcal{E}_{*,Y}(0)}{\rho_*} - H_{1,0}(x) - H_{0,1}(x) + H_{1,1}(x) + \const_* \int \iota_* (H_{1,0} + H_{0,1} - H_{1,1}) \intd \mu = f_{0,0}(x).
    \end{align*}
    Plugging $H_{1,1}$, $H_{0,1}$ and $H_{1,0}$ from the first three equations into the last equation, with notations from Theorem~\ref{thm: invariant density ratio shape} , we obtain the following linear integral equation for unknown $f_{0,0}$:
    \begin{align}
        & \mathcal{J} w_*(0,0) \const_* \iota_* f_{0,0} \nonumber \\
        &= \mathcal{I} + \expect_{P_*} \left[ \left( \mathcal{J}(X) - \frac{1}{w_*(0,0) \const_* \iota_*(X)} \right) w_*(0,0) \const_* \iota_*(X) f_{0,0}(X) \mid Y=0,A=0 \right]. \label{eq: invariant density ratio shape integral eq}
    \end{align}
    Multiplying both sides by the function $\{ \mathcal{J} - 1/(w_*(0,0) \const_* \iota_*) \}/\mathcal{J}$ and integrating both sides with respect to the measure $P_{*,X \mid Y=0,A=0}$, we have that
    \begin{align*}
        & \expect_{P_*} \left[ \left( \mathcal{J}(X) - \frac{1}{w_*(0,0) \const_* \iota_*(X)} \right) w_*(0,0) \const_* \iota_*(X) f_{0,0}(X) \mid Y=0,A=0 \right] \\
        &= \expect_{P_*} \left[ \frac{\left( \mathcal{J}(X) - \frac{1}{w_*(0,0) \const_* \iota_*(X)} \right) \mathcal{I}(X) }{\mathcal{J}(X)} \mid Y=0,A=0 \right] \\
        &\quad+ \expect_{P_*} \left[ \left( \mathcal{J}(X) - \frac{1}{w_*(0,0) \const_* \iota_*(X)} \right) w_*(0,0) \const_* \iota_*(X) f_{0,0}(X) \mid Y=0,A=0 \right] \\
        &\qquad\times \expect_{P_*} \left[ \frac{\left( \mathcal{J}(X) - \frac{1}{w_*(0,0) \const_* \iota_*(X)} \right)}{\mathcal{J}(X)} \mid Y=0,A=0 \right].
    \end{align*}
    Thus,
    $$\expect_{P_*} \left[ \left( \mathcal{J}(X) - \frac{1}{w_*(0,0) \const_* \iota_*(X)} \right) w_*(0,0) \const_* \iota_*(X) f_{0,0}(X) \mid Y=0,A=0 \right] = \mathcal{D}.$$
    Plug this into \eqref{eq: invariant density ratio shape integral eq} to obtain that the solution in $f_{0,0}$ is
    $$f_{0,0} = \frac{\mathcal{I} + \mathcal{D}}{\mathcal{J} w_*(0,0) \const_* \iota_*}.$$
    Plug this solution into \eqref{eq: invariant density ratio shape assumed projection} to obtain the projection and thus the efficient influence function claimed in Theorem~\ref{thm: invariant density ratio shape}.
\end{proof}

\bibliographystylesupp{abbrvnat}
\bibliographysupp{ref}

\end{document}